\documentclass[11pt]{article}
\usepackage{setspace}
\usepackage{cprotect}
\usepackage{fullpage}
\usepackage[shortlabels]{enumitem}
\usepackage{cite}
\usepackage{amsmath,amssymb,amsfonts}
\usepackage{graphicx}
\usepackage{textcomp}
\usepackage{xcolor}
\def\BibTeX{{\rm B\kern-.05em{\sc i\kern-.025em b}\kern-.08em
    T\kern-.1667em\lower.7ex\hbox{E}\kern-.125emX}}
    
\usepackage{cprotect}
\usepackage{float,amsthm, mathtools, bbm, accents}
\usepackage{enumitem}
\usepackage{soul,changepage}
\usepackage{subcaption} 
\usepackage[english]{babel}

\usepackage{url, hyperref}

\usepackage{algpseudocode}
\usepackage{algorithm}

\newtheorem{claim}{Claim}[section] 
\newtheorem{lem}{Lemma}[section]
\newtheorem{thm}{Theorem}

\newtheorem{definition}[claim]{Definition}

\newcommand{\mc}{\mathcal}
\newcommand{\beq}{\begin{equation}}
\newcommand{\eeq}{\end{equation}}

\DeclareMathOperator*{\argmin}{arg\,min}
\newcommand{\abs}[1]{\left\lvert #1 \right\rvert}

\newcommand{\prob}[1]{\mathbb{P}\left\lbrace #1\right\rbrace}
\def\trace{{\sf trace}}
\def\bigo{\mathcal{O}}
\newcommand{\bigmid}{\,\middle\vert\,}
\DeclarePairedDelimiter{\ceil}{\lceil}{\rceil} 
\def\E{\mathbb{E}} 
\def\supp{{\sf supp}}

\def\i{\textrm{i}} 
\def\normal{{\sf N}}
\def\reals{{\mathbb R}}
\def\complex{{\mathbb C}}
\def\nonNegInt{\mathbb{N}_0} 

\def\tv{\tilde{v}}
\def\bv{\boldsymbol{v}}
\def\tbv{\tilde{\boldsymbol{v}}}

\def\ba{\boldsymbol{a}}
\def\bu{\boldsymbol{u}}
\def\bw{\boldsymbol{w}}
\def\tbw{\tilde{\boldsymbol{w}}}
\def\be{\boldsymbol{e}}

\def\tbu{\tilde{\boldsymbol{u}}}
\def\tu{\tilde{u}}

\def\bx{\boldsymbol{x}}

\def\hbx{\hat{\boldsymbol{x}}}
\def\hx{\hat{x}}

\def\tk{\tilde{k}}
\def\tn{\tilde{n}}
\def\by{\boldsymbol{y}}

\def\bz{\boldsymbol{z}}

\def\bX{\boldsymbol{X}}
\def\hbX{\hat{\boldsymbol{X}}}
\def\hX{\hat{X}}
\def\bW{\boldsymbol{W}}

\def\bA{\boldsymbol{A}}
\def\bY{\boldsymbol{Y}}
\def\bB{\boldsymbol{B}}
\def\bH{\boldsymbol{H}}

\def\bS{\boldsymbol{S}}
\def\bs{\boldsymbol{s}}
\def\bH{\boldsymbol{H}}
\def\bh{\boldsymbol{h}}
\def\bSigma{\boldsymbol{\Sigma}}
\def\bI{\boldsymbol{I}}
\def\bPsi{\boldsymbol{\Psi}}

\def\barV{\bar{V}}
\def\tZ{\tilde{Z}}

\def\bzero{\boldsymbol{0}}

\newcommand{\Ut}[1]{ \mc{U}(#1) }
\newcommand{\Vt}[1]{ \mc{V}(#1) }
\newcommand{\uSt}[1]{ S^{(u, #1)} }
\newcommand{\vSt}[1]{ S^{(v, #1)} }
\def\vtS{S^{v}}

\def\hk{\hat{k}}

\newcommand{\pois}[1]{\text{Pois}\left(#1\right)}
\newcommand{\bin}[2]{\text{Bin}\left(#1, #2\right)}
\newcommand{\bern}[1]{\text{Bern}\left(#1\right)}

\newcommand\whitesqr{\begin{tiny} \square\end{tiny}}
\newcommand\blacksqr{\crule[black]{0.25cm}{0.25cm}}
\newcommand\pinksqr{\crule[red!50!white!100]{0.25cm}{0.25cm}} 
\newcommand\bluesqr{\crule[blue!50!white!100]{0.25cm}{0.25cm}}
\newcommand\cyansqr{\crule[cyan!50!white!100]{0.25cm}{0.25cm}}

\newcommand\pinkbluesqr{\crule[red!50!white!100]{0.125cm}{0.25cm} \crule[blue!50!white!100]{0.125cm}{0.25cm}}
\newcommand\pinkbluecyansqr{\crule[red!50!white!100]{0.08cm}{0.25cm} \crule[blue!50!white!100]{0.08cm}{0.25cm} \crule[cyan!50!white!100]{0.08cm}{0.25cm} }

\def\nMeas{m} 
\def\nBins{R} 

\def\nlnodes{t} 

\def\deltaRangeInLine{\delta\in (0, 1)}

\def\NDiag{N_D} 
\def\usfB{ \sfB^{u} }
\def\vsfB{ \sfB^{v} }

\def\utZ{ Z^u}
\def\vtZ{ Z^v}

\def\sfA{{\sf{A}}}

\def\alphaUB{\alpha^*} 

\def\sfB{{\sf{B}}}

\def\Balti{\sfB_{\text{alt},i}}

\def\nConst{\textcolor{black}{17}} 

\newcommand\crule[3][black]{\textcolor{#1}{\rule{#2}{#3}}} 
\usepackage{tikz}
\usetikzlibrary{patterns,shapes,arrows}
\usetikzlibrary{decorations.pathreplacing}
\tikzstyle{cnode}=[rectangle, draw, fill=blue!20, minimum size=.5cm, scale=0.8]
\tikzstyle{vnode}=[circle, draw, fill=blue!20, minimum size=.5cm, scale=0.8]
\tikzstyle{cnodeFaded}=[rectangle, draw, fill=blue!20, opacity=0.2, minimum size=.5cm, scale=0.8]
\tikzstyle{vnodeFaded}=[circle, draw, fill=blue!20, opacity=0.2, minimum size=.5cm, scale=0.8]

\tikzstyle{cnodex}=[rectangle, draw, fill=red!20, minimum size=.5cm, scale=0.7]
\tikzstyle{vnodex}=[circle, draw, fill=red!20, minimum size=.5cm, scale=0.8]
\tikzstyle{cnodexFaded}=[rectangle, draw, fill=red!20, opacity=0.2,  minimum size=.5cm,  scale=0.7]
\tikzstyle{vnodexFaded}=[circle, draw, fill=red!20, opacity=0.2, minimum size=.5cm, scale=0.8]

\tikzstyle{line} = [draw, -latex']

\tikzstyle{arrow} = [draw, -latex']

\newcommand{\edit}[1]{{\textcolor{black}{#1}}}

\begin{document}

\title{Sketching sparse low-rank matrices with near-optimal \\ sample- and time-complexity \edit{using message passing}}
\author{Xiaoqi Liu\thanks{Xiaoqi Liu was supported in part by a Schlumberger Cambridge International Scholarship funded by the Cambridge Trust. X. Liu and R. Venkataramanan are with the Department of Engineering, University of Cambridge. Email: \texttt{\{xl394, rv285\}@cam.ac.uk}. This paper was presented in part at the 2022 IEEE International Symposium on Information Theory.}\;\;\;and\;\;\;Ramji Venkataramanan}

\maketitle

\begin{abstract}
We consider the problem of recovering an $n_1 \times n_2$ low-rank matrix with $k$-sparse singular vectors from a small number of linear measurements (sketch). We propose a sketching scheme and an algorithm that can recover the singular vectors with high probability, with a sample complexity and running time that both depend only on $k$ and not on the ambient dimensions $n_1$ and $n_2$. Our sketching operator, based on a scheme for compressed sensing by Li et al. \cite{li2019sublinear} and Bakshi et al. \cite{bakshi2016SHOFA}, uses a combination of a sparse parity check matrix and a partial DFT matrix. Our main contribution is the design and analysis of a two-stage iterative algorithm which recovers the singular vectors by exploiting the simultaneously sparse and low-rank structure of the matrix. We derive a nonasymptotic  bound on the probability of exact recovery, \edit{which holds for any $n_1\times n_2 $ sparse, low-rank matrix}.  We also  show how the scheme can be adapted to tackle matrices that are approximately sparse and low-rank. The theoretical results are validated by numerical simulations \edit{and comparisons with existing schemes that use convex programming for recovery.}
\end{abstract}


\section{Introduction} \label{sec:intro}

We consider the problem of \emph{sketching} a large data matrix which is low-rank with sparse singular vectors.  A sketch is a compressed representation of the matrix, obtained via linear measurements of the matrix entries  \cite{woodruff2014sketching,tropp2017practical}. The goal is to design a sketching scheme so that the singular vectors can be efficiently recovered from a sketch of a much smaller dimension.  

Given a  data matrix $\bX \in \reals^{n_1  \times n_2}$, its sketch \edit{$\by \in \complex^{m}$} is obtained as $\by = \mc{A}(\bX)$, where $\mc{A}$ is a linear operator defined via $\nMeas$  prespecified matrices $\bA_1, \ldots, \bA_\nMeas \in \complex^{{n_1} \times {n_2}}$. The entries of the sketch are:
\beq
\label{eq:general_matrix_recovery_model}
y_i = \trace(\bA_i^T \bX) = \sum_{j=1}^{n_1} \sum_{\ell=1}^{n_2}  (A_i)_{j,\ell} X_{j,\ell}, \quad i =1, \ldots, \nMeas.
\eeq
We consider   data matrices $\bX$ of the form
 \beq
 \bX = \bX_0 + \bW, 
 \label{eq:noisy_observation_model}
 \eeq 
 with
 \begin{equation}\label{eq:def_X0}
        \bX_0 = \sum_{i=1}^r \lambda_i \bv_i \bv_i^{T}
    \ (\text{symmetric case}), \ \text{ or } \  \bX_0 = \sum_{i=1}^r \sigma_i \bu_i \bv_i^{T}  \  (\text{non-symmetric case}).
 \end{equation}
 In the symmetric case, $\lambda_1 \ge \lambda_2 \ge \ldots \ge \lambda_r$ are the eigenvalues, $\{ \bv_i \}_{i \in [r]}$  are the corresponding unit-norm eigenvectors. In the non-symmetric case, $\sigma_1 \ge \sigma_2 \ge \ldots \ge \sigma_r >0$ are the singular values, and $\{ \bu_i, \bv_i \}_{i \in [r]}$ the corresponding unit-norm singular vectors. In both cases, $\bW$ is a noise matrix (assumed to be symmetric when $\bX_0$ is symmetric). 
 Letting $n := \min (n_1, n_2)$,  we consider the regime where $n$ grows and the rank $r$ of $\bX_0$ is a constant that does not scale with $n$.  In the rest of the paper, we refer to the $\{ \bv_i \}$ in the symmetric case and the $\{ \bu_i, \bv_i \}$ in the non-symmetric case as \emph{signal vectors}.

\paragraph{Sparsity constraint} 
\edit{We assume that the signal vectors are $k$-sparse, i.e., they each have at most $k$ nonzero entries, where  $k  \ll n$, and both $k$ and $n$ are large. This is a regime that is particularly relevant in  applications such as genomics and neuroscience, where
$n$ is very large because of the vast number of features and samples, but the number of  factors $k$ determining the structure of the data may be much smaller
\cite{ma2014learning, ma2014adaptive}.
}



The goal is to design a sketching  operator $\mc{A}$ and an efficient algorithm to recover the  signal vectors from the sketch $\by$. In particular, we want the sample complexity and the running time of the  algorithm to both depend only on the sparsity level $k$, and not on the ambient dimension $n$. (The sample complexity is the sketch size $\nMeas$.) We note that the signal vectors can be recovered only up to a sign ambiguity as  the pair $(\bu_i, \bv_i)$ cannot be distinguished from $(-\bu_i, -\bv_i)$. Furthermore, if there are repeated eigenvalues or singular values,  the corresponding signal vectors are recovered up to any rotations within the subspace they span.

 Simultaneously sparse and low-rank matrices  arise in applications such as sparse PCA \cite{zou2006sparse,hastie2015statistical}, sparse bilinear inverse problems \cite{lee2018near},  community detection   \cite{girvan2002community} and biclustering \cite{jiang2004cluster, lee2010biclustering}.  In particular, the adjacency matrices of graphs with community structures such as social networks and protein interaction datasets are symmetric, sparse and block-diagonal in appropriate bases  \cite{richard2012estimation}. In biclustering, the sample-variable associations in high dimensional data, like gene expression microarray datasets, often correspond to sparse non-overlapping submatrices (or biclusters) in the data matrix.

\subsection{Main contributions}

Our sketching operator is defined via a combination of a sparse parity check matrix and a partial DFT matrix. This operator was proposed for compressed sensing of (vector) signals in \cite{li2019sublinear,bakshi2016SHOFA}. \edit{The sparsity and  DFT structure of the operator ensure that the cost of computing the sketch is low.} Our main contribution is the design and analysis of a two-stage iterative algorithm to recover the signal vectors from the sketch $\by$. The two stages use the sparsity and low-rank properties of the signal matrix $\bX_0$ to iteratively identify and solve equations with a single unknown.

\paragraph{Noiseless setting} The sketching scheme and the recovery algorithm are described in Section \ref{sec:scheme_and_algorithm} for the noiseless case, that is, when $\bW = 0$  in \eqref{eq:noisy_observation_model}.  In Section \ref{sec:main_results}, we  provide theoretical guarantees on the performance of the scheme.  
For sufficiently large $k$, Theorem \ref{thm:main_result_symm} shows that the scheme has the following features when $\bX_0$ is symmetric in \eqref{eq:def_X0}:

\begin{enumerate}[label=\arabic*)]
\item \label{enum:main_contribution_disjoint_supports} When the supports of the signal vectors  $\{\bv_i\}$ are disjoint, fix any $\deltaRangeInLine$, with probability at least  $1 \, -  2r\exp(-\frac{1}{30} k^{1-\delta})$,  the algorithm recovers the signal vectors with sample complexity $ \nMeas = 3 rk^2/ (\delta \ln k)$ and running time $\bigo(rk^2/\ln k)$.  \edit{Though this result does not assume any randomness  of the signal vectors}, we note that with $k = \bigo(n^{\nu})$  for $\nu \in (0, \frac{1}{2})$, and the support locations of each $\bv_i$ uniformly random, the supports of the $\{\bv_i\}$ will be disjoint with probability at least $ 1 -3r^2k^2/n$ .

\item For the general case, where the supports of the  signal vectors are not disjoint, with probability at least $1-\bigo(k^{-2})$, the algorithm recovers the  signal vectors  with sample complexity $\nMeas = 2rk^2$ and running time $\bigo\left((rk)^3\right)$. 
\end{enumerate}

Theorem \ref{thm:nonsym_main_results} provides similar guarantees when $\bX_0$ is non-symmetric as defined in \eqref{eq:def_X0}. The numerical simulations in Section \ref{subsec:numerical_results} validate the theoretical guarantees. \edit{The probabilistic statements in Theorems \ref{thm:main_result_symm} and \ref{thm:nonsym_main_results} are with respect to the randomness in the sketching operator $\mc{A}$, and hold for any rank-$r$ matrix with $k$-sparse singular vectors.}


\emph{Near-optimal sample complexity}:  In the symmetric case, with $r$ orthonormal eigenvectors $\{\bv_i\}$, each $k$-sparse, even if the locations of the nonzeros in $\{\bv_i\}$ are known, the degrees of freedom in the eigendecomposition of $\bX_0$ is  close to $rk$. 
In the non-symmetric case,  the degrees of freedom in the singular value decomposition of $\bX_0$ is close to $2r(k-\frac{1}{2})$.
 Therefore the sample complexity of the proposed scheme is larger than the degrees of freedom by a factor of at most $O(k)$. Moreover, neither the sample complexity or the running time depends on the ambient dimension $n$. 

\paragraph{Noisy setting} The sketching scheme and recovery algorithm are extended to the noisy case in Section \ref{sec:noisy_description} (i.e., when $\bW\neq 0$ in \eqref{eq:noisy_observation_model}). Here, the sample complexity is \edit{$\bigo(r k^2 \log (n/k))$}, a factor of \edit{$\bigo(\log (n/k))$} greater than that in the noiseless case. This is because additional sketches are needed to reliably identify the locations of nonzero matrix entries in the presence of noise. This also increases the running time of the  recovery algorithm to 
$\bigo(\max\{n^2\log(n/k), (rk)^3\})$. 
We do not provide theoretical performance guarantees in  the noisy setting, but show via numerical simulations that the scheme is robust to moderate levels of noise. 

 \subsection{Related work}

 The sketching operator we use was  proposed in \cite{bakshi2016SHOFA,li2019sublinear} for compressed sensing of vectors.  Variations of the operator have  been used for numerous applications including  sparse DFT \cite{pawar2018FFAST,janakiraman2015exploring}, sparse  Walsh-Hadamard Transform \cite{li2015SPRIGHT,scheibler2015Fast,chen2015robust}, compressive phase retrieval \cite{ShengSUPER14, pedarsani2017phaseCode}, sparse covariance estimation \cite{pedarsani2015sparse}, sparse polynomial learning and graph sketching \cite{li2015active}, and learning mixtures of sparse linear regressions \cite{yin2019learning}. 

The two-stage recovery algorithm we propose is analogous to the peeling decoder for Low Density Parity Check (LDPC) codes over an erasure channel \cite{luby2001efficient,richardson2008modern}, and the first stage of the algorithm is similar to the one used for compressed sensing in \cite{bakshi2016SHOFA,li2019sublinear}. However, a key difference from these works is that our sketching matrix has row weights that scale with the sparsity level $k$ and therefore, the existing peeling decoder analysis based on density evolution \edit{and Doob Martingales} cannot be applied. 
We characterize the performance of the algorithm  by obtaining nonasymptotic probability bounds on the number of unknown nonzero entries after each stage.  For the first stage, this is done by establishing negative association (defined in Section \ref{subsec:prelim}) between the right node degrees in the associated bipartite graph, and using  concentration inequalities for negatively associated random variables \cite{dubhashi1998balls}. \edit{To the best of our knowledge, this technique has not been used
previously in the sparse-graph codes literature.} 
For the second stage, we model the algorithm as a random graph process on another bipartite graph, and we analyze an alternative random graph process whose evolution is easier to characterize than the original one.

 
 \emph{Other related work}: Recovering low-rank matrices from linear measurements (without sparsity constraints) has been widely studied in the past decade; see \cite{davenport2016anoverview} for an overview. A key result in this area is that if  the linear measurement operator satisfies a matrix restricted isometry property, then the low-rank matrix can be recovered via nuclear-norm minimization \edit{\cite{recht2010guaranteed,candes2011tight}}. 
  For matrices that are simultaneously low-rank and sparse (with $k= o(n)$), such optimization-based approaches are highly sub-optimal with respect to sample complexity. 
  
Several authors have investigated sketching schemes for sparse matrices \cite{wimalajeewa2013recovery,dasarathy2015sketching}, and simultaneously sparse and low-rank matrices  \cite{lee2018near,oymak2015simultaneously,bahmani2016nearopt, ma2014adaptive}. Moreover, \cite{chen2015exact, bahmani2015sketching} studied the related problem of recovering a sparse low-rank covariance matrix from rank-1 sketches of a  sample covariance (the sketching matrices $\bA_1, \ldots, \bA_m$ are rank-1). The recovery algorithms in all these works 
\edit{are based on convex or non-convex optimization, and have running time  polynomial in $n$.} The sample complexity also depends at least logarithmically on $n$.  \edit{The signal models considered in the works vary, some  noise-free and others with additive  noise. We refer the reader to Table \ref{tab:noiseless} (noiseless case) and Table \ref{tab:noisy} (noisy case)
 for a summary of existing works in comparison to our work.}


\begin{table}[H]
\color{black}
\centering
\begin{adjustwidth}{}{}
\small
\begin{tabular}{||c|c|c|c|c||} 
 \hline
 Reference & Model Assumption & Algorithm & Sample ($m$) & Time \\ [0.5ex] 
 \hline\hline
 \cite[Thm. 3(a1)]{oymak2015simultaneously}  & Sparse \& LR & Convex opt.& $\Omega( rn)$  & -\\ 
 \hline
 \cite[Thm. 4]{chen2015exact}  & Sparse \& rank-$1$, symm.  & Convex opt.&  $\bigo(k^2\log n)$ & - \\
 \hline
 \cite[Thm. 1]{pedarsani2015sparse}  & Sparse only, PSD & MP &$\bigo(rk^2)$ &$\bigo(rk^2)$\\
 \hline
 \cite[Thm. 1.1]{candes2013PhaseLift} &  Rank-1 only, symm.& Convex opt.& $\bigo(n\log n)$ & -\\
 \hline
 \cite[Thm. 1]{dasarathy2015sketching}  &Sparse only & Convex opt. & $\bigo(kn \log^2 n)$ & -\\
 \hline
 Our Thm. 1 \& 2 part 1) &Sparse \& LR, disjoint supp. & MP & $\bigo(rk^2/\log k)$ & $\bigo(rk^2/\log k)$\\
 \hline
 Our Thm. 1 \& 2 part 2) & Sparse \& LR, overlapping supp. &MP  & $\bigo(rk^2)$ & $\bigo((rk)^3)$ \\ 
 \hline
\end{tabular}
\end{adjustwidth}
\caption{\small\edit{Comparison of our sketching scheme with existing works on low-rank (LR) and/or sparse matrix recovery. The signal matrix  $\bX \in  \reals^{n \times n}$ has    rank $r$, with $k$-sparse singular vectors.  ``MP" denotes message passing, and the dashes indicate unreported results. The running time of convex optimization programs is typically $\bigo(\textrm{poly} (n))$.}}\label{tab:noiseless}
\end{table}

\vspace{-0.3cm}

\begin{table}[H]
\centering
\color{black}
\begin{adjustwidth}{-0.5cm}{}
\small
\begin{tabular}{||c|c|c|c|c||} 
 \hline
 Reference & Model Assumption& Algorithm & Sample ($m$) & Time  \\ [0.5ex] 
 \hline\hline
 \cite[Thm. 1]{bahmani2015sketching} \cite[Thm. 3]{bahmani2016nearopt}& Sparse \& LR, PSD &Convex opt. & $\bigo(rk\log (n/k))$ & -\\
  \hline
 \cite[Thm. 6]{lee2018near}   & Sparse \& LR &  Non-convex opt. & $\bigo(rk \log(n/k))$ & $\bigo(rkn^2)$  \\
 \hline
 \cite[Thm. 4]{chen2015exact}  & Sparse \& rank-$1$, symm.   & Convex opt.&  $\bigo(n^2\log n)$ & - \\
 \hline
 \cite[Thm. 3]{pedarsani2015sparse}  & Sparse only, PSD & MP &$\bigo(rk^2 \log^2 n)$& - \\
 \hline
 \cite[Thm. 2.3]{candes2011tight} & LR only  & Convex opt.&  $\bigo(rn)$& -\\
 \hline 
 \cite[Sec. 2]{wimalajeewa2013recovery}   & Sparse only (same as \cite{dasarathy2015sketching}) & Convex opt.& -& $\bigo(n^3)$\\
 \hline
 Our scheme& Sparse \& LR & MP& $\bigo(rk^2 \log (n/k))$ & $\bigo(n^2 \log(n/k))$
\\[1ex] 
\hline
\end{tabular}
\end{adjustwidth}
\caption{\small \edit{Comparison of our scheme with existing works on low-rank (LR) and/or sparse matrix recovery in the noisy case, i.e., the signal matrix is $\bX=\bX_0+\bW$ where $\bX_0 \in \reals^{n \times n}$ has rank $r$, with $k$-sparse singular vectors. 
}}
\label{tab:noisy}
\end{table}


\edit{We remark that our scheme is similar in spirit to the algorithm in \cite{bahmani2015sketching, bahmani2016nearopt}, which exploits the sparsity and low-rank structure  in different stages. The  scheme in \cite{bahmani2015sketching, bahmani2016nearopt} uses a nested linear sketching operator, with one part being a restricted isometry for low-rank matrices and the other a restricted isometry for sparse matrices. The recovery algorithm, based on convex programming, correspondingly has two stages, one for low-rank estimation and the other sparse estimation. The sample complexity of the scheme is  $\bigo(rk \log(n/k))$   which is similar to the $\bigo(rk^2\log(n/k))$ required by our scheme in the noisy setting. However the  recovery algorithm is less robust to noise and  significantly slower than ours, as evidenced by the numerical experiments in Section \ref{sec:noisy_comparison}.}

The problem of estimating sparse eigenvectors has also been widely studied in the context of sparse PCA \cite{zou2006sparse, johnstone2009consistency,AminiWain09,BirnbaumSPCA13}. In these works, the principal eigenvector of the population covariance matrix is assumed to be sparse. The goal is to recover the principal eigenvector from the sample covariance matrix. This is distinct from the sketching problem considered here.

Some \edit{bilinear} inverse problems such as phase retrieval \cite{candes2013PhaseLift} and blind deconvolution \cite{ahmed2014blind} can be \edit{linearized by lifting. That is, they can be reformulated into  problems of} recovering a rank-1 signal matrix from linear measurements. \edit{This is equivalent to  low-rank matrix recovery in the rank-1 case, but is different for higher rank. Moreover,} the structure of the linear operator in these problems is constrained by the applications, whereas  the operator in our setting can be designed flexibly.

\subsection{Notation}
We write $[n]$ for the set of integers $\left\lbrace1, 2, \dots, n\right\rbrace$, $\nonNegInt$ for $\{ 0, 1, 2, \ldots \}$, and use $\i$ to denote $\sqrt{-1}$.
For a length-$n$ vector $\bu$, we denote the set of  nonzero locations by  $\supp(\bu) : = \{\ell\in [n]: u[\ell]\neq 0 \}$.  The notation $x_n =  o(n^c)$ is used to denote a positive number $x_n$ such that $\frac{x_n}{n^c} \to 0$ as $n \to \infty$.  

The Bernoulli distribution with parameter $p\in[0,1]$ is denoted by $\bern{p}$, and the Binomial distribution with parameters $n \in  \mathbb {N} $ and $p \in[0,1]$ is denoted by  $\bin{n}{p}$.  $\pois{\lambda}$  denotes a Poisson distribution with parameter $\lambda>0$.

\section{Baseline sketching scheme and recovery algorithm} \label{sec:scheme_and_algorithm}
We first focus on  the noiseless symmetric case in Sections \ref{subsec:sketching_scheme}--\ref{subsec:recovery_alg_rankr}, and then extend the scheme to non-symmetric matrices in Section \ref{subsec:nonsymm}. The noisy setting is discussed in Section \ref{sec:noisy_description}.
\subsection{Sketching scheme}
\label{subsec:sketching_scheme}
Consider the noiseless symmetric sparse, low-rank matrix $\bX = \sum_{i=1}^r \lambda_i \bv_i\bv_i^T$ in \eqref{eq:noisy_observation_model}. 
Let $\tn = \binom{n}{2} +n$, and let  $\bx \in \reals^{\tn}$  be the vectorized upper-triangular part of $\bX$. 
The  sketch is $\by = \bB \bx$, with  the  sketching matrix $\bB \in \mathbb{C}^{2\nBins \times \tn}$ described below.  The sample complexity is $\nMeas= 2\nBins$ with the exact value of $\nBins = \bigo(rk^2/\ln k) $  given in Theorem \ref{thm:main_result_symm}.
\edit{We emphasize that vectorizing the signal matrix is just a way of explaining the scheme that makes the notation and analysis in the sequel cleaner.}

The sketching matrix $\bB \in\mathbb{C}^{2\nBins\times \tn}$ is constructed as in \cite{li2019sublinear} by taking the column-wise Kronecker product of two matrices: a sparse  parity check matrix
\beq 
\bH :=\begin{bmatrix}
 \bh_1  & \bh_2 &\dots & \bh_{\tn}
\end{bmatrix} \in \{0,1\}^{\nBins\times \tn}, 
\eeq 
and a matrix $\bS$ consisting of the first two rows of an $\tn$-point DFT matrix:
\beq
\bS := \begin{bmatrix}
1&1&1& \ldots &1  \\
1&W& W^2& \ldots & W^{\tn-1} 
\end{bmatrix}\quad \text{where }W=\exp\left( \frac{2\pi \i }{\tn}\right). 
\label{eq:two_row_DFT}
\eeq
Then,
\beq\label{eq:def_B_matrix_as_tensor_product_of_H_and_S}
\small{\bB : =\begin{bmatrix}
\bh_1\otimes 
\begin{bmatrix}
1\\ 1
\end{bmatrix} & 
\bh_2\otimes 
\begin{bmatrix}
1\\ W
\end{bmatrix} & \dots & 
\bh_{\tn}\otimes 
\begin{bmatrix}
1\\ W^{\tn-1}
\end{bmatrix}
\end{bmatrix}.}
\eeq
For  column vectors  $\boldsymbol{a}, \boldsymbol{b}$ with lengths $n_a, n_b$, we recall that the Kronecker product $\boldsymbol{a} \otimes \boldsymbol{b}$  is the length-$(n_a+n_b)$ vector $[a_1\boldsymbol{b}^T,  a_2\boldsymbol{b}^T, \dots a_{n_a} \boldsymbol{b}^T]^T$.  As an example, let $\tn=6$, $\nBins=4$  and
\begin{align} 
\label{eq:H_and_S_matrix_example}
&
\bH=\begin{bmatrix}
1&1&1&0&1&1  \\
0&0&0&1&1&1 \\
1&0&0&0&0&0  \\
0&1&1&1&0&0 
\end{bmatrix}, \ \ 
\bS=\begin{bmatrix}
1&1& \ldots &1  \\
1&W&\ldots & W^5 
\end{bmatrix}.
\end{align}

Then the sketching matrix $\bB$  is
\begin{align}\label{eq:B_matrix_example}
& \bB=\begin{bmatrix}
1&1&1          &0       &1&1 \\
1&W&W^2     &0      &W^4&W^5 \\
0&0&0          &1      &1&1 \\
0&0&0          &W^3  &W^4&W^5 \\
1&0&0          &0       &0&0 \\
1&0&0          &0       &0&0  \\
0&1&1          &1       &0&0 \\
0&W&W^2     &W^3  &0&0
\end{bmatrix}.
\end{align}

We choose $\bH$ to be column-regular, with each column containing $d \ge 2$ ones at locations chosen uniformly at random; the example in \eqref{eq:H_and_S_matrix_example} uses $d=2$. The sparse  matrix $\bH$  determines which  nonzero  entries of  $\bx$ contribute to each entry of the sketch $\by$.

It is convenient to view the sketch $\by = \bB \bx \in \complex^{2\nBins}$ as having $\nBins$ pairs of entries.  We denote the $j$-th pair   
$ \begin{bmatrix}
y_{(2j-1)}, y_{(2j)}]
\end{bmatrix}^T$ by $\by_j \in \complex^2$ and call it the $j$-th sketch \emph{bin}. 
Let $\bB_{j} \in \complex^{2\times \tn}$ denote the matrix containing the $j$-th pair of rows in $\bB$, for $j\in [\nBins]$. We observe that the $j$-th bin can be expressed as 
\beq \label{eq:y_j_general_expression}
 \by_j = \bB_j \bx =\bS \bx^*_{j} ,
 \eeq 
 where the entries of $\bx^*_{j} \in \reals^{\tn} $ are  
\begin{equation}\label{eq:def_sparsified_x_vector}
 (x^*_{j})_{\ell}=\left\{
                \begin{array}{ll}
                  x_\ell & \ell\in \mc{N}(j)\\
                  0 & \ell \notin \mc{N}(j)
                \end{array}
              \right. \quad \text{for } 
              \ell \in [\tn] .
\end{equation}
Here, $\mc{N}(j) := \{\ell\in[\tn]: \,  H_{j,\ell} = 1\}$ contains the  locations of the ones in the $j$-th row of $\bH$.    In words, \eqref{eq:y_j_general_expression} shows that the bin $\by_j$ is a  linear combination of the $\bx$ entries  that correspond to the ones in the $j$-th row of $\bH$. Specifically,  $\by_j$ is the sum of these $\bx$ entries, weighted by their respective columns in $\bS$.
Note that \eqref{eq:y_j_general_expression} can be recast into two equations in the form of \eqref{eq:general_matrix_recovery_model}, one for each component in the bin $\by_j$.  Moreover, when only one nonzero $\bx$ entry contributes to $\by_j$, the DFT structure of $\bS$ enables exact recovery of that nonzero  from $\by_j$.

\subsection{Recovery algorithm for rank-1 symmetric matrices}
\label{subsec:recovery_algorithm}

For ease of exposition, we first describe the recovery algorithm for rank-1 matrices (i.e., $\bX=\lambda \bv\bv^T$),  and then explain how to extend the algorithm to tackle rank-$r$ matrices. The algorithm has two stages, which we refer to as stage A  and  stage B. 
\begin{itemize}
\item In stage A, by  exploiting the sparsity of $\bX$, the algorithm  iteratively recovers as many of the nonzero entries in $\bX$ as possible from the sketch  $\by$. This stage is similar to the compressed sensing recovery algorithm used in \cite{bakshi2016SHOFA, li2019sublinear}.  For sufficiently large $k$, stage A  recovers at least one nonzero diagonal entry and a  fraction $ k^{-\delta}$ of the nonzero above-diagonal entries in $\bX$ with high probability, for a constant $\deltaRangeInLine$.

\item In stage B, by exploiting the rank-1 structure of $\bX$, the algorithm iteratively recovers  the nonzero entries in  $\bv$  from the partially recovered matrix $\bX$. Theorem \ref{thm:main_result_symm}  shows that with a sketch of size $3k^2/\ln k$, the algorithm recovers $\bv$ by the end of stage B with high probability, for sufficiently large $k$.
\end{itemize}

We note that the recovery algorithm does not require knowledge of the sparsity $k$.

\subsubsection{Stage A of the algorithm}

Following the terminology in \cite{li2019sublinear},  we classify each of the $\nBins$ bins of $\by$ based on how many  nonzero entries of $\bX$ contribute to the pair of linear combinations defining  the bin. Specifically,  a bin  $\by_j$ is referred to as a 
\emph{zeroton, singleton},  or \emph{multiton} respectively if it involves zero, one,  or more than one nonzero entries  of $\bX$, that is, when $\abs{\supp(\bx) \cap \mc{N}(j) } =0, 1$ or $>1$ with $\mc{N}(j) $ as defined in \eqref{eq:def_sparsified_x_vector}. This is illustrated by the following example with the sketching matrix $\bB$ in \eqref{eq:B_matrix_example}:
%
\begin{align*}
\label{eq:colour_mixture_model_example}
\by\in \complex^{2\nBins} &\quad\quad\quad\quad\bB\in \complex^{2\nBins \times \tn}\quad\quad\quad \bx \in \reals^{\tn} \nonumber\\
\begin{matrix}
\by_1 \\ \\ \by_2 \\\\ \by_3 \\\\ \by_4 
\end{matrix}
\begin{bmatrix}
\pinkbluecyansqr\\
\pinkbluecyansqr\\
\cyansqr\\
\cyansqr \\
\whitesqr\\
\whitesqr\\
\pinkbluesqr\\
\pinkbluesqr
\end{bmatrix}
=&\begin{bmatrix}
 \blacksqr & \blacksqr &\blacksqr & \whitesqr & \blacksqr& \blacksqr\\
 \blacksqr & \blacksqr &\blacksqr & \whitesqr & \blacksqr & \blacksqr  \\
 \whitesqr & \whitesqr &\whitesqr & \blacksqr & \blacksqr & \blacksqr  \\
 \whitesqr & \whitesqr & \whitesqr & \blacksqr & \blacksqr &  \blacksqr \\
\blacksqr &\whitesqr & \whitesqr & \whitesqr & \whitesqr & \whitesqr  \\
\blacksqr &\whitesqr & \whitesqr & \whitesqr & \whitesqr & \whitesqr \\
\whitesqr & \blacksqr& \blacksqr& \blacksqr&\whitesqr & \whitesqr  \\
\whitesqr & \blacksqr& \blacksqr& \blacksqr&\whitesqr & \whitesqr 
\end{bmatrix}
\begin{bmatrix}
\whitesqr\\ \pinksqr \\ \bluesqr \\\whitesqr \\\whitesqr \\\cyansqr  
\end{bmatrix}
\begin{matrix}
x_1 \equiv X_{11}\\x_2 \equiv X_{12}\\x_3 \equiv X_{13}\\x_4 \equiv X_{22}\\x_5\equiv X_{23}\\x_6 \equiv X_{33}
\end{matrix}  \ .
\end{align*}
%
The white squares (\begin{small}$\square$\end{small}) represent zeros, and the black squares (\begin{small}$\blacksquare$\end{small}) represent  the DFT coefficients in $\bB$. The coloured squares in $\bx$ represent the nonzero entries. The  sketch  $\by$  consists of  $\nBins=4$  bins, of which $\by_{3}$ (\begin{small}$\square$\end{small}) is a zeroton, 
 $\by_2$ (\begin{small}$\cyansqr$\end{small}) is a singleton,
and $\by_1$ (\begin{small}$\pinkbluecyansqr$\end{small}) and $\by_4$ (\begin{small}$\pinkbluesqr$\end{small}) are multitons.

The first step is to classify each bin $\by_j$ into one of the three types. For brevity, we denote the two components of  $\by_j$ as $y_{(2j-1)}$  and $y_{(2j)}$. If both $y_{(2j-1)}$  and $y_{(2j)}$  are zero, then $\by_j$ is declared  a zeroton. This is because the DFT coefficients in $\bB$ ensure that $y_{(2j)}$ can be zero only when the $\bx$  entries involved in the linear combination are all zero, i.e., $ \supp(\bx)\cap \mc{N}(j)= \emptyset$. Next, to identify singletons,  \eqref{eq:y_j_general_expression} and \eqref{eq:def_sparsified_x_vector} indicate that a singleton takes the form 
 \beq
 \by_j= \begin{bmatrix} y_{(2j-1)}\\ y_{(2j)}\end{bmatrix} 
 =x_{\ell}\begin{bmatrix}
 1\\ W^{\ell-1}
 \end{bmatrix},
 \label{eq:singleton_form_noiseless}
 \eeq
 where $x_\ell$ is the sole nonzero entry contributing to the bin $\by_j$, i.e., $\supp (\bx)\cap \mc{N}(j) =\ell$. 
 While $y_{(2j-1)}$ directly captures the value of $x_\ell$,  the components $y_{(2j-1)}$ and $y_{(2j)}$ together capture the location $\ell$.  Therefore, to identify a singleton,  the algorithm computes 
 \begin{equation}
\hat{\ell} =\frac{\arg\left\lbrace y_{(2j)}/y_{(2j-1)} \right\rbrace}{ 2\pi/\tn} +1,\quad \hat{x}_{\hat{\ell}}=y_{(2j-1)}.\label{eq:noiseless_singleton_test}
\end{equation}
 If $\abs{y_{(2j-1)}}=\abs{y_{(2j)}}$ and the estimated index $\hat{\ell}$ takes an integer value in $\left[\tn \right]$, the algorithm declares $\by_j$ a singleton and  the $\hat{\ell}$-th entry of $\bx$ as $\hat{x}_{\hat{\ell}}$.  The algorithm declares $\by_j$ a multiton if it is found to be neither a zeroton nor a singleton.  
 
 \edit{Note that when the target signal matrix is complex-valued, the linear measurements in  a bin may be all-zero even when  the matrix entries involved in the bin are not all zeros. Therefore, one way to tackle complex signal matrices would be to measure and recover their real and imaginary parts  separately using the proposed scheme.
}

 Let $t\in \nonNegInt$ denote the iteration number and  let $\mc{S}(t)$ be the set of singletons at time $t$. At $t=0$, the algorithm  identifies the initial set of singletons $\mc{S}(0)$ among  the $\nBins$ bins. Then at each $t \ge 1$,   the algorithm picks a singleton uniformly at random from $\mc{S}(t-1)$, and recovers the support of $\bx$ underlying the singleton using the  index-value pair $(\hat{\ell}, \hat{x}_{\hat{\ell}})$ in  \eqref{eq:noiseless_singleton_test}.
  The algorithm then subtracts (or `peels off') the contribution of the $\hat{\ell}$-th entry of $\bx$  from each of the $d$ bins the entry is involved in,   re-categorizes these bins, and updates the set of singletons to include any new singletons created by the peeling of the $\hat{\ell}$-th entry. The updated set is $\mc{S}(t)$.  The algorithm continues until singletons run out, i.e., when $\mc{S}(t) =  \emptyset$. \edit{The pseudocode for stage A is provided in Algorithm  \ref{alg:nl_stage_A} below, with line 11 providing the formula for each peeling step.}

It is useful to visualize the peeling algorithm using a bipartite graph constructed from the parity check matrix $\bH \in \{0,1\}^{\nBins \times \tn}$  (which was used to define the sketching matrix $\bB$).
 As shown in Fig. \ref{subfig:pruned_graph_1A_t=0}, the $\tk=\binom{k}{2}+k$  left nodes correspond to the unknown nonzero entries in $\bX$ at $t=0$, and the $R$ right nodes correspond to the bins. Each left node connects to $d$ distinct right nodes. The left nodes  connected to each right node (bin) are those that appear in the linear constraints defining the bin. We call Fig. \ref{subfig:pruned_graph_1A_t=0} a `pruned' graph as we have removed the left nodes detected to be zeros (via zeroton bins).  As illustrated in Figs. \ref{subfig:pruned_graph_1A_t=1} and \ref{subfig:pruned_graph_1A_t=2},  the algorithm first recovers and peels off ${X}_{24}$  from the singleton $\by_7$ at $t=1$. This creates a singleton $\by_6$  from which ${X}_{44}$ is recovered and peeled off at $t=2$. The algorithm terminates at $t=2$ as there are no more singletons.

\begin{algorithm}[!t]
\color{black}
\caption{\edit{Peeling decoder in stage A}}\label{alg:nl_stage_A}
\textsc{Input:} sketch vector $\by\in \complex^{2R}$, coding matrix $\bH\in \{0,1\}^{R\times \tn}$.\\
 \textsc{Output:} an estimate  of the vectorized signal matrix $\hat{\bx} \in \reals^{\tn}$. 
\begin{algorithmic}[1]
\State\textbf{Initialize:} $\hat{\bx}\gets\mathbf{0}_{\tn}$, $\mc{R}\gets\{1, 2, \dots,R\}$.
 \Comment{$\mc{R}$: set of right nodes to be checked}
\While{$\mc{R} \neq \emptyset$}
\State $\mc{R}' \gets \mc{R}$
\Comment{$\mc{R}'$: right nodes to be checked in the next  iteration}
\For{each $j\in \mc{R}$} 
\If{$\by_j$ is a zeroton}
\State $\mc{R}' \gets \mc{R}' \setminus j$
\ElsIf{$\by_j$ is a singleton (determined via \eqref{eq:noiseless_singleton_test})}
\State recover the $\hat{\ell}$-th entry of $\hat{\bx}$ as $\hat{x}_{\hat{\ell}}$ according to \eqref{eq:noiseless_singleton_test}
\State $\mc{C} \gets \{i\in [R]:  \ H_{i\hat{\ell}}=1\}$ 
\For{each $i\in \mc{C}$} \Comment{peel off connecting right nodes}
\State $\by_i \gets \by_i - \hat{x}_{\hat{\ell}}\begin{bmatrix}
1\\\exp( 2\pi\i(\hat{\ell}-1)/\tn)
\end{bmatrix}$
\EndFor
\State $\mc{R}' \gets \mc{R}'\cup \mc{C}\setminus j$
\Else
\State $\mc{R}' \gets \mc{R}' \setminus j$ 
\EndIf
\EndFor
\State $\mc{R}\gets \mc{R}'$
\EndWhile\\
\Return $\hbx$
\end{algorithmic}
\end{algorithm}

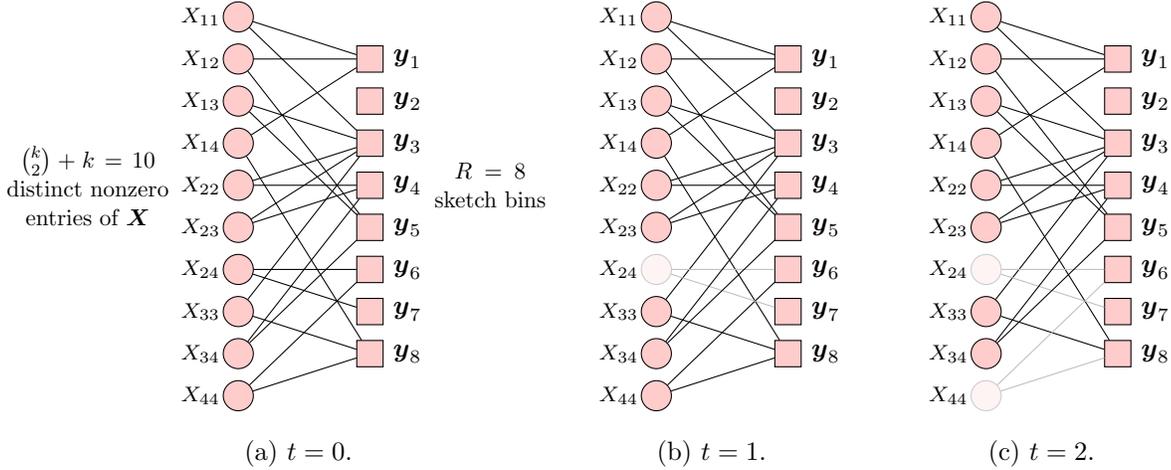
\begin{figure}
\linespread{1}
\begin{subfigure}[B]{0.4\linewidth}
 \centering
 \begin{adjustwidth}{-0.8cm}{}
\begin{tikzpicture}
\node[vnodex](vnodex11){};
\node[left of=vnodex11, node distance=0.5cm, text width=3cm, align=center](vnodex11Label){\scriptsize $X_{11}$};
\node[vnodex, below of= vnodex11, node distance=0.7cm](vnodex12){};
\node[left of=vnodex12, node distance=0.5cm, text width=3cm, align=center](vnodex12Label){\scriptsize $X_{12}$};
\node[vnodex, below of= vnodex11, node distance=0.7*2cm](vnodex13){};
\node[left of=vnodex13, node distance=0.5cm, text width=3cm, align=center](vnodex13Label){\scriptsize $X_{13}$};
\node[vnodex, below of= vnodex11, node distance=0.7*3cm](vnodex14){};
\node[left of=vnodex14, node distance=0.5cm, text width=3cm, align=center](vnodex14Label){\scriptsize $X_{14}$};
\node[vnodex, below of= vnodex11, node distance=0.7*4cm](vnodex22){};
\node[left of=vnodex22, node distance=0.5cm, text width=3cm, align=center](vnodex22Label){\scriptsize $X_{22}$};
\node[vnodex, below of= vnodex11, node distance=0.7*5cm](vnodex23){};
\node[left of=vnodex23, node distance=0.5cm, text width=3cm, align=center](vnodex23Label){\scriptsize $X_{23}$};
\node[vnodex, below of= vnodex11, node distance=0.7*6cm](vnodex24){};
\node[left of=vnodex24, node distance=0.5cm, text width=3cm, align=center](vnodex24Label){\scriptsize $X_{24}$};
\node[vnodex, below of= vnodex11, node distance=0.7*7cm](vnodex33){};
\node[left of=vnodex33, node distance=0.5cm, text width=3cm, align=center](vnodex33Label){\scriptsize $X_{33}$};
\node[vnodex, below of= vnodex11, node distance=0.7*8cm](vnodex34){};
\node[left of=vnodex34, node distance=0.5cm, text width=3cm, align=center](vnodex34Label){\scriptsize $X_{34}$};
\node[vnodex, below of= vnodex11, node distance=0.7*9cm](vnodex44){};
\node[left of=vnodex44, node distance=0.5cm, text width=3cm, align=center](vnodex44Label){\scriptsize $X_{44}$};
\node[cnodex, right of= vnodex12, node distance=2.5cm](cnodey1){};
\node[right of=cnodey1, node distance=0.5cm, text width=3cm, align=center](cnodey1Label){\small $\by_{1}$};
\node[cnodex, right of= vnodex13, node distance=2.5cm](cnodey2){};
\node[right of=cnodey2, node distance=0.5cm, text width=3cm, align=center](cnodey2Label){\small $\by_{2}$};
\node[cnodex, right of= vnodex14, node distance=2.5cm](cnodey3){};
\node[right of=cnodey3, node distance=0.5cm, text width=3cm, align=center](cnodey3Label){\small $\by_{3}$};
\node[cnodex, right of= vnodex22, node distance=2.5cm](cnodey4){};
\node[right of=cnodey4, node distance=0.5cm, text width=3cm, align=center](cnodey4Label){\small $\by_{4}$};
\node[cnodex, right of= vnodex23, node distance=2.5cm](cnodey5){};
\node[right of=cnodey5, node distance=0.5cm, text width=3cm, align=center](cnodey5Label){\small $\by_{5}$};
\node[cnodex, right of= vnodex24, node distance=2.5cm](cnodey6){};
\node[right of=cnodey6, node distance=0.5cm, text width=3cm, align=center](cnodey6Label){\small $\by_{6}$};
\node[cnodex, right of= vnodex33, node distance=2.5cm](cnodey7){};
\node[right of=cnodey7, node distance=0.5cm, text width=3cm, align=center](cnodey7Label){\small $\by_{7}$};
\node[cnodex, right of= vnodex34, node distance=2.5cm](cnodey8){};
\node[right of=cnodey8, node distance=0.5cm, text width=3cm, align=center](cnodey8Label){\small $\by_{8}$};
\draw (vnodex11)--(cnodey1);
\draw (vnodex11)--(cnodey3);
\draw (vnodex22)--(cnodey3);
\draw (vnodex22)--(cnodey4);
\draw (vnodex33)--(cnodey3);
\draw (vnodex33)--(cnodey8);
\draw (vnodex44)--(cnodey6);
\draw (vnodex44)--(cnodey8);
\draw (vnodex12)--(cnodey1);
\draw (vnodex12)--(cnodey5);
\draw (vnodex13)--(cnodey3);
\draw (vnodex13)--(cnodey5);
\draw (vnodex14)--(cnodey1);
\draw (vnodex14)--(cnodey8);
\draw (vnodex23)--(cnodey3);
\draw (vnodex23)--(cnodey4);
\draw (vnodex24)--(cnodey7);
\draw (vnodex24)--(cnodey6);
\draw (vnodex34)--(cnodey5);
\draw (vnodex34)--(cnodey4);
\node[left of=vnodex22, node distance=2cm, text width=3cm, align=center, scale=0.8](numLnodeLabel){$\binom{k}{2}+k=10$ distinct nonzero entries of $\bX$}; 
\node[right of=cnodey4, node distance=1.6cm, text width=2.5cm, align=center, scale=0.8](numRnodeLabel){$\nBins=8$ sketch bins};
\end{tikzpicture}
\end{adjustwidth}
\subcaption{\small $t=0$.}\label{subfig:pruned_graph_1A_t=0}
\end{subfigure}
\begin{subfigure}[B]{0.25\linewidth}
 \centering
 \begin{adjustwidth}{-0.8cm}{}
\begin{tikzpicture}
\node[vnodex](vnodex11){};
\node[left of=vnodex11, node distance=0.5cm, text width=3cm, align=center](vnodex11Label){\scriptsize $X_{11}$};
\node[vnodex, below of= vnodex11, node distance=0.7cm](vnodex12){};
\node[left of=vnodex12, node distance=0.5cm, text width=3cm, align=center](vnodex12Label){\scriptsize $X_{12}$};
\node[vnodex, below of= vnodex11, node distance=0.7*2cm](vnodex13){};
\node[left of=vnodex13, node distance=0.5cm, text width=3cm, align=center](vnodex13Label){\scriptsize $X_{13}$};
\node[vnodex, below of= vnodex11, node distance=0.7*3cm](vnodex14){};
\node[left of=vnodex14, node distance=0.5cm, text width=3cm, align=center](vnodex14Label){\scriptsize $X_{14}$};
\node[vnodex, below of= vnodex11, node distance=0.7*4cm](vnodex22){};
\node[left of=vnodex22, node distance=0.5cm, text width=3cm, align=center](vnodex22Label){\scriptsize $X_{22}$};
\node[vnodex, below of= vnodex11, node distance=0.7*5cm](vnodex23){};
\node[left of=vnodex23, node distance=0.5cm, text width=3cm, align=center](vnodex23Label){\scriptsize $X_{23}$};
\node[vnodexFaded, below of= vnodex11, node distance=0.7*6cm](vnodex24){};
\node[left of=vnodex24, node distance=0.5cm, text width=3cm, align=center](vnodex24Label){\scriptsize $X_{24}$};
\node[vnodex, below of= vnodex11, node distance=0.7*7cm](vnodex33){};
\node[left of=vnodex33, node distance=0.5cm, text width=3cm, align=center](vnodex33Label){\scriptsize $X_{33}$};
\node[vnodex, below of= vnodex11, node distance=0.7*8cm](vnodex34){};
\node[left of=vnodex34, node distance=0.5cm, text width=3cm, align=center](vnodex34Label){\scriptsize $X_{34}$};
\node[vnodex, below of= vnodex11, node distance=0.7*9cm](vnodex44){};
\node[left of=vnodex44, node distance=0.5cm, text width=3cm, align=center](vnodex44Label){\scriptsize $X_{44}$};
\node[cnodex, right of= vnodex12, node distance=2.5cm](cnodey1){};
\node[right of=cnodey1, node distance=0.5cm, text width=3cm, align=center](cnodey1Label){\small $\by_{1}$};
\node[cnodex, right of= vnodex13, node distance=2.5cm](cnodey2){};
\node[right of=cnodey2, node distance=0.5cm, text width=3cm, align=center](cnodey2Label){\small $\by_{2}$};
\node[cnodex, right of= vnodex14, node distance=2.5cm](cnodey3){};
\node[right of=cnodey3, node distance=0.5cm, text width=3cm, align=center](cnodey3Label){\small $\by_{3}$};
\node[cnodex, right of= vnodex22, node distance=2.5cm](cnodey4){};
\node[right of=cnodey4, node distance=0.5cm, text width=3cm, align=center](cnodey4Label){\small $\by_{4}$};
\node[cnodex, right of= vnodex23, node distance=2.5cm](cnodey5){};
\node[right of=cnodey5, node distance=0.5cm, text width=3cm, align=center](cnodey5Label){\small $\by_{5}$};
\node[cnodex, right of= vnodex24, node distance=2.5cm](cnodey6){};
\node[right of=cnodey6, node distance=0.5cm, text width=3cm, align=center](cnodey6Label){\small $\by_{6}$};
\node[cnodex, right of= vnodex33, node distance=2.5cm](cnodey7){};
\node[right of=cnodey7, node distance=0.5cm, text width=3cm, align=center](cnodey7Label){\small $\by_{7}$};
\node[cnodex, right of= vnodex34, node distance=2.5cm](cnodey8){};
\node[right of=cnodey8, node distance=0.5cm, text width=3cm, align=center](cnodey8Label){\small $\by_{8}$};
\draw (vnodex11)--(cnodey1);
\draw (vnodex11)--(cnodey3);
\draw (vnodex22)--(cnodey3);
\draw (vnodex22)--(cnodey4);
\draw (vnodex33)--(cnodey3);
\draw (vnodex33)--(cnodey8);
\draw (vnodex44)--(cnodey6);
\draw (vnodex44)--(cnodey8);
\draw (vnodex12)--(cnodey1);
\draw (vnodex12)--(cnodey5);
\draw (vnodex13)--(cnodey3);
\draw (vnodex13)--(cnodey5);
\draw (vnodex14)--(cnodey1);
\draw (vnodex14)--(cnodey8);
\draw (vnodex23)--(cnodey3);
\draw (vnodex23)--(cnodey4);
\draw [gray!50](vnodex24)--(cnodey7);
\draw [gray!50](vnodex24)--(cnodey6);
\draw (vnodex34)--(cnodey5);
\draw (vnodex34)--(cnodey4);
\node[above of=vnodex11, node distance=0.9cm, text width=3cm, align=center, scale=0.8, xshift=-0.2cm, text=white](numLnodeLabel){}; 
\end{tikzpicture}
\end{adjustwidth}
\subcaption{\small $t=1$.}\label{subfig:pruned_graph_1A_t=1}
\end{subfigure}
\hspace{0cm}
\begin{subfigure}[B]{0.25\linewidth}
 \centering
 \begin{adjustwidth}{-0.8cm}{}
\begin{tikzpicture}
\node[vnodex](vnodex11){};
\node[left of=vnodex11, node distance=0.5cm, text width=3cm, align=center](vnodex11Label){\scriptsize $X_{11}$};
\node[vnodex, below of= vnodex11, node distance=0.7cm](vnodex12){};
\node[left of=vnodex12, node distance=0.5cm, text width=3cm, align=center](vnodex12Label){\scriptsize $X_{12}$};
\node[vnodex, below of= vnodex11, node distance=0.7*2cm](vnodex13){};
\node[left of=vnodex13, node distance=0.5cm, text width=3cm, align=center](vnodex13Label){\scriptsize $X_{13}$};
\node[vnodex, below of= vnodex11, node distance=0.7*3cm](vnodex14){};
\node[left of=vnodex14, node distance=0.5cm, text width=3cm, align=center](vnodex14Label){\scriptsize $X_{14}$};
\node[vnodex, below of= vnodex11, node distance=0.7*4cm](vnodex22){};
\node[left of=vnodex22, node distance=0.5cm, text width=3cm, align=center](vnodex22Label){\scriptsize $X_{22}$};
\node[vnodex, below of= vnodex11, node distance=0.7*5cm](vnodex23){};
\node[left of=vnodex23, node distance=0.5cm, text width=3cm, align=center](vnodex23Label){\scriptsize $X_{23}$};
\node[vnodexFaded, below of= vnodex11, node distance=0.7*6cm](vnodex24){};
\node[left of=vnodex24, node distance=0.5cm, text width=3cm, align=center](vnodex24Label){\scriptsize $X_{24}$};
\node[vnodex, below of= vnodex11, node distance=0.7*7cm](vnodex33){};
\node[left of=vnodex33, node distance=0.5cm, text width=3cm, align=center](vnodex33Label){\scriptsize $X_{33}$};
\node[vnodex, below of= vnodex11, node distance=0.7*8cm](vnodex34){};
\node[left of=vnodex34, node distance=0.5cm, text width=3cm, align=center](vnodex34Label){\scriptsize $X_{34}$};
\node[vnodexFaded, below of= vnodex11, node distance=0.7*9cm](vnodex44){};
\node[left of=vnodex44, node distance=0.5cm, text width=3cm, align=center](vnodex44Label){\scriptsize $X_{44}$};
\node[cnodex, right of= vnodex12, node distance=2.5cm](cnodey1){};
\node[right of=cnodey1, node distance=0.5cm, text width=3cm, align=center](cnodey1Label){\small $\by_{1}$};
\node[cnodex, right of= vnodex13, node distance=2.5cm](cnodey2){};
\node[right of=cnodey2, node distance=0.5cm, text width=3cm, align=center](cnodey2Label){\small $\by_{2}$};
\node[cnodex, right of= vnodex14, node distance=2.5cm](cnodey3){};
\node[right of=cnodey3, node distance=0.5cm, text width=3cm, align=center](cnodey3Label){\small $\by_{3}$};
\node[cnodex, right of= vnodex22, node distance=2.5cm](cnodey4){};
\node[right of=cnodey4, node distance=0.5cm, text width=3cm, align=center](cnodey4Label){\small $\by_{4}$};
\node[cnodex, right of= vnodex23, node distance=2.5cm](cnodey5){};
\node[right of=cnodey5, node distance=0.5cm, text width=3cm, align=center](cnodey5Label){\small $\by_{5}$};
\node[cnodex, right of= vnodex24, node distance=2.5cm](cnodey6){};
\node[right of=cnodey6, node distance=0.5cm, text width=3cm, align=center](cnodey6Label){\small $\by_{6}$};
\node[cnodex, right of= vnodex33, node distance=2.5cm](cnodey7){};
\node[right of=cnodey7, node distance=0.5cm, text width=3cm, align=center](cnodey7Label){\small $\by_{7}$};
\node[cnodex, right of= vnodex34, node distance=2.5cm](cnodey8){};
\node[right of=cnodey8, node distance=0.5cm, text width=3cm, align=center](cnodey8Label){\small $\by_{8}$};
\draw (vnodex11)--(cnodey1);
\draw (vnodex11)--(cnodey3);
\draw (vnodex22)--(cnodey3);
\draw (vnodex22)--(cnodey4);
\draw (vnodex33)--(cnodey3);
\draw (vnodex33)--(cnodey8);
\draw [gray!50](vnodex44)--(cnodey6);
\draw [gray!50](vnodex44)--(cnodey8);
\draw (vnodex12)--(cnodey1);
\draw (vnodex12)--(cnodey5);
\draw (vnodex13)--(cnodey3);
\draw (vnodex13)--(cnodey5);
\draw (vnodex14)--(cnodey1);
\draw (vnodex14)--(cnodey8);
\draw (vnodex23)--(cnodey3);
\draw (vnodex23)--(cnodey4);
\draw [gray!50](vnodex24)--(cnodey7);
\draw [gray!50](vnodex24)--(cnodey6);
\draw (vnodex34)--(cnodey5);
\draw (vnodex34)--(cnodey4);

\node[above of=vnodex11, node distance=0.9cm, text width=3cm, align=center, scale=0.8, xshift=-0.2cm, text=white](numLnodeLabel){}; 
\end{tikzpicture}
\end{adjustwidth}
\subcaption{\small $t=2$.}\label{subfig:pruned_graph_1A_t=2}
\end{subfigure}
\caption{\small (a): Pruned bipartite graph for stage A corresponding to $\bH \in \{0,1\}^{\nBins\times \tn}$. Here, $\bX = \lambda\bv\bv^T$ with  $v_i\neq 0$ for $i\in[4]$ and $v_i=0$ for $5\le i \le n$, and $\nBins=8$. (b)--(c):
The graph process that models the recovery of the nonzero entries of $\bX$. The faded nodes and edges are those that have been peeled off.}
  \label{fig:graph_process_1A}
\end{figure}

\subsubsection{Stage B of the algorithm} \label{subsubsec:stageB}
  Recalling that $\bX=\lambda \bv\bv^T$, we write $ \tbv: =\sqrt{\abs{\lambda}} \bv$ for the unsigned, unnormalized eigenvector. Stage B of the algorithm uses the rank-1 structure of $\bX$ to recover the nonzeros in  $\tbv$  from the nonzero entries of $\bX$ that were recovered in stage A. 
 When fully recovered, $\tbv$ is normalized to give $\bv$. Since the $j$-th diagonal entry of $\bX$ is $\lambda v_j^2$, the sign of $\lambda$ can be determined from any nonzero diagonal entry recovered in stage A.  We now describe stage B  for  $\lambda>0$,  and then explain the small adjustment needed for $\lambda <0$.

Since $\bX$ is rank-1, its  
  $\binom{k}{2}$ nonzero
   above-diagonal entries are pairwise products of the form  $\tv_i\tv_j$. The proof of Theorem  \ref{thm:main_result_symm} shows that  in stage A, with high probability, at least a fraction $k^{-\delta}$ of these pairwise products
 are recovered, for a constant $\deltaRangeInLine$.  (See Lemma \ref{lem:1A_rank_r}). 
  Stage B of the algorithm can be visualized using another bipartite graph shown in Fig. \ref{subfig:graph_1B_t=-1}. The  $k$ left nodes represent the unknown nonzero entries in $\tbv$, and the right nodes represent the  nonzero pairwise products in $\bX$  recovered in stage A. The right nodes are the product constraints that the left nodes satisfy. This is a pruned graph because any zero entries  in $\tbv$ or $\bX$ have been excluded from the graph.

At $t=0$,   the algorithm picks one of the  nonzero diagonal entries recovered in stage A, say $X_{jj}$, and declares  $\sqrt{X_{jj}}$ as the value of the left node $\tv_j$. The algorithm then removes (peels off) the contribution of $\tv_j$ from the right nodes that $\tv_j$ connects to by dividing these right nodes  by $\sqrt{X_{jj}}$.  Fig. \ref{subfig:graph_1B_t=0} shows this initial step with $\tv_4$ being a  left node whose squared value was recovered in stage A and  is  peeled off at $t=0$. Once the first left node is peeled off, all of its connecting right nodes reduce from degree-2 nodes to degree-1. For $t\ge 1$, the algorithm proceeds similarly to stage A, noting that in stage B, each right node represents the product of its two connecting left nodes rather than a linear combination. At each step, a degree-1 right node is picked uniformly at random from the available ones and its connecting left node recovered and peeled off to update the set of degree-1 right nodes.   The algorithm continues until there are no more degree-1 right nodes. All the unrecovered entries of $\tbv$ are set to zero.  In Fig. \ref{subfig:graph_1B_t=1} and \ref{subfig:graph_1B_t=2}, the left nodes $\tv_1$ and $\tv_5$ are recovered via  $\tv_1\tv_4$ and $\tv_4\tv_5$, respectively, and peeled off. Moreover, $\tv_2$ and $\tv_3$ can be recovered via $\tv_2\tv_5$ and $\tv_2\tv_3$ before degree-1 right nodes run out. \edit{The pseudocode for stage B  is provided in Algorithm \ref{alg:nl_stage_B} below, which calls the subroutine recover\_eigenvec defined in Algorithm \ref{alg:nl_stage_B_one_eigenvec}.}

\paragraph{The $\lambda <0$ case}  We flip the sign of each right node (pairwise product)  before running stage B of the  algorithm as above to recover $\tbv = \sqrt{\abs{\lambda}}\bv$. \edit{See lines 4--9 of Algorithm \ref{alg:nl_stage_B_one_eigenvec} for the pseudocode of this operation.}

 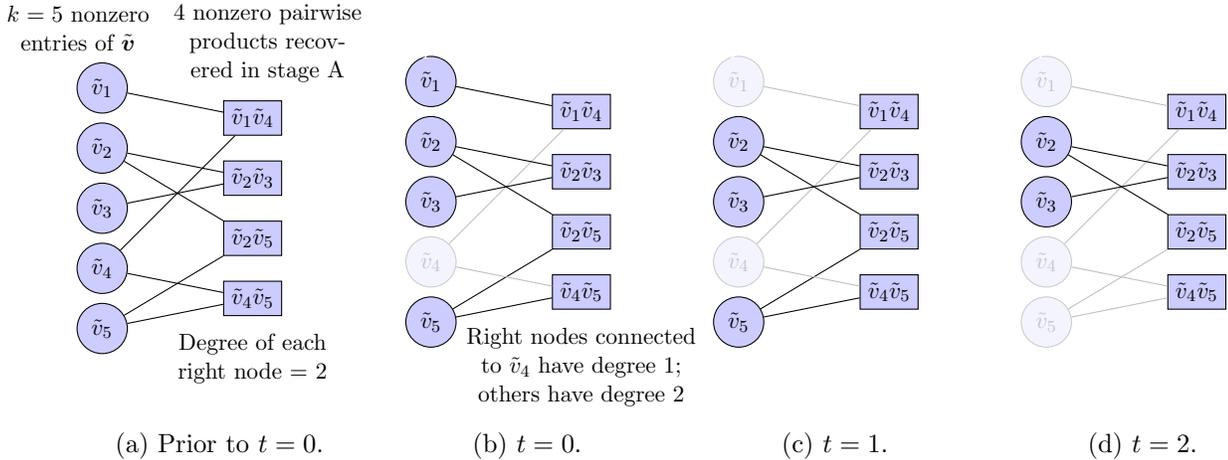
\begin{figure}[!t]
\captionsetup[subfigure]{justification=centering}
\centering
\linespread{1}
 \begin{subfigure}[b]{0.24\linewidth} 
 \centering 
 \begin{adjustwidth}{-1cm}{0cm}
\begin{tikzpicture}
\node[vnode](vnodev1){$\tv_1$};
\node[vnode, below of=vnodev1, node distance=1cm](vnodev2){$\tv_2$};
\node[vnode, below of=vnodev1, node distance=2cm](vnodev3){$\tv_3$};
\node[vnode, below of=vnodev1, node distance=3cm] (vnodev4){$\tv_4$};
\node[vnode, below of=vnodev1, node distance=4cm](vnodev5){$\tv_5$};

\node[cnode, right of= vnodev1, yshift = -0.5cm, node distance=2.5cm](cnodev1v4){$\tv_1\tv_4$};
\node[cnode, right of= vnodev3, node distance=2.5cm, yshift=0.5cm](cnodev2v3){$\tv_2\tv_3$};
\node[cnode, right of= vnodev3, node distance=2.5cm, yshift=-0.5cm](cnodev2v5){$\tv_2\tv_5$};
\node[cnode, right of= vnodev5, yshift = +0.5cm, node distance=2.5cm](cnodev4v5){$\tv_4\tv_5$};

\draw (vnodev1)--(cnodev1v4);
\draw (vnodev4)--(cnodev1v4);
\draw (vnodev2)--(cnodev2v3);
\draw (vnodev3)--(cnodev2v3);
\draw  (vnodev2)--(cnodev2v5);
\draw (vnodev5)--(cnodev2v5);
\draw  (vnodev4)--(cnodev4v5);
\draw (vnodev5)--(cnodev4v5);

\node[above of=vnodev1, node distance=0.8cm, text width=2.5cm, align=center, scale=0.8, xshift=-0.4cm](numLnodeLabel){$k=5$ nonzero entries of $\tbv$};
\node[above of=cnodev1v4, node distance=1cm, xshift =0.2cm, text width=4cm, align=center, scale=0.8](numRnodeLabel){$4$ nonzero pairwise products recovered in stage A};
\node[below of=cnodev4v5, node distance=0.8cm, text width=3cm, align=center, scale=0.8](rightDist1){Degree of each right node = 2};
\end{tikzpicture}
\end{adjustwidth}
\vspace{0.2cm}
\caption{Prior to $t=0$.}
\label{subfig:graph_1B_t=-1}
\end{subfigure}
\begin{subfigure}[b]{0.24\linewidth} 
\centering
 \begin{tikzpicture}
\node[vnode](vnodev1){$\tv_1$};
\node[vnode, below of=vnodev1, node distance=1cm](vnodev2){$\tv_2$};
\node[vnode, below of=vnodev1, node distance=2cm](vnodev3){$\tv_3$};
\node[vnodeFaded,  below of=vnodev1, node distance=3cm] (vnodev4){$\tv_4$};
\node[vnode, below of=vnodev1, node distance=4cm](vnodev5){$\tv_5$};
\node[cnode, right of= vnodev1, yshift = -0.5cm, node distance=2.5cm](cnodev1v4){$\tv_1\tv_4$};
\node[cnode, right of= vnodev3, node distance=2.5cm, yshift=0.5cm](cnodev2v3){$\tv_2\tv_3$};
\node[cnode, right of= vnodev3, node distance=2.5cm, yshift=-0.5cm](cnodev2v5){$\tv_2\tv_5$};
\node[cnode, right of= vnodev5, yshift = +0.5cm, node distance=2.5cm](cnodev4v5){$\tv_4\tv_5$};
\draw (vnodev1)--(cnodev1v4);
\draw [gray!50](vnodev4)--(cnodev1v4);
\draw (vnodev2)--(cnodev2v3);
\draw (vnodev3)--(cnodev2v3);
\draw  (vnodev2)--(cnodev2v5);
\draw (vnodev5)--(cnodev2v5);
\draw [gray!50] (vnodev4)--(cnodev4v5);
\draw  (vnodev5)--(cnodev4v5);
\node[above of=vnodev1, node distance=0.8cm, text width=1.5cm, align=center, scale=0.8, text=white](numLnodeLabel){$k=5$ supports of $\tbv$}; 
\node[below of=cnodev4v5, node distance=1cm, text width=4cm, align=center, scale=0.8](rightDist1){Right nodes connected to $\tv_4$ have degree 1; others have degree 2};
\end{tikzpicture}
\caption{$t=0$.}
\label{subfig:graph_1B_t=0}
\end{subfigure}
\begin{subfigure}[b]{0.24\linewidth} 
\centering
 \begin{tikzpicture}
\node[vnodeFaded](vnodev1){$\tv_1$};
\node[vnode, below of=vnodev1, node distance=1cm](vnodev2){$\tv_2$};
\node[vnode, below of=vnodev1, node distance=2cm](vnodev3){$\tv_3$};
\node[vnodeFaded,  below of=vnodev1, node distance=3cm] (vnodev4){$\tv_4$};
\node[vnode, below of=vnodev1, node distance=4cm](vnodev5){$\tv_5$};
\node[cnode, right of= vnodev1, yshift = -0.5cm, node distance=2.5cm](cnodev1v4){$\tv_1\tv_4$};
\node[cnode, right of= vnodev3, node distance=2.5cm, yshift=0.5cm](cnodev2v3){$\tv_2\tv_3$};
\node[cnode, right of= vnodev3, node distance=2.5cm, yshift=-0.5cm](cnodev2v5){$\tv_2\tv_5$};
\node[cnode, right of= vnodev5, yshift = +0.5cm, node distance=2.5cm](cnodev4v5){$\tv_4\tv_5$};
\draw [gray!50](vnodev1)--(cnodev1v4);
\draw [gray!50](vnodev4)--(cnodev1v4);
\draw (vnodev2)--(cnodev2v3);
\draw (vnodev3)--(cnodev2v3);
\draw  (vnodev2)--(cnodev2v5);
\draw (vnodev5)--(cnodev2v5);
\draw [gray!50] (vnodev4)--(cnodev4v5);
\draw  (vnodev5)--(cnodev4v5);
\node[above of=vnodev1, node distance=0.8cm, text width=1.5cm, align=center, scale=0.8, text=white](numLnodeLabel){$k=5$ supports of $\tbv$}; 
\node[below of=cnodev4v5, xshift = 0.2cm, node distance=1cm, text width=4cm, align=center, scale=0.8, text=white](rightDist1){Right nodes connected to $\tv_4$ have degree 1; others have degree 2};
\end{tikzpicture}
\caption{$t=1$.}
\label{subfig:graph_1B_t=1}
\end{subfigure}
\begin{subfigure}[b]{0.24\linewidth} 
\centering
 \begin{tikzpicture}
\node[vnodeFaded](vnodev1){$\tv_1$};
\node[vnode, below of=vnodev1, node distance=1cm](vnodev2){$\tv_2$};
\node[vnode, below of=vnodev1, node distance=2cm](vnodev3){$\tv_3$};
\node[vnodeFaded,  below of=vnodev1, node distance=3cm] (vnodev4){$\tv_4$};
\node[vnodeFaded, below of=vnodev1, node distance=4cm](vnodev5){$\tv_5$};

\node[cnode, right of= vnodev1, yshift = -0.5cm, node distance=2.5cm](cnodev1v4){$\tv_1\tv_4$};
\node[cnode, right of= vnodev3, node distance=2.5cm, yshift=0.5cm](cnodev2v3){$\tv_2\tv_3$};
\node[cnode, right of= vnodev3, node distance=2.5cm, yshift=-0.5cm](cnodev2v5){$\tv_2\tv_5$};
\node[cnode, right of= vnodev5, yshift = +0.5cm, node distance=2.5cm](cnodev4v5){$\tv_4\tv_5$};
\draw [gray!50](vnodev1)--(cnodev1v4);
\draw [gray!50](vnodev4)--(cnodev1v4);
\draw (vnodev2)--(cnodev2v3);
\draw (vnodev3)--(cnodev2v3);
\draw (vnodev2)--(cnodev2v5);
\draw [gray!50] (vnodev5)--(cnodev2v5);
\draw [gray!50] (vnodev4)--(cnodev4v5);
\draw [gray!50](vnodev5)--(cnodev4v5);
\node[above of=vnodev1, node distance=0.8cm, text width=1.5cm, align=center, scale=0.8, text=white](numLnodeLabel){$k=5$ supports of $\tbv$}; 
\node[below of=cnodev4v5, xshift = 0.2cm, node distance=1cm, text width=4cm, align=center, scale=0.8, text=white](rightDist1){Right nodes connected to $\tv_4$ have degree 1; others have degree 2};
\end{tikzpicture}
\caption{$t=2$.}
\label{subfig:graph_1B_t=2}
\end{subfigure}
\caption{\small (a): Pruned bipartite graph for stage B with $\bX=\tbv\tbv^T$, where $\tv_i\neq 0$ for $i\in [5]$ and $\tv_i=0$ for $6\le i\le n$. (b)--(d): The graph process that models the recovery of the nonzero entries in $\tbv$ from the nonzero matrix entries recovered in stage A. 
The faded nodes and edges are those that have been peeled off. At $t=0$, the left node $\tv_4$ is peeled off based on the  recovered nonzero diagonal entry $X_{44}$.}
\label{fig:actual_graph_process_1B}
\end{figure}

\setlength{\textfloatsep}{2pt}
\begin{algorithm}[!b]
\color{black}
\caption{\edit{Peeling decoder in stage B (for rank-$r$ symmetric matrices) }}\label{alg:nl_stage_B}
\textsc{Input:} partially recovered signal matrix $\hat{\bX}\in \reals^{n\times n}$, the rank $r$.\\
 \textsc{Output:} an estimate of the unnormalized eigenvectors $\tbv_1, \tbv_2, \dots, \tbv_r \in \reals^n$. 
\begin{algorithmic}[1]
\Require the true signal matrix $\bX$ is symmetric rank-$r$ with eigenvectors having disjoint supports.
\State \textbf{Initialize:} $\tbv_1, \tbv_2, \dots, \tbv_r \gets \bzero_n$
\State $\mc{S} \gets \{i\in [n]: \hX_{ij} \neq 0 \; \text{for any} \; j\in [n]\}$, $\hk \gets \abs{\mc{S}}$
\State  $\hbX^s \gets \hbX_{\mc{S},\mc{S}}$ 
\Comment{extract nonzero submatrix to process}
\For{$i=1$ \text{to} $r$}
\State $\tbv^s_i, \hbX^s \gets $recover\_eigenvec($\hbX^s$)
\Comment{recover one eigenvector and update $\hbX^s$}
\State $\tbv_{i\mc{S}} \gets \tbv^s_i$
\EndFor\\
\Return $\tbv_1, \tbv_2, \dots, \tbv_r$
\end{algorithmic}
\end{algorithm}

\begin{algorithm}[!b]
\color{black}
\caption{\edit{recover\_eigenvec: Recovery of one eigenvector  for a symmetric matrix}}\label{alg:nl_stage_B_one_eigenvec}
\textsc{Input:} partially recovered nonzero submatrix $\hat{\bX^s}\in \reals^{\hk\times \hk}$ (unrecovered entries stored as zeros).\\
 \textsc{Output:} an estimate of one unnormalized eigenvector $\tbv^s\in \reals^{\hk}$,\\
 \text{\qquad\quad\quad\;}$\hbX^s$ with its entries used to infer $\tbv^s$ reset to zeros.
\begin{algorithmic}[1]
\Require the true signal matrix $\bX$ is symmetric rank-$r$ with eigenvectors having disjoint supports.
\State \textbf{Initialize:}  $\tbv^s \gets \bzero_{\hk}$
\If{$\{ i \in [\hk]: \hX^s_{ii}\neq 0\}\neq \emptyset$}
\State
$\ell \gets $ an index chosen randomly from $\{ i \in [\hk]: \hX^s_{ii}\neq 0\}$
\If{$\hX^s_{\ell\ell} >0$}  
\State flip\_sign $\gets$ False
\Else
 \State flip\_sign $\gets$ True
 \State $\hbX^s\gets - \hbX^s$ 
\EndIf
 \State
 $\tbv^s \gets \tbv^s + \begin{bmatrix}\hat{X}^s_{1\ell}, \hat{X}^s_{2\ell}, \dots, \hat{X}^s_{\hk\ell}\end{bmatrix}^T / \sqrt{\hat{X}^s_{\ell\ell}}$
 \Comment{peel off $\ell$-th left node}
\State $\mc{I}\gets \{i\in [\hk]: \tv_i^s\neq 0\}$
\State $\hX_{ij}^s \gets 0 \quad \forall \quad i,j \in \mc{I}$ 
\State $\mc{L}\gets \mc{I}\setminus \ell$ \Comment{$\mc{L}$: set of left nodes connecting to singletons}
\While{$\mc{L} \neq \emptyset$ \textbf{and} $\abs{\mc{I}} <\hk$}
\State $\ell \gets $ an index chosen randomly from $\mc{L}$
\State $\tbv^s \gets \tbv^s + \begin{bmatrix}\hat{X}^s_{1\ell}, \hat{X}^s_{2\ell}, \dots, \hat{X}^s_{\hk\ell}\end{bmatrix}^T / \tv_\ell^s$
 \Comment{peel off $\ell$-th left node}
\State $\mc{L}\gets (\mc{L}\setminus \ell )\cup \{i\in [\hk]: \hat{X}_{i\ell}^s\neq 0\}$
\State $\mc{I}\gets \{i\in [\hk]: \tv_i^s\neq 0\} $
\State $\hX_{ij}^s \gets 0 \quad \forall \quad i,j \in \mc{I}$
\EndWhile
\If{flip\_sign} $\hbX^s\gets - \hbX^s$ 
\EndIf
\EndIf\\
\Return $\tbv^s$, $\hbX^s$
\end{algorithmic}
\end{algorithm}

\subsection{Recovery algorithm for rank-$r$ symmetric matrices} \label{subsec:recovery_alg_rankr}
\paragraph{Disjoint suppports} When the eigenvectors $\{\bv_i\}$
 are known to have disjoint supports, the same two-stage 
  algorithm can be used. The only difference is that the graphical representation for stage B (Fig. \ref{subfig:graph_1B_t=-1}) will now consist of $r$ disjoint bipartite graphs. For each of these graphs, stage B of the algorithm is initialized using a nonzero diagonal entry recovered in stage A.  \edit{See Algorithm  \ref{alg:nl_stage_B} below for details.}

\paragraph{General case} When the supports of the $r$ eigenvectors have overlaps, we cannot use stage B of the algorithm as it depends on the entries of $\bX$ being pairwise products of the eigenvector entries. We therefore take a large enough sketch  to identify all the nonzero entries in $\bX$ in stage A. This is equivalent to  compressed sensing recovery of the complete vectorized matrix using the scheme of \cite{li2019sublinear, bakshi2016SHOFA}. We then perform an eigendecomposition on the recovered submatrix of nonzero entries to obtain the nonzero entries of $ \{ \bv_1, \ldots, \bv_r \}$. 
This  submatrix has size at most $rk \times rk$, hence the complexity of its eigendecomposition is   $\bigo((rk)^3)$, which does not depend on the ambient dimension $n$.  Recovering the entire matrix $\bX$ in stage A  requires a sample complexity  $\nMeas =\bigo(rk^2)$, which is  larger by a factor of $\ln k$ 
than that needed in the case of disjoint supports. 

 \subsection{Recovery of non-symmetric matrices} \label{subsec:nonsymm}
 Consider the non-symmetric case where $\bX$ has the form 
$\bX = \sum_{i=1}^r \sigma_i \bu_i \bv_i^{T}$ where $\{\bu_i, \bv_i\}$ each has at most $k$ nonzero entries. The sketching operator is the same as in the symmetric case, except that here it acts on the entire matrix $\bX$. Consider first the case where the singular vectors $\{ \bu_i\}$  have disjoint supports as do $\{ \bv_i\}$.  Based on the sketch, stage A of the algorithm can be used to  recover a small fraction of the nonzero entries in $\bX$ like in the symmetric case. Using the recovered nonzero matrix entries, stage B of the algorithm  iteratively recovers the nonzeros in $\{ \bu_i, \bv_i\}$. While  stage B can still be viewed as  a peeling decoder, the associated bipartite graph  has a  different structure from the symmetric case. The initialization step is also different.

We take the rank-1 case as an example, where $\bX = \sigma \bu \bv^T$. 
 The peeling algorithm recovers $\tbu$ and $\tbv$, which are scaled versions of $\bu$ and $\bv$ such that $\tbu\tbv^T = \bX $; these vectors are then normalized to obtain $\sigma, \bu$ and $\bv$. Fig. \ref{subfig:graph_1B_non_sym_t=-1} illustrates the pruned stage B graph in this setting. The left nodes represent the nonzero entries in $\tbu$ and $\tbv$, and the right nodes represent the nonzero pairwise products recovered in stage A. Each right node connects to one nonzero in $\tbu$ and one nonzero in $\tbv$ and  equals the product of  these two nonzeros.

\setlength{\textfloatsep}{15pt}
\begin{figure}[!b]
\captionsetup[subfigure]{justification=centering}
\centering
\linespread{1}
 \begin{subfigure}[B]{0.24\linewidth} 
 \centering 
 \begin{adjustwidth}{-1cm}{0cm}
\begin{tikzpicture}
\node[vnode](vnodeu1){$\tu_{1}$};
\node[vnode, below of=vnodeu1, node distance=1cm](vnodeu2){$\tu_{2}$};
\node[vnode, below of=vnodeu1, node distance=2cm](vnodeu3){$\tu_3$};
\node[vnode, below of=vnodeu1, node distance=3cm](vnodeu4){$\tu_4$};

\node[vnode, below of =vnodeu4, node distance=1.5cm](vnodev1){$\tv_1$};
\node[vnode, below of=vnodev1, node distance=1cm](vnodev2){$\tv_2$};
\node[vnode, below of=vnodev1, node distance=2cm](vnodev3){$\tv_3$};

\node[cnode, right of= vnodeu1, yshift=-1cm, node distance=2.5cm](cnodeu1v1){$\tu_1\tv_1$};
\node[cnode, below of= cnodeu1v1, node distance=1cm](cnodeu4v1){$\tu_4 \tv_1$};
\node[cnode, below of= cnodeu1v1, node distance=2cm](cnodeu3v2){$\tu_3 \tv_2$};
\node[cnode, below of= cnodeu1v1, node distance=3cm](cnodeu1v3){$\tu_1\tv_3$};
\node[cnode, below of= cnodeu1v1, node distance=4cm](cnodeu2v2){$\tu_2\tv_2$};

\draw (vnodeu1)--(cnodeu1v1);
\draw (vnodev1)--(cnodeu1v1);
\draw (vnodeu1)--(cnodeu1v3);
\draw (vnodev3)--(cnodeu1v3);
\draw  (vnodeu3)--(cnodeu3v2);
\draw (vnodev2)--(cnodeu3v2);
\draw  (vnodeu4)--(cnodeu4v1);
\draw (vnodev1)--(cnodeu4v1);
\draw  (vnodeu2)--(cnodeu2v2);
\draw (vnodev2)--(cnodeu2v2);

\node[left of=vnodeu2, node distance=1.2cm, text width=1.8cm, align=center, scale=0.8, yshift=-0.5cm](numLnodeLabel){$4$ nonzero entries of $\tbu$}; 
\node[left of=vnodev2, node distance=1.2cm, text width=1.8cm, align=center, scale=0.8](numLnodeLabel){$3$ nonzero entries of $\tbv$}; 
\node[above of=cnodeu1v1, node distance=1cm, xshift =0cm, text width=3.5cm, align=center, scale=0.8](numRnodeLabel){$5$ nonzero pairwise products recovered in stage A};
\node[below of=cnodeu1v3, node distance=2cm, text width=3cm, align=center, scale=0.8](rightDist1){Degree of each right node = 2};
\end{tikzpicture}
\end{adjustwidth}
\caption{Prior to $t=0$.}
\label{subfig:graph_1B_non_sym_t=-1}
\end{subfigure}
\begin{subfigure}[B]{0.24\linewidth} 
\centering 
\begin{adjustwidth}{0.8cm}{0cm}
\begin{tikzpicture}
\node[vnodeFaded](vnodeu1){$\tu_1$};
\node[vnode, below of=vnodeu1, node distance=1cm](vnodeu2){$\tu_2$};
\node[vnode, below of=vnodeu1, node distance=2cm](vnodeu3){$\tu_3$};
\node[vnode, below of=vnodeu1, node distance=3cm](vnodeu4){$\tu_4$};

\node[vnode, below of =vnodeu4, node distance=1.5cm](vnodev1){$\tv_1$};
\node[vnode, below of=vnodev1, node distance=1cm](vnodev2){$\tv_2$};
\node[vnode, below of=vnodev1, node distance=2cm](vnodev3){$\tv_3$};

\node[cnode, right of= vnodeu1, yshift=-1cm, node distance=2.5cm](cnodeu1v1){$\tu_1 \tv_1$};
\node[cnode, below of= cnodeu1v1, node distance=1cm](cnodeu4v1){$\tu_4 \tv_1$};
\node[cnode, below of= cnodeu1v1, node distance=2cm](cnodeu3v2){$\tu_3 \tv_2$};
\node[cnode, below of= cnodeu1v1, node distance=3cm](cnodeu1v3){$\tu_1 \tv_3$};
\node[cnode, below of= cnodeu1v1, node distance=4cm](cnodeu2v2){$\tu_2\tv_2$};

\draw [gray!50] (vnodeu1)--(cnodeu1v1);
\draw (vnodev1)--(cnodeu1v1);
\draw [gray!50] (vnodeu1)--(cnodeu1v3);
\draw (vnodev3)--(cnodeu1v3);
\draw  (vnodeu3)--(cnodeu3v2);
\draw (vnodev2)--(cnodeu3v2);
\draw  (vnodeu4)--(cnodeu4v1);
\draw (vnodev1)--(cnodeu4v1);
\draw  (vnodeu2)--(cnodeu2v2);
\draw (vnodev2)--(cnodeu2v2);

\node[above of=cnodeu1v1, node distance=0.8cm, xshift =0.2cm, text width=4cm, align=center, scale=0.8, text=white](numRnodeLabel){};
\node[below of=cnodeu1v3, node distance=2cm, xshift =0.2cm, text width=3.9cm, align=center, scale=0.8](rightDist1){Right nodes connected to $\tu_1$ have degree 1; others have degree 2};
\end{tikzpicture}
\end{adjustwidth}
\vspace{-0.1cm}
\caption{$t=0$.}
\label{subfig:graph_1B_non_sym_t=0}
\end{subfigure}
\begin{subfigure}[B]{0.24\linewidth} 
\centering 
\begin{adjustwidth}{1cm}{0cm}
\begin{tikzpicture}
\node[vnodeFaded](vnodeu1){$\tu_1$};
\node[vnode, below of=vnodeu1, node distance=1cm](vnodeu2){$\tu_2$};
\node[vnode, below of=vnodeu1, node distance=2cm](vnodeu3){$\tu_3$};
\node[vnode, below of=vnodeu1, node distance=3cm](vnodeu4){$\tu_4$};

\node[vnodeFaded, below of =vnodeu4, node distance=1.5cm](vnodev1){$\tv_1$};
\node[vnode, below of=vnodev1, node distance=1cm](vnodev2){$\tv_2$};
\node[vnode, below of=vnodev1, node distance=2cm](vnodev3){$\tv_3$};

\node[cnode, right of= vnodeu1, yshift=-1cm, node distance=2.5cm](cnodeu1v1){$\tu_1 \tv_1$};
\node[cnode, below of= cnodeu1v1, node distance=1cm](cnodeu4v1){$\tu_4 \tv_1$};
\node[cnode, below of= cnodeu1v1, node distance=2cm](cnodeu3v2){$\tu_3 \tv_2$};
\node[cnode, below of= cnodeu1v1, node distance=3cm](cnodeu1v3){$\tu_1 \tv_3$};
\node[cnode, below of= cnodeu1v1, node distance=4cm](cnodeu2v2){$\tu_2 \tv_2$};

\draw [gray!50] (vnodeu1)--(cnodeu1v1);
\draw [gray!50] (vnodev1)--(cnodeu1v1);
\draw [gray!50] (vnodeu1)--(cnodeu1v3);
\draw (vnodev3)--(cnodeu1v3);
\draw (vnodeu3)--(cnodeu3v2);
\draw (vnodev2)--(cnodeu3v2);
\draw (vnodeu2)--(cnodeu2v2);
\draw (vnodev2)--(cnodeu2v2);
\draw (vnodeu4)--(cnodeu4v1);
\draw [gray!50] (vnodev1)--(cnodeu4v1);
\node[above of=cnodeu1v1, node distance=0.8cm, xshift =0.2cm, text width=4cm, align=center, scale=0.8, text=white](numRnodeLabel){};
\node[below of=cnodeu1v3, node distance=1.8cm, text width=3.5cm, align=center, scale=0.8, text=white](rightDist1){Right nodes connected to $\tu_1$ have degree 1; others have degree 2};
\end{tikzpicture}
\end{adjustwidth}
\vspace{-0.1cm}
\caption{$t=1$.}
\label{subfig:graph_1B_non_sym_t=1}
\end{subfigure}
\begin{subfigure}[B]{0.24\linewidth} 
\centering 
\begin{adjustwidth}{1cm}{0cm}
\begin{tikzpicture}
\node[vnodeFaded](vnodeu1){$\tu_1$};
\node[vnode, below of=vnodeu1, node distance=1cm](vnodeu2){$\tu_2$};
\node[vnode, below of=vnodeu1, node distance=2cm](vnodeu3){$\tu_3$};
\node[vnode, below of=vnodeu1, node distance=3cm](vnodeu4){$\tu_4$};

\node[vnodeFaded, below of =vnodeu4, node distance=1.5cm](vnodev1){$\tv_1$};
\node[vnode, below of=vnodev1, node distance=1cm](vnodev2){$\tv_2$};
\node[vnodeFaded, below of=vnodev1, node distance=2cm](vnodev3){$\tv_3$};

\node[cnode, right of= vnodeu1, yshift=-1cm, node distance=2.5cm](cnodeu1v1){$\tu_1 \tv_1$};
\node[cnode, below of= cnodeu1v1, node distance=1cm](cnodeu4v1){$\tu_4 \tv_1$};
\node[cnode, below of= cnodeu1v1, node distance=2cm](cnodeu3v2){$\tu_3 \tv_2$};
\node[cnode, below of= cnodeu1v1, node distance=3cm](cnodeu1v3){$\tu_1 \tv_3$};
\node[cnode, below of= cnodeu1v1, node distance=4cm](cnodeu2v2){$\tu_2 \tv_2$};

\draw [gray!50] (vnodeu1)--(cnodeu1v1);
\draw [gray!50] (vnodev1)--(cnodeu1v1);
\draw [gray!50] (vnodeu1)--(cnodeu1v3);
\draw [gray!50] (vnodev3)--(cnodeu1v3);
\draw (vnodeu3)--(cnodeu3v2);
\draw (vnodev2)--(cnodeu3v2);
\draw (vnodeu4)--(cnodeu4v1);
\draw [gray!50] (vnodev1)--(cnodeu4v1);
\draw (vnodeu2)--(cnodeu2v2);
\draw (vnodev2)--(cnodeu2v2);
\node[above of=cnodeu1v1, node distance=0.8cm, xshift =0.2cm, text width=4cm, align=center, scale=0.8, text=white](numRnodeLabel){};
\node[below of=cnodeu1v3, node distance=1.8cm, text width=3.5cm, align=center, scale=0.8, text=white](rightDist1){Right nodes connected to $\tu_1$ have degree 1; others have degree 2};
\end{tikzpicture}
\end{adjustwidth}
\vspace{-0.1cm}
\caption{$t=2$.}
\label{subfig:graph_1B_non_sym_t=2}
\end{subfigure}
\caption{\small 
(a): Pruned bipartite graph for stage B when $\bX $ is non-symmetric.  Here, $\bX= \sigma \bu \bv^T = \tilde{\bu} \tilde{\bv}^T$ with $\tu_j\neq 0$ for $j\in [4]$, $\tu_j=0$ for $5\le j\le n$, $\tv_j\neq 0$ for $j\in [3]$ and $\tv_j=0$ for $4\le j\le n$. 
(b)--(d): The graph process that models the recovery of $\tilde{\bu}$ and $\tilde{\bv}$. At $t=0$, $\tu_1$ is recovered as 1 and peeled off.  Following the recovery of $\tu_1$,  the left nodes $\tv_1, \tv_3$ (and  $\tu_4$) can be peeled off too. }
\label{fig:graph_process_1B_non_sym}
\end{figure}
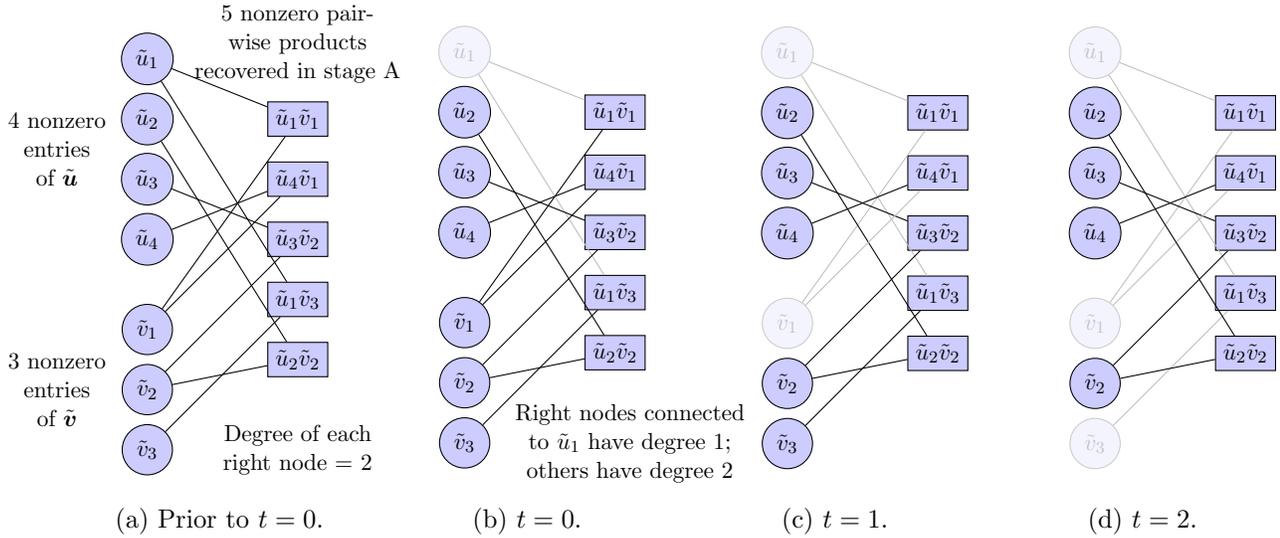

Since a non-symmetric matrix has no entries in squared form,  the peeling process is  initiated  by arbitrarily assigning the value 1 to one of the left nodes who have at least one neighbouring right node. At $t=0$, this left node is  peeled off from the graph, and its connecting right node then have degree 1.   In  Fig. \ref{subfig:graph_1B_non_sym_t=0}, the algorithm recovers $\tu_1$ as 1 and peels it off at $t=0$. This turns the right nodes $\tu_1\tv_1$ and $\tu_1\tv_3$ from degree-2 into degree-1.
 
 For $t\ge 1$, like in the symmetric case, the algorithm peels one degree-1 right node at a time along with the left node connecting to the right node, until no degree-1 right nodes remain in the graph.   In Figs. \ref{subfig:graph_1B_non_sym_t=1} and \ref{subfig:graph_1B_non_sym_t=2}, at $t=1$ and $t=2$, $\tv_1$ and $\tv_3$ are recovered and peeled off sequentially. 
 
For rank $r>1$, when both sets of singular vectors $\{\bu_i\}$ and $\{\bv_i\}$ have disjoint supports, the bipartite graph for  stage B consists of $r$ disjoint subgraphs. The peeling algorithm in this case is equivalent to the rank-1 algorithm run in parallel on the $r$ subgraphs.

For  the general case with $r>1$ and the $\{\bu_i\}$ or $\{\bv_i\}$ having overlapping supports, like in the symmetric case, we use a large enough sketch  such that stage A  recovers all the nonzeros in $\bX$. A singular value decomposition on the recovered nonzero submatrix gives the nonzeros in $\{\bu_i, \bv_i\}$.


\section{Main results\label{sec:main_results}}
\begin{thm}[Noiseless symmetric case]
\label{thm:main_result_symm}%
Consider the matrix  $\bX = \sum_{i=1}^r \lambda_i \bv_i \bv_i^{T} $ \edit{where each eigenvector  $\bv_i$ has  $k$ nonzero entries.}
For sufficiently large $k$, the sketching scheme with recovery algorithm described in Sections \ref{subsec:sketching_scheme}--\ref{subsec:recovery_alg_rankr} 
has the following guarantees.
\begin{enumerate}[label=\arabic*)]
\item \label{thm:results_disjoint_supports} For the case $r=1$  or $r>1$ with the supports of $\{\bv_i\}$ disjoint, fix $\deltaRangeInLine$ and let  $\nBins= d r\tk/(\delta \ln k)$. Here, we recall that $\tk = \binom{k}{2}+k$, and 
$d \ge 2$ is  the number of ones in each column of $\bH$ (i.e., the degree of each left node in the pruned stage A bipartite graph).

Then, with probability at least $1 \, -  2r\exp(-\frac{1}{30} k^{1-\delta})$,
the two-stage algorithm recovers $\{\bv_i\}$ (up to a sign ambiguity) from the sketch of size $\nMeas=2\nBins$  with running time $\bigo( rk^2 / \ln k)$. 

\item \label{thm:results_overlapping_supports} When  $r>1$ and  the supports of $\{\bv_i\}$ may overlap, with probability at least $1-\bigo(k^{-2})$,  stage A of the  algorithm followed by  eigendecomposition of the recovered nonzero submatrix  recovers  $\{\bv_i\}$ 
 from a sketch of size $\nMeas = 3r\tk$ with running time $\bigo\left((rk)^3\right)$.
\end{enumerate}
\end{thm}

\textbf{Remarks}:
\begin{enumerate}
    \item The eigenvectors $\{\bv_i\}$ can only be recovered up to a sign ambiguity since $\bv_i \bv_i^T = (-\bv_i)(-\bv_i)^T$. Moreover, when there is an eigenvalue with multiplicity   $m>1$,  the corresponding eigenvectors can only be recovered up to a rotation within the $m$-dimensional subspace they span. 

\item When the eigenvectors have different sparsities $\{k_1, \ldots, k_r \}$, Theorem \ref{thm:main_result_symm}  holds with the sample complexity depending on $\max\{ k_1, \ldots, k_r \}$ and the success probability  on $\min\{ k_1, \ldots, k_r \}$.

    \item \edit{The probability guarantees are with respect to the randomness in the sketching matrix (specifically, in  the locations of the ones in the parity check matrix $\bH$).
No randomness assumptions are made on  the eigenvectors $\{ \bv_i \}$, and the result holds for    any rank-$r$ matrix with $k$-sparse eigenvectors.}

    \item The bound on the success probability  in part \ref{thm:results_disjoint_supports}  is not tight. The bound also does not depend $d$ since this parameter appears in a term in the proof that is of a smaller order. In the simulations in Sections \ref{subsec:numerical_results} and \ref{subsec:numerical_results_noisy}, we use $d=2$. 
    
    \item \edit{The upper bound on the failure probability depends only on $k$.  In contrast, conventional sketching schemes using optimization based recovery, e.g. \cite{oymak2015simultaneously}, have failure probabilities decaying with $n$. This is because  the sketch size  in our scheme depends only on $k$, and not on the ambient dimension $n$, unlike the conventional schemes. In practice, the failure probability of our scheme is found to be extremely small even for small values of $\frac{k}{n}$  (where our failure probability is expected to decay slower). See the sharp phase transitions in performance measures in Figs. \ref{fig:prob_all_supports_recovered} and \ref{fig:noiseless_oymak}.}

    \item In part \ref{thm:results_disjoint_supports}, the first stage of the algorithm is similar to the peeling decoder for LDPC codes over an erasure channel \cite{luby2001efficient,richardson2008modern} and the decoder used for compressed sensing in \cite{bakshi2016SHOFA,li2019sublinear}. However, our sketching matrix $\bB$ (defined via the parity check matrix $\bH$) has row weights that scale with the sparsity level $k$. \edit{This  implies that each iteration of the peeling algorithm can introduce large changes in the degree distribution, which makes it challenging to bound the evolution of the peeling process.} Thus the existing peeling decoder analysis based on density evolution \edit{and Doob martingales} cannot be applied.  See comments following Lemma \ref{lem:1A_rank_r}.
    
\end{enumerate}

 \begin{thm}[Noiseless non-symmetric case]
\label{thm:nonsym_main_results}
Consider the matrix  $\bX = \sum_{i=1}^r \sigma_i \bu_i \bv_i^{T} $, \edit{where  each left singular vector $\bu_i$ has $k$ nonzero entries and each right singular vector $\bv_i$ has $\beta k$ nonzero entries}, for 
 some constant $\beta \in (0,1]$. 
For sufficiently large $k$, the sketching scheme with recovery algorithm described in Section \ref{subsec:nonsymm} has the following guarantees.
\begin{enumerate}[label=\arabic*)]
\item \label{thm:nonsym_results_disjoint_supports} 
 For the case $r=1$  or $r>1$ with each set of singular vectors $\{\bu_i\}$  and $\{\bv_i\}$  having disjoint supports,  fix $\delta \in (0, \frac{1}{2})$ and let $\nBins= d r \beta  k^2/(\delta \ln k)$. Then, with probability at least $1-2r \exp(-\frac{\beta}{8} k^{1-2\delta} )$, the two-stage algorithm recovers $\{\bu_i, \bv_i\}$  from the sketch of size $\nMeas=2\nBins$  with running time $\bigo( rk^2 / \ln k)$.
\item \label{thm:nonsym_results_overlapping_supports} When  $r>1$ and  the $\{\bu_i\}$  or $\{\bv_i\}$ have overlapping supports, with probability at least $1-\bigo(k^{-2})$,  stage A of the  algorithm followed by  singular value decomposition of the recovered nonzero submatrix  recovers $\{\bu_i, \bv_i\}$ 
 from a sketch of size $\nMeas = 3rk^2$ with running time $\bigo\left((rk)^3\right)$.
\end{enumerate}
\end{thm}


Remarks similar to those for the symmetric case (below Theorem \ref{thm:main_result_symm})  hold for the  non-symmetric case as well. 

\subsection{Proof of Theorem \ref{thm:main_result_symm}}

\paragraph{Part \ref{thm:results_disjoint_supports}.} In the setting of part \ref{thm:results_disjoint_supports},  the eigenvectors $\{\bv_i \}$ have disjoint supports with each $\bv_i$ being $k$-sparse. In this case, the nonzero entries of $\bX=\sum_{i=1}^r \lambda_i\bv_i\bv_i^T$ form $r$ disjoint submatrices, each of size $k \times k$. The result of Theorem \ref{thm:main_result_symm}  for this setting is proved via  two lemmas, which give high probability bounds on the number of nonzero matrix entries recovered in stage A and the number of nonzero eigenvector entries recovered in stage B, respectively.

\begin{lem}\label{lem:1A_rank_r}
Consider the setting of part \ref{thm:results_disjoint_supports} of Theorem \ref{thm:main_result_symm}.  
Let $\sf{A}$ be the fraction (out of $r \tk$)  of nonzero  entries in the upper triangular part of $\bX$ recovered in  stage A.   Then, with $\deltaRangeInLine$ as defined in Theorem \ref{thm:main_result_symm},  there exists   $\alphaUB = k^{-\delta} - o(k^{-\delta}) $ such that for sufficiently large $k$,
\beq\label{eq:bound_on_A_rank_r}
  \prob{\sfA \ge \alphaUB} \ge 1-4  \exp\left(-\frac{dr}{\nConst\delta^2} \frac{k^{2-\delta}}{\ln^2 k }\right) \, . 
\eeq
Moreover, in the $k \times k$ nonzero submatrix corresponding to $\lambda_i \bv_i \bv_i^{T}$, let $\sfA_i$ be the fraction (out of $\binom{k}{2}$)  of  above-diagonal entries  recovered and let $N_{Di}$ be the number of diagonal entries recovered, for $i \in [r]$.   Then there exists   $\alphaUB_{i} = k^{-\delta} - o(k^{-\delta}) $ for $i\in[r]$ such that for sufficiently large $k$,
\beq\label{eq:bound_on_all_Ai_rank_r}
\prob{ N_{Di} \ge 1 \text{ and }  \sfA_i \ge \alphaUB_{i}, \; \forall\; i\in [r]} \ge
  1-2r\exp \left(-\frac{1}{2}  k^{1-\delta}\right).
\eeq
\end{lem}

The proof of the lemma, given in Section \ref{subsec:proof_lem1A}, first establishes that in the bipartite graph at the start of stage A (Fig. \ref{subfig:pruned_graph_1A_t=0}), the degrees of the right nodes are each Binomial with mean $\delta \ln k$ and negatively associated. (Negative association is defined in Section \ref{subsec:prelim}). A Chernoff bound for negatively associated random variables is then used to obtain a high probability guarantee on the number of left nodes that are connected to singleton right nodes (and can hence be recovered).

 The next lemma shows that the conditional probability of recovering all the nonzero entries in each eigenvector is close to 1, given the high probability event in \eqref{eq:bound_on_all_Ai_rank_r}.
  \begin{lem} \label{lem:1B_rank_r}
 Consider the setting of part \ref{thm:results_disjoint_supports} of Theorem \ref{thm:main_result_symm} and let $ N_{Di}, \sfA_i,  \alpha_i^*$ be as defined in Lemma \ref{lem:1A_rank_r}.  Let $\sfB_i$ be the fraction (out of $k$)  of nonzero entries of $\bv_i$ that are recovered by the end of stage B.  Then, for  all $N_{Di}\geq 1, \sfA_i \ge \alphaUB_{i}$ and  sufficiently large $k$ we have:
\begin{equation}\label{eq:bound_on_Bi_rank_r}
\prob{\sfB_i < 1 \bigmid N_{Di} \, , \, \sfA_i } < \exp\left(-\frac{1}{30} k^{1-\delta}\right),  \quad  i \in [r].
\end{equation}
\end{lem}
The proof of the lemma is given in Section \ref{subsec:proof_lem1B}. 
Recall that on each  subgraph $i\in [r]$ in stage B, the algorithm sequentially identifies and peels off left nodes connected to degree-1 right nodes. To successfully recover all the $k$ nonzeros in the corresponding eigenvector $\bv_i$, the residual graph needs to have at least one degree-1 right node at the end of each iteration $0 \le t \le (k-2)$.  Conditioned on $\sfA_i$ and the high-probability event in \eqref{eq:bound_on_all_Ai_rank_r}, we show that at the start of each iteration $t$, the number of degree-1 right nodes connected to each remaining left node is approximately Binomial  with mean $\sfA_i t$ (See Lemma  
\ref{lem:num_singletons}). This is then used to show that with high probability, there is at least one degree-1 right node in each iteration until all the $k$ nonzeros  are recovered.

\emph{Proof of part \ref{thm:results_disjoint_supports} of Theorem \ref{thm:main_result_symm}}:  From Lemmas \ref{lem:1A_rank_r} and \ref{lem:1B_rank_r}, for sufficiently large $k$ we have:
\begin{equation}
\begin{split}\label{eq:proof_Thm1_from_two_lemmas}
    & \prob{ \sfB_i = 1,   \ \forall i \in [r] } \\
    & \ge \, \E \left[ \,  \prob{\sfB_i = 1,   \ \forall i \in [r] \mid (N_{Di}, \, \sfA_i)_{i \in [r]} }  \mathbbm{1}\{ N_{Di}  \ge 1\text{ and } \sfA_i \ge \alpha_i^*, \,  \forall i \in [r] \}  \right] \\
    & \stackrel{(\rm{i})}{\ge} \left(1 - \, r\exp\Big(-\frac{1}{30} k^{1-\delta}\Big) \right) \prob{N_{Di}  \ge 1 \text{ and } \sfA_i \ge \alpha_i^*, \,  \forall i \in [r]} \\
    &  \stackrel{(\rm{ii})}{\ge} \left(1 - \, r\exp\Big(-\frac{1}{30} k^{1-\delta}\Big) \right)
    \left( 1 \,  - \, 2r\exp\Big(-\frac{1}{2} k^{1-\delta}\Big) \right) \\
    & >  1  - \, 2r\exp\Big(-\frac{1}{2} k^{1-\delta}\Big) \, - \, r\exp\Big(-\frac{1}{30} k^{1-\delta}\Big) > 
    1  -  2r\exp\Big(-\frac{1}{30} k^{1-\delta}\Big),
    \end{split}
\end{equation}
where the inequality (\rm{i}) is obtained from  \eqref{eq:bound_on_Bi_rank_r}  along with a union bound. The inequality (\rm{ii}) follows from \eqref{eq:bound_on_all_Ai_rank_r}. This completes the proof of part \ref{thm:results_disjoint_supports} of Theorem \ref{thm:main_result_symm}. 
\hfill \qed
\paragraph{Part \ref{thm:results_overlapping_supports}.} When $\{ \bv_i \}$ have overlapping supports,  recall from Section \ref{subsec:recovery_alg_rankr} that the algorithm recovers all the nonzero entries of $\bX$, and then performs an eigendecomposition on the recovered nonzero submatrix.
In this case, the first stage is equivalent to the compressed sensing recovery of the  vectorized matrix $\bX$  using the scheme of \cite{li2019sublinear, bakshi2016SHOFA}. Theorem 4 in \cite{li2019sublinear} shows that the compressed sensing scheme can recover a $K$-sparse vector with probability at least $1- \bigo(1/K)$ with a sample complexity of $3K$ and running time $\bigo(K)$. This result directly yields the high probability guarantee in part \ref{thm:results_overlapping_supports} by noting that the number of nonzeros in the vectorized upper-triangular part of $\bX$ is  bounded below by $(\binom{k}{2}+k)$ and above by $r ({k \choose 2} +  k)$. 
    
\subsection{Proof of Theorem \ref{thm:nonsym_main_results} }
\paragraph{Part \ref{thm:nonsym_results_disjoint_supports}}  
Part \ref{thm:nonsym_results_disjoint_supports} of the theorem is proved using the following two lemmas  which characterize the high probability performance of stage A and stage B, respectively. 
\begin{lem}\label{lem:non_sym_Stage_A}
Consider the setting of part \ref{thm:nonsym_results_disjoint_supports} of  Theorem \ref{thm:nonsym_main_results}. 
Let $\sfA_i$ denote the fraction of entries recovered in the $k\times \beta k$
nonzero submatrix corresponding to $ \sigma_i \bu_i\bv_i^T$, for $i\in [r]$. 
Then, with $\delta\in (0,\frac{1}{2})$ as defined in Theorem  \ref{thm:nonsym_main_results}, there exists $\alpha_i^* = k^{-\delta}- o(k^{-\delta})$ for $i\in[r]$ such that for sufficiently large $k$, 
\beq
\prob{ \sfA_i \ge \alpha_i^*\,, \forall i\in[r]} 
\ge 1- 2r \exp\left(- \frac{\beta}{4} k^{\frac{3}{2}-\delta}\right).
\label{eq:Ai_LB_nonsymm}
\eeq
\end{lem}

\begin{lem}\label{lem:non_sym_Stage_B} Consider the setting of part \ref{thm:nonsym_results_disjoint_supports} of  Theorem \ref{thm:nonsym_main_results}, and let $\sfA_i, \alpha_i^* $ be as defined in Lemma \ref{lem:non_sym_Stage_A}. Let $\usfB_{i}$ and $\vsfB_{i} $ be the fraction of nonzeros in $\bu_i$ and $\bv_i$, respectively,  that are recovered in stage B. Then, for all $\sfA_i \ge \alpha_i^*$ and sufficiently large $k$, we have
\beq\label{eq:non_sym_Stage_B_main_result}
\prob{\usfB_i <1 \, \text{ or }   \, \vsfB_i<1 \mid  \sfA_i }  \le \exp\left(-\frac{\beta}{8} k^{1-2\delta}\right),\quad i\in[r].
\eeq
\end{lem}

The proofs of Lemmas \ref{lem:non_sym_Stage_A} and \ref{lem:non_sym_Stage_B} are similar to those of Lemmas \ref{lem:1A_rank_r} 
and \ref{lem:1B_rank_r}, respectively. We describe the main steps and highlight the key differences from the symmetric case in Sections \ref{sec:proof_Stage_A_main_lemma_nonsym} and  \ref{sec:proof_Stage_B_main_lemma_nonsym}.

Part \ref{thm:nonsym_results_disjoint_supports} of Theorem \ref{thm:nonsym_main_results}  follows from Lemmas \ref{lem:non_sym_Stage_A} and \ref{lem:non_sym_Stage_B}. Indeed, for sufficiently large $k$, we have:
\begin{equation}
    \begin{split}
        & \prob{\usfB_i = 1 \text{ and }\vsfB_i=1\,, \forall i\in[r]} \\
       & \ge \E\left[ \,   
\prob{\usfB_i =1 \text{ and } \vsfB_i = 1, \ \forall i\in[r] \mid \sfA_i,\  i\in [r] }
\mathbbm{1}\{ \sfA_i \ge \alphaUB_i, \, \forall \, i \in [r] \} \right] \\
& \stackrel{(\rm{i})}{\ge} \left( 1 - r \exp\Big(-\frac{\beta }{8}k^{1-2\delta} \Big)  \right) \prob{ \sfA_i \ge \alphaUB_i, \, \forall \, i \in [r] } \\
&\stackrel{(\rm{ii})}{\ge}  \left( 1 - r \exp\Big(-\frac{\beta}{8}k^{1-2\delta} \Big)  \right) \left( 1- 2r\exp\Big(-\frac{\beta}{4} k^{\frac{3}{2}-\delta} \Big) \right) \\
& \ge 1 \, -  \,  r \exp\Big(-\frac{\beta}{8} k^{1-2\delta}\Big)
\, - \, 2 r\exp\Big(- \frac{\beta}{4} k^{\frac{3}{2}-\delta} \Big) 
\ge 
1 \, -  \,  2 r \exp\Big(-\frac{\beta}{8}k^{1-2\delta}  \Big).
    \end{split}
    \label{eq:LB_prob_success_nonsym}
\end{equation}
Here the inequality (\rm{i}) is obtained using Lemma \ref{lem:non_sym_Stage_B}, and (\rm{ii}) using Lemma \ref{lem:non_sym_Stage_A}. 
\hfill \qed

\paragraph{Part \ref{thm:nonsym_results_overlapping_supports}.} This proof is similar to the symmetric case (part  \ref{thm:results_overlapping_supports} of Theorem \ref{thm:main_result_symm}). Indeed, applying \cite[Theorem 4]{li2019sublinear}  guarantees that with probability $1- \bigo(k^{-2})$, stage A of the algorithm recovers all the nonzero entries in the matrix.  The nonzeros in the singular vectors are then obtained via a singular value decomposition on the recovered nonzero submatrix. 
    

\subsection{Computational cost of the recovery algorithm} \label{sec:proof_sketch_computational_cost}
We discuss the running time in the symmetric case, with the non-symmetric case being analogous.
\paragraph{Eigenvectors with disjoint supports.} Consider the setting of part \ref{thm:results_disjoint_supports} of Theorem \ref{thm:main_result_symm}, where the supports of  $\{\bv_i\}$ are disjoint. Lemma \ref{lem:1A_rank_r} guarantees that for any $\delta \in (0,1)$ and sufficiently large $k$, the fraction $\sfA$ of nonzero entries of $\bX$ recovered in stage A is at least $k^{-\delta}(1 - o(1))$ with high probability. We will analyze the complexity of the two stages assuming that $\sfA = k^{-\delta}(1 - o(1))$. 
This is without loss of generality as  one can terminate stage A of the algorithm (prematurely) once a fraction $ k^{-\delta}(1 - o(1))$
of nonzero entries of $\bX$ have been recovered. From the proof of Lemma \ref{lem:1A_rank_r}, this also means that the fraction of entries recovered in the nonzero submatrix corresponding to $\lambda_i \bv_i\bv_i^T$  is $\sfA_i\sim k^{-\delta}$, for each $i\in [r].$  By Lemma \ref{lem:1B_rank_r}, this is sufficient for recovering each $\bv_i$ with high probability.

\emph{Stage A}: To begin, each of the $R$ sketch bins  requires $\bigo (1)$ numerical operations to be classified via a zeroton test and a singleton test (specified in \eqref{eq:noiseless_singleton_test}). This requires a total of $\bigo(R) = \bigo( {rk^2}/{\ln k})$ operations. In each peeling iteration,  the contribution of a nonzero matrix entry is subtracted from the $d$ bins that it is involved in and these bins are re-classified. Since $d$ is a constant, each peeling iteration requires $\bigo(1)$ operations. By the termination assumption above, the number of iterations in stage A is $\bigo(rk^2  \sfA) = \bigo(rk^{2-\delta})$, corresponding to  $\bigo(rk^{2-\delta})$ operations. Therefore, the total computational cost for stage A is $\bigo({rk^2}/{\ln k}) + \bigo(rk^{2-\delta}) = \bigo({rk^2}/{\ln k})$.

\emph{Stage B}: Recall that the graph at the start of stage B consists of $r$ disjoint subgraphs, each with $k$ left nodes representing the nonzero entries in the corresponding eigenvector and right nodes representing  nonzero pairwise products recovered in stage A.  
Lemma \ref{lem:initial_left_degree_bound_1B} shows that with high probability, the degree of each left node in the $i$-th subgraph is $\bigo(\sfA_i k)$.  Since we assumed that $\sfA_i \sim k^{-\delta}$, 
the degree of each left node is  $\bigo(k^{1-\delta})$ with high probability.
The algorithm peels off one left node at a time, and the  cost of each peeling iteration is proportional to the number of edges peeled off during the iteration. Thus, to peel off all the $k$ left nodes  from the $i$-th subgraph, the computational cost is $\bigo(k^{2-\delta})$. Since there are $r$ subgraphs, the total computational cost for stage B is  $\bigo(rk^{2-\delta})$.

Finally, adding the costs for the two stages gives a total computational cost of $\bigo(rk^2/ \ln k)$.

\paragraph{Eigenvectors with overlapping supports.} In the setting of part \ref{thm:results_overlapping_supports} of Theorem \ref{thm:main_result_symm}, the algorithm recovers all the nonzero matrix entries in stage A. In this case, the number of nonzeros and the number of bins $R$ are both $\bigo(r k^2)$. Therefore, the computational cost of stage A is $\bigo(r k^2)$. The nonzeros in the eigenvectors are then recovered by an eigendecompostion of the recovered nonzero  submatrix. Since this submatrix has size at most $rk \times rk$, the computational cost of the eigendecomposition is   $\bigo((rk)^3)$, which dominates the total cost of the algorithm.

\subsection{Numerical results} \label{subsec:numerical_results}

We investigate the empirical performance of the  scheme for both symmetric and non-symmetric matrices, with $k$-sparse signal vectors $\{\bu_i, \bv_i\}$ that have disjoint or overlapping supports.  In the simulations,  each signal vector $\bu_i$ or $\bv_i$ is obtained by first sampling its $k$ nonzero entries from the mixture of Gaussians $\frac{1}{2}\normal(-5,1) + \frac{1}{2}\normal(5,1)$ and then normalizing so that the resulting vector has unit norm. When $\{\bu_i\}$ and $\{\bv_i\}$ have overlapping supports, we use the Gram–Schmidt process to ensure both sets of signal vectors are orthogonal. The number of ones in each column of the parity check matrix $\bH$ is chosen to be $d=2$.

\edit{In Fig. \ref{fig:prob_all_supports_recovered}, }we declare exact recovery if: i)  the locations of the nonzeros in the signal vectors are correctly recovered and the values of the nonzeros  are recovered within an absolute deviation of $10^{-7}$ from the ground truth, and ii) the eigenvalues or singular values are also recovered within an absolute deviation of $10^{-7}$ from the ground truth.  In each subfigure of  Fig. \ref{fig:prob_all_supports_recovered}, we plot the fraction of trials in which exact recovery is achieved versus  the sketch size $\nMeas$, for different sparsity levels $k$.
 In Figs. \ref{subfig:full_recovery_rank3_sym_disjoint}--\ref{subfig:full_recovery_rank3_nonsym_disjoint},  
the dashed lines show the sketch sizes specified in part 1) of Theorems \ref{thm:main_result_symm} or \ref{thm:nonsym_main_results} for the values of $\delta$ indicated in the caption.   The figures illustrate that for a fixed $\delta$, the success probability at the sketch size corresponding to the dashed lines increases with $k$. In Figs. \ref{subfig:full_recovery_rank3_sym_overlap}--\ref{subfig:full_recovery_rank3_nonsym_overlap}, the dashed lines indicate the sketch size specified in part 2) of the theorems.  The empirical success probability at the dashed lines similarly  increases with $k$.

 Each subfigure in Fig. \ref{fig:runtime_ambient_dim_diff_k} plots  the running time of the recovery algorithm (Matlab implementation) versus the ambient dimension $n$, for three different   sparsity levels $k$ and sketch sizes $\nMeas$. The plots confirm that the running time of the algorithm does not depend on  $n$.  The running times in different subfigures are not comparable  because the experiments were executed at different times on a shared machine.
\begin{figure}[!t]
  	\begin{subfigure}[H]{0.24\linewidth}
    		\includegraphics[width=\linewidth]{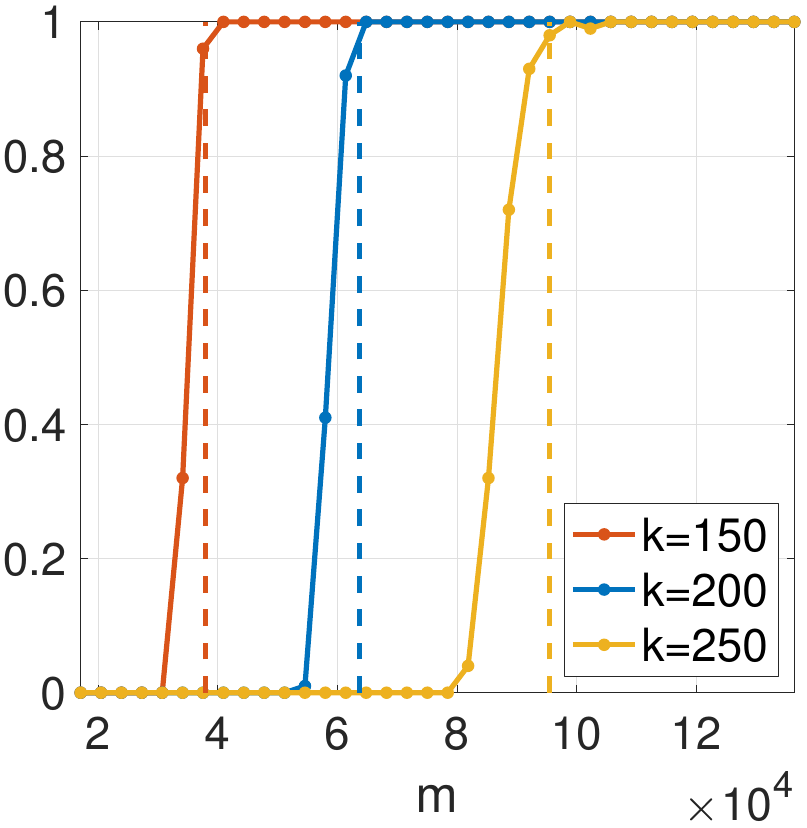}
    		\caption{\footnotesize Symmetric, $\{\bv_i\}$ with disjoint supports.}
    		\label{subfig:full_recovery_rank3_sym_disjoint}
    		\end{subfigure}
 \begin{subfigure}[H]{0.24\linewidth}
    		\includegraphics[width=\linewidth]{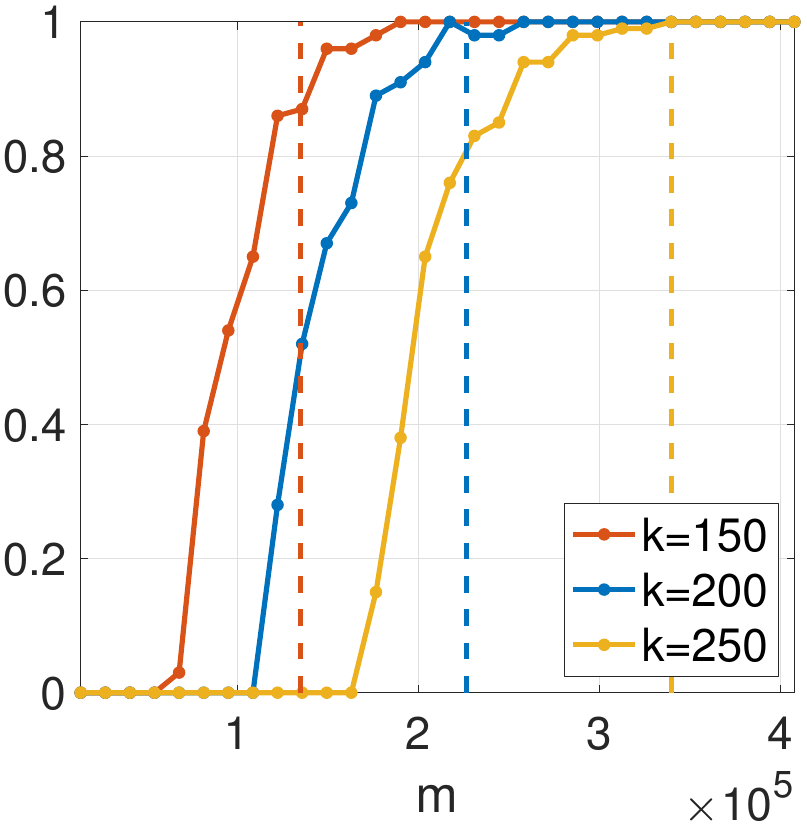}
    		\caption{\footnotesize Nonsymmetric, $\{\bu_i, \bv_i\}$ with  disjoint supports.}
    		\label{subfig:full_recovery_rank3_nonsym_disjoint}
    		\end{subfigure}  
    \begin{subfigure}[H]{0.24\linewidth}
    		\includegraphics[width=\linewidth]{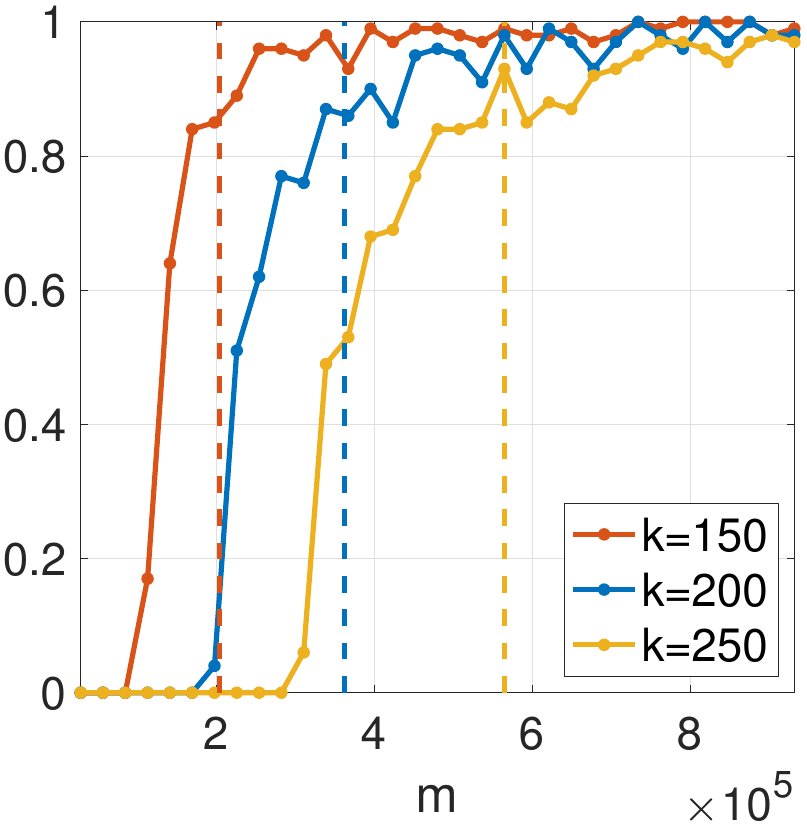}
    		\caption{\footnotesize Symmetric, $\{\bv_i\}$ with  overlapping supports.}
    		\label{subfig:full_recovery_rank3_sym_overlap}
    		\end{subfigure}
  	\begin{subfigure}[H]{0.24\linewidth}
   		\includegraphics[width=\linewidth]{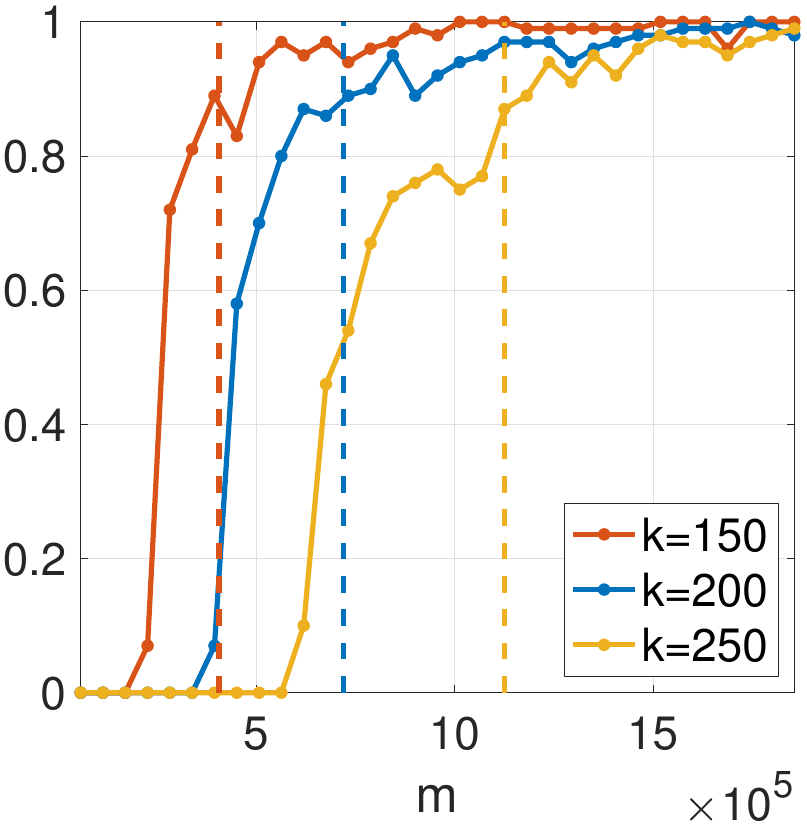}
    		\caption{\footnotesize Nonsymmetric, $\{\bu_i, \bv_i\}$ with overlapping supports.}
    		\label{subfig:full_recovery_rank3_nonsym_overlap}
    		\end{subfigure}    				
  	\caption{\small Probability of exact recovery ($y$-axis) versus sketch size $\nMeas$ ($x$-axis), for different sparsity levels $k$. The matrices used in all cases have rank $r=3$ and size $n \times n$ with $n=10^4$. The dashed lines indicate the sketch sizes stated in Theorems \ref{thm:main_result_symm} and \ref{thm:nonsym_main_results}; the dashed lines in  (a) and  (b) are  determined using $\delta =\frac{5}{7}$  and $\delta = \frac{2}{5}$, respectively.
Results are averaged over 100 trials.}
  	\label{fig:prob_all_supports_recovered}
	\vspace{-2pt} 	
\end{figure}
\begin{figure}[!t]
\centering
\begin{subfigure}[H]{0.244\linewidth}
\includegraphics[width=\linewidth]{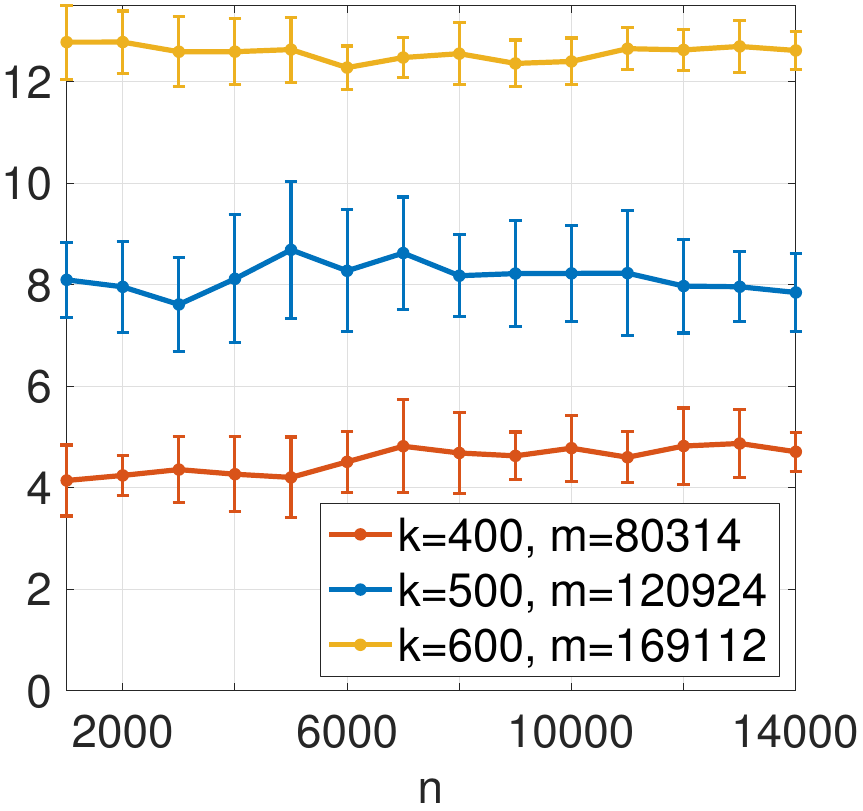}
   \caption{\\ \footnotesize Symmetric, $r=1$. }
 \label{subfig:runtime_rank1_sym}
 \end{subfigure}  
 \begin{subfigure}[H]{0.244\linewidth}
\includegraphics[width=\linewidth]{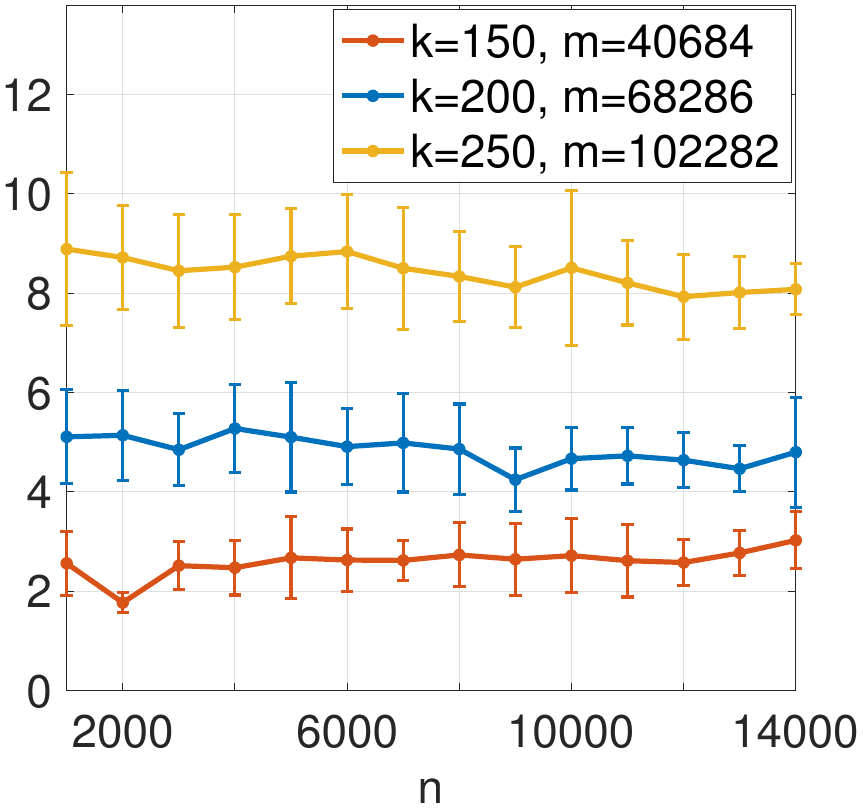}
   \caption{\footnotesize Symmetric, $r=3$. $\{\bv_i\}$ with  disjoint supports.}
 \label{subfig:runtime_rank3_sym_disjoint}
 \end{subfigure} 
 \begin{subfigure}[H]{0.244\linewidth}
 \includegraphics[width=\linewidth]{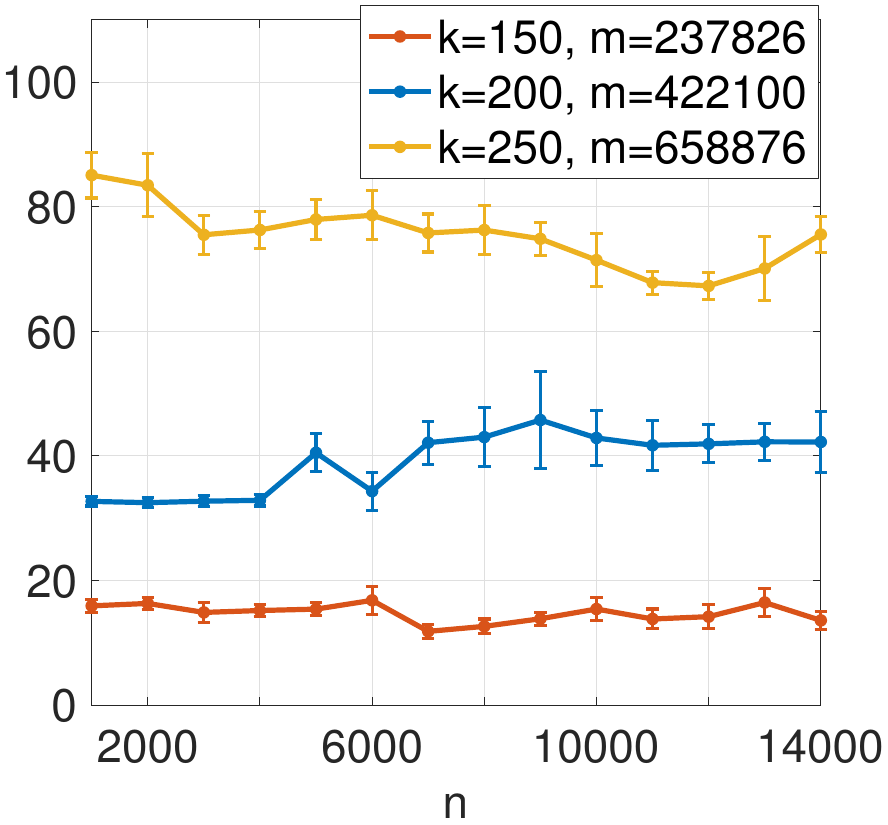}
   \caption{\footnotesize Symmetric, $r=3$. $\{\bv_i\}$ with  overlapping supports.}
 \label{subfig:runtime_rank3_sym_overlap}
 \end{subfigure} 
 \begin{subfigure}[H]{0.245\linewidth}
 \includegraphics[width=\linewidth]{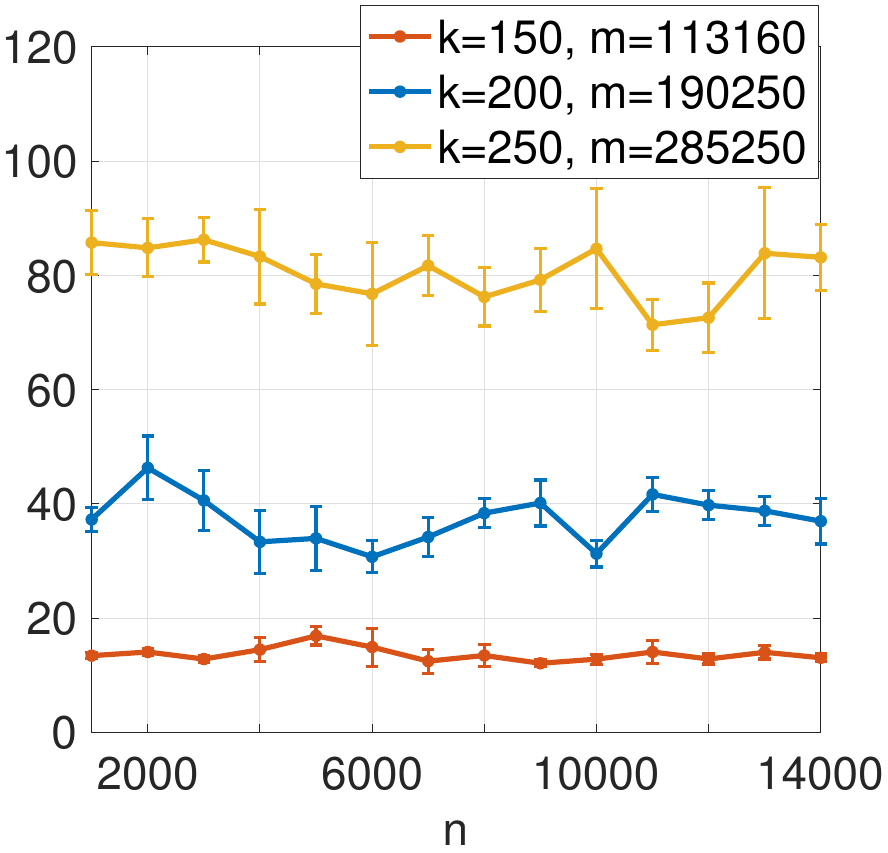}
   \caption{\footnotesize Nonsymmetric, $r=3$. $\{\bu_i, \bv_i\}$ w. disjoint supports}
 \label{subfig:runtime_rank3_nonsym_disjoint}
 \end{subfigure} 
\caption{\small Running time (in secs) of the recovery algorithm versus ambient dimension $n$, for  different  sparsity levels $k$ and sketch sizes $\nMeas$. Exact  recovery is achieved in every trial. Results are averaged over 50 trials and error bars indicate one standard deviation.}
\label{fig:runtime_ambient_dim_diff_k}	  
\vspace{-4pt} 	
\end{figure}


\begin{figure}[!b]
\begin{subfigure}[H]{\linewidth}
\centering
    \includegraphics[width=0.34\linewidth]{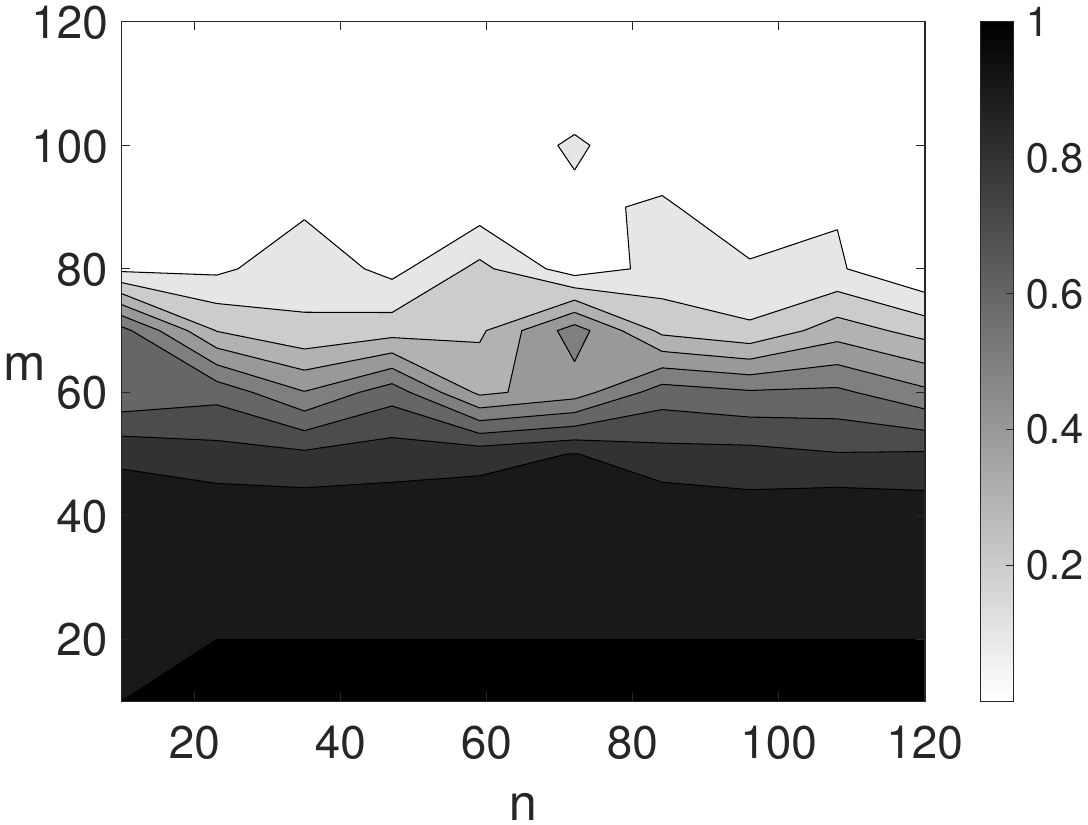}
    \hspace{1cm}
    \includegraphics[width=0.34\linewidth]{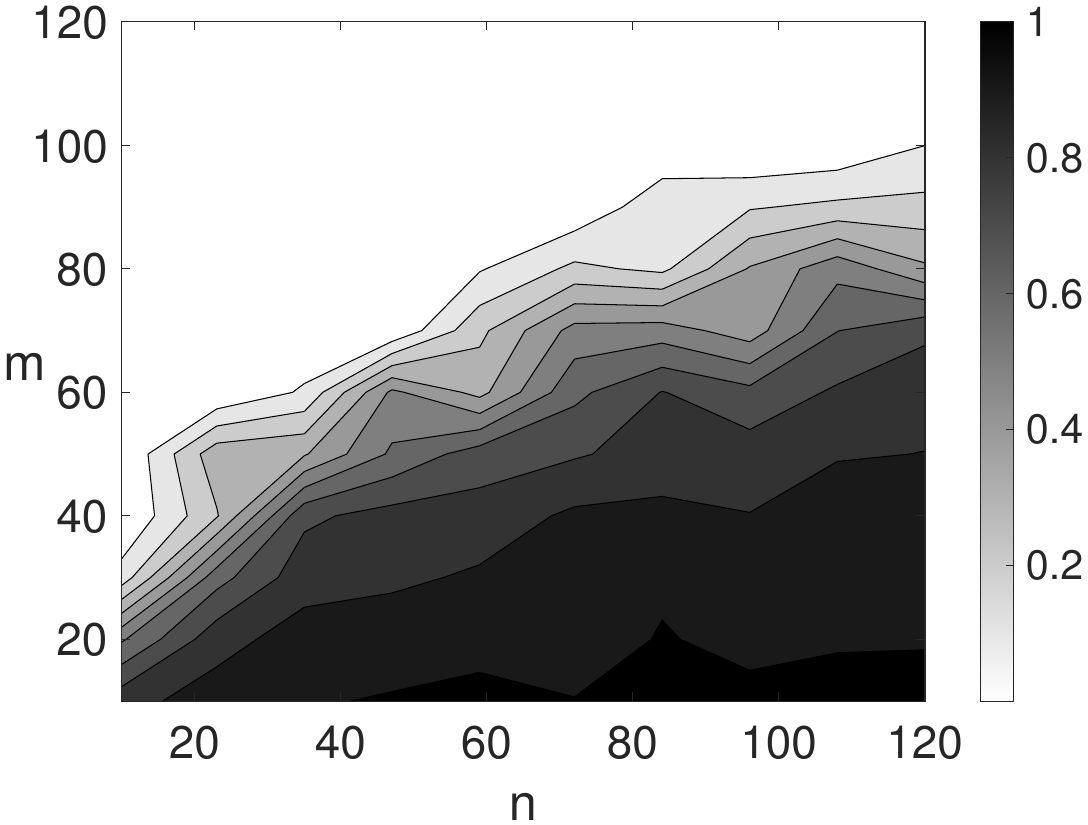}
    \vspace{-0.1cm}
    \caption{Disjoint supports in eigenvectors; $k=5$.}
    \label{subfig:nl_oymak_disjoint}
\end{subfigure}
 \begin{subfigure}[H]{\linewidth}
 \centering
    \includegraphics[width=0.34\linewidth]{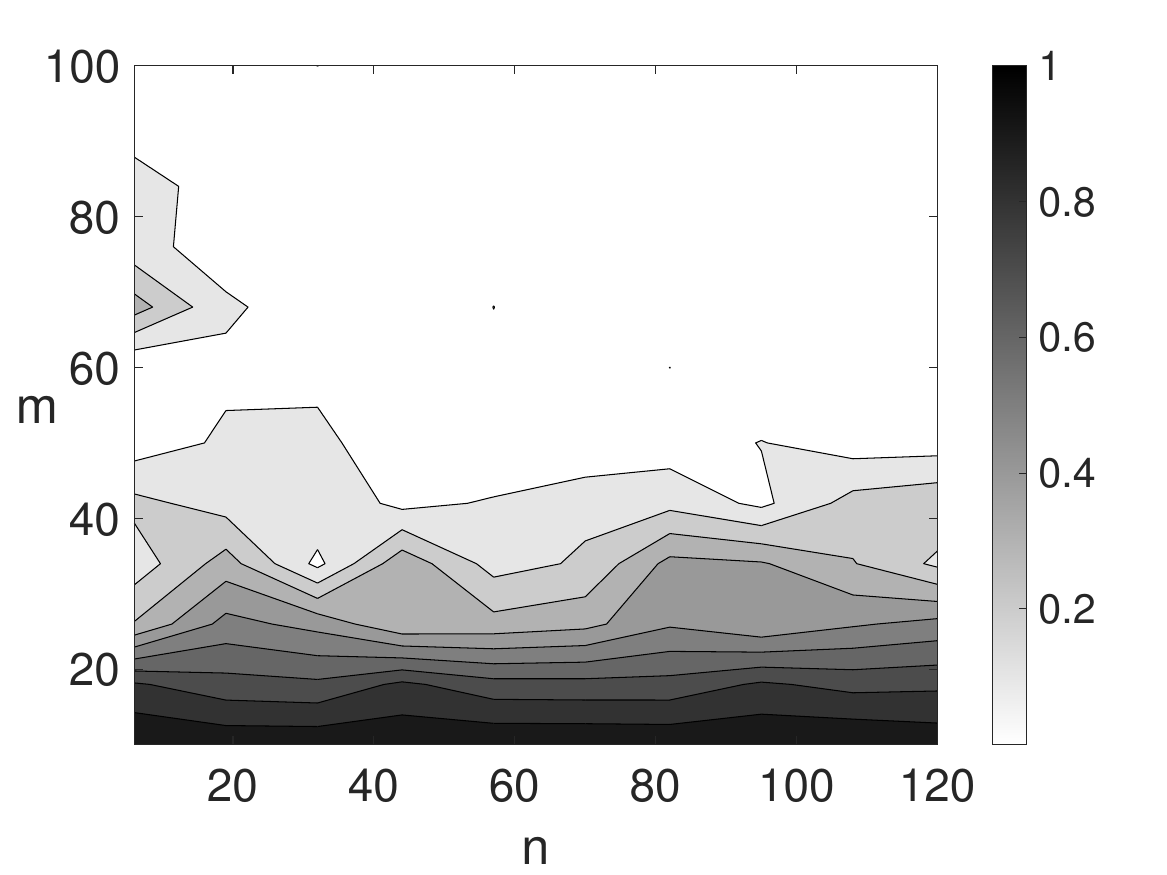}
    \hspace{1cm}
    \includegraphics[width=0.34\linewidth]{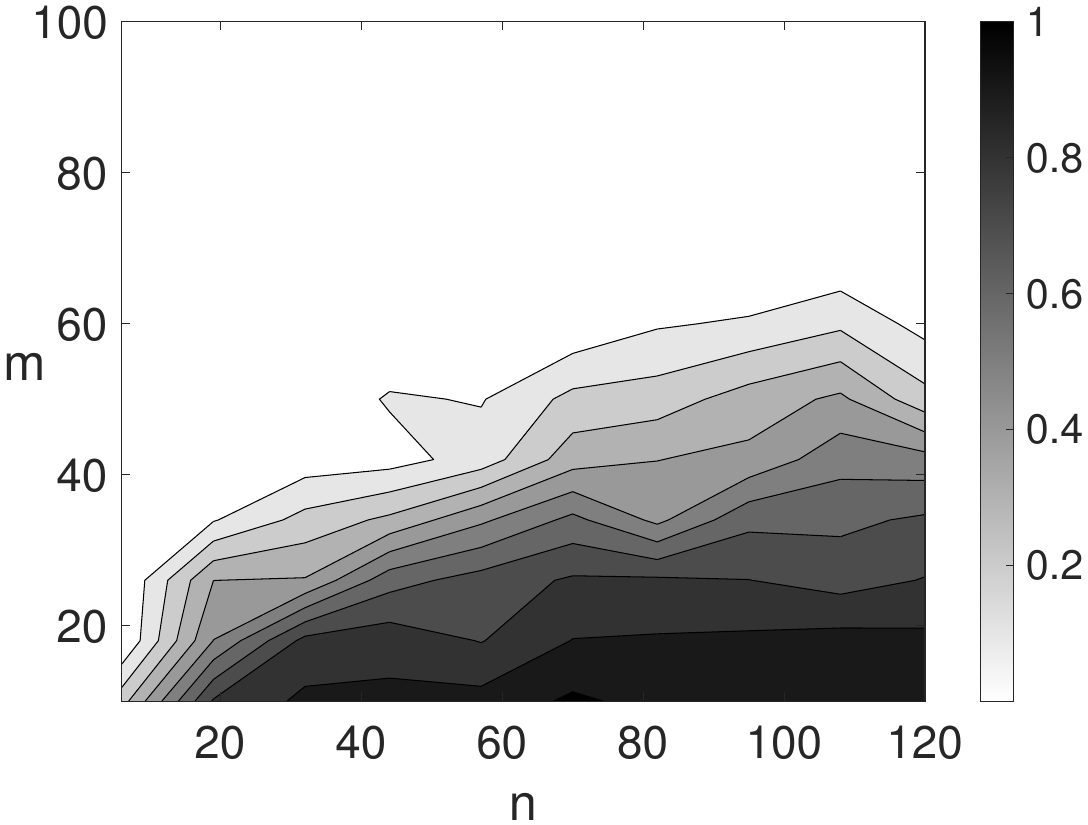}
    \vspace{-0.1cm}
    \caption{Overlapping supports in eigenvectors; $k=3$.}
   \label{subfig:nl_oymak_overlap}
\end{subfigure}  
\caption{\small \edit{Normalized error $\|\hat{\bX}-\bX\|_F/\|\bX\|_F$ of our scheme (left column), and the conventional scheme (right) using a  Gaussian sketching operator and recovery via the convex program in \eqref{eq:nl_oymak_obj}. Each heatmap shows the normalized error across different ambient dimension $n$ ($x$-axis) and  number of measurements $m$ ($y$-axis) for a fixed $k$.  
Results are averaged over 10 trials.}}
\label{fig:noiseless_oymak}
	\vspace{-2pt} 	
\end{figure}

\begin{figure}[!t]
\centering
\begin{subfigure}[H]{0.38\linewidth}
    \includegraphics[width=\linewidth]{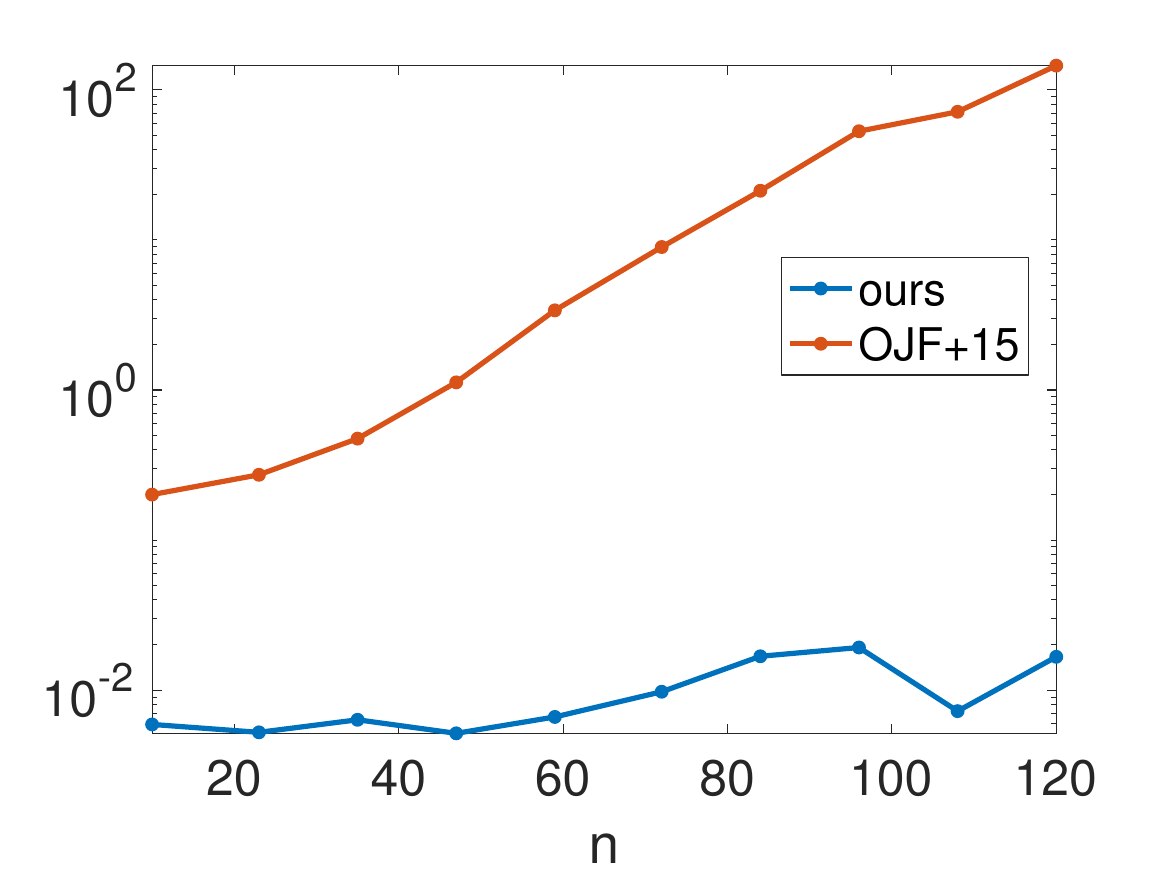}
    \caption{Same setting as in Fig. \ref{subfig:nl_oymak_disjoint}. Fix $k=5, m=100$ and vary $n$.}
\end{subfigure}
    \hspace{1cm}
\begin{subfigure}[H]{0.38\linewidth}
    \includegraphics[width=\linewidth]{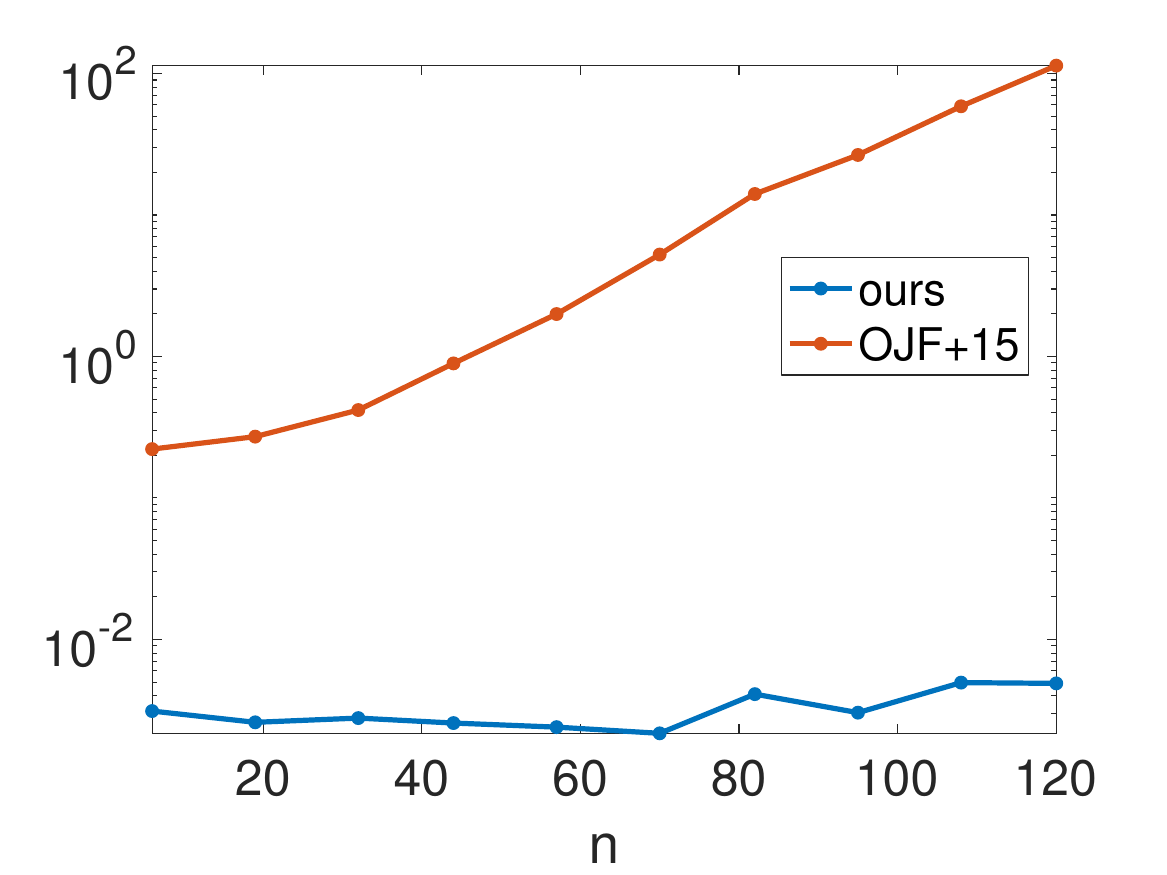}
    \caption{Same setting as in Fig. \ref{subfig:nl_oymak_overlap}. Fix $k=3, m=84$ and vary $n$.}
\end{subfigure}
\caption{\small \edit{Running time of our scheme and the conventional scheme using a  Gaussian sketching operator and recovery via the convex program in  \eqref{eq:nl_oymak_obj}. Same setting as in Fig. \ref{fig:noiseless_oymak}. Each plot shows the running time against $n$, for a fixed $k$ and $m.$
}}
%
\label{fig:nl_oymak_runtime}
	\vspace{-2pt} 	
\end{figure}

\edit{Fig. \ref{fig:noiseless_oymak} compares our scheme with a conventional scheme which  uses an i.i.d. Gaussian sketching operator and recovers the signal matrix $\bX$ (assumed to be positive semidefinite (PSD)) by solving  the convex program:
\begin{equation}
\hat{\bX} \in \argmin_{\bar{\bX}\succeq 0}\; \|\bar{\bX}\|_*+\lambda\|\bar{\bX}\|_1 \quad \text{subject to}
\quad \mc{A}(\bar{\bX}) = \mc{A}(\bX),
\label{eq:nl_oymak_obj}
\end{equation}
where $\lambda\geq 0$.  Here $\| \cdot \|_*$ and $\| \cdot \|_1$ denote the nuclear norm and $\ell_1$-norm, respectively.
Each eigenvector of  $\bX$ has $k$ nonzero entries chosen independently from the Gaussian mixture distribution $\frac{1}{2}\normal(-5,1) + \frac{1}{2}\normal(5,1)$. We use $\lambda=0.25$ in \eqref{eq:nl_oymak_obj}, following the choice of  in \cite[Fig. 7]{oymak2015simultaneously}.  
It was shown in  \cite[Thm. 3(c1)]{oymak2015simultaneously} that recovering $\bX$ via this convex program requires $m=\Omega(\min\{k^2, rn\})$ measurements, as opposed to $m=\bigo(rk^2)$ measurements required by our scheme. This indicates that our scheme requires smaller sample complexity in the very sparse regime (i.e., when $k=o(\sqrt{n})$), whereas  the conventional scheme is more sample efficient when the signal  vectors  are relatively dense. This is reflected in Figs. \ref{subfig:nl_oymak_disjoint} and \ref{subfig:nl_oymak_overlap}, where $k$ is fixed in each plot: the sketch size $m$ required for accurate recovery increases with $n$ for the conventional scheme, while it remains largely constant for our scheme.    Moreover,  our algorithm exhibits a sharp  transition in normalized error  as $m$ increases for any fixed $n$, whereas the conventional approach has a sharp transition only when $n$ is relatively small (the less sparse regime).
}

\edit{The optimization program \eqref{eq:nl_oymak_obj} is solved using CVX \cite{cvx,grant2008graph}  with the  commercial solver MOSEK (version 10.0.25) \cite{mosek}, which runs faster than other solvers such as SDPT3 and SeDuMi. Alternatively, we tried solving the dual problem using TFOCS \cite{tfocs}, which was developed in particular for sparse recovery applications, and discovered that this was slower than the CVX implementation. We also observe that the convex program  converges faster with the PSD constraint in \eqref{eq:nl_oymak_obj} than without, while our algorithm doesn't need this constraint.}

\edit{Even though we use the most powerful packages, the convex  program becomes prohibitively slow  when the matrix size grows beyond $100 \times 100$. For this reason, we restrict our experiments  to small matrices up to $120\times 120$. The running time of the experiments in Fig. \ref{fig:noiseless_oymak} is plotted in Fig. \ref{fig:nl_oymak_runtime} for a specific $m$. We  observe that our algorithm is $2$ to $4$ orders of magnitude  faster than the convex optimization algorithm. Moreover, the running time of the convex program grows with $n$, while our iterative algorithm doesn't. (The small increase in our running time in Fig. \ref{fig:nl_oymak_runtime}  with $n$ is an artefact since $k$ is small here). Hence our scheme scales well to much larger $n$ and $k$, as demonstrated by Fig. \ref{fig:runtime_ambient_dim_diff_k}.
}

\section{Recovery in the presence of noise}
\label{sec:noisy_description}

Consider the general model for the data matrix $\bX$ in \eqref{eq:noisy_observation_model}:
 \beq
 \bX = \bX_0 + \bW, \ \text{ where } \ \bX_0 = \sum_{i=1}^r \lambda_i \bv_i \bv_i^{T} \  (\text{symmetric}) \ \text{ or }  \ 
 \bX_0 = \sum_{i=1}^r \sigma_i \bu_i \bv_i^{T} \  (\text{non-symmetric}).
 \eeq
The vectors $\bv_i$ and $\bu_i$ are $k$-sparse, for $i \in [r]$. In the non-symmetric case, the elements of the noise matrix $\bW$ are assumed to be independent and identically distributed with zero mean and variance $\sigma^2$. \edit{We assume that $\sigma^2$ or at least a good upper bound on $\sigma^2$ is known.} In the symmetric case, the same assumption holds for the upper triangular entries, with $ W_{\ell, j} = W_{j, \ell}$ for $j < \ell \in [n]$.  

\subsection{Sketching scheme} \label{subsec:sketching_scheme_noisy}

We describe the scheme for symmetric matrices, with the non-symmetric case being analogous.
We write $\bx, \bx_0, \bw$ for the  vectorized upper-triangular parts of the symmetric matrices $\bX, \bX_0, \bW$, respectively. As before, we let $\tn =  {n \choose 2}+n$ and $\tk =  {k \choose 2}+k$. As in  the noiseless setting (Section \ref{subsec:sketching_scheme}), the sketching matrix $\bB $ is formed by combining two matrices: 
a sparse parity check matrix $\bH$, and a matrix $\bS$ that helps classify  each sketch  bin as a zeroton, singleton, or multiton. When there is noise, we cannot reliably identify  zerotons and singletons  by simply taking $\bS$ to be the first two rows of an $\tn$-point DFT matrix as in \eqref{eq:two_row_DFT}.  Therefore, following 
\cite{li2019sublinear}, we choose $\bS \in \reals^{P \times \tn}$  with
\begin{equation}
    P = \bigo\left(\log \left(\frac{n}{k}\right)\right)  \  \text{ and } \quad  S_{\kappa, \ell} \stackrel{\text{iid}}{\sim} \normal(0, 1), \  \kappa \in [P], \, \ell \in [\tn].
\end{equation}

The sparse parity check matrix $\bH \in \{ 0,1 \}^{\nBins \times \tn}$ is chosen in the same way as in the noiseless setting: column-regular with $d \ge 2$ ones per column in uniformly random locations. The sketching matrix $\bB \in \reals^{P \nBins \times \tn}$ is then formed as the column-wise Kronecker product of $\bH$ and $
\bS$, computed as illustrated in \eqref{eq:H_and_S_matrix_example}--\eqref{eq:B_matrix_example}. The sketch $\by = \bB \bx$ has size $m= P \nBins$.

We can view the sketch $\by$ as comprising $\nBins$ bins, each with $P$ entries. For $j\in [\nBins]$,  we let  $\by_j: = \begin{bmatrix} y_{(Pj-P +1)}, \dots, y_{(Pj)} \end{bmatrix}^T\in \reals^P$ be the $j$-th sketch bin. Similarly to \eqref{eq:y_j_general_expression}  and \eqref{eq:def_sparsified_x_vector}, writing $\bB_j \in \reals^{P \times \tn}$ for the $j$-th set of $P$ rows in $\bB$, 
 we have 
  \begin{align}\label{eq:expression_for_yj_noisy}
\by_j =\bB_j \bx  
 = \bS \bx_{0j}^*+ \bS\bw_j^*, \quad j \in [\nBins],
\end{align}
 where   the entries of $\bx_{0j}^* \in \reals^{\tn}$ and $\bw_j^* \in \reals^{\tn}$ are 
\begin{equation}\label{eq:def_sparsified_x_vector_w_vector_noisy}
(x^*_{0j})_\ell =\left\{
                \begin{array}{ll}
                  x_{0\ell} & \ell\in \mc{N}(j)\\
                  0 & \ell \notin \mc{N}(j)
                \end{array}
              \right., \quad 
(w^*_{j})_\ell=\left\{
                \begin{array}{ll}
                  w_\ell & \ell\in \mc{N}(j)\\
                  0 & \ell \notin \mc{N}(j)
                \end{array}
              \right. \quad \text{for } \ell\in \left[ \tn \right].
\end{equation}
As in \eqref{eq:def_sparsified_x_vector},  $\mc{N}(j)= \{\ell \in [\tn]: \, H_{j,\ell} =1\}$ contains the locations of the ones in the $j$-th row of $\bH$. The vectors $\bx^*_{0j}$ and $\bw^*_{j}$ store the signal  and  noise entries that contribute to the $j$-th bin. The $j$-th bin $\by_j$ is the sum of these entries  weighted by their respective columns in $\bS$.

\subsection{Recovery algorithm} \label{sec:recovery_algorithm_noisy}

We first consider symmetric rank-$r$ matrices where $\{ \bv_1, \ldots, \bv_r \}$ have disjoint supports. Fix  $\delta \in (0,1)$,  let  $\nBins = d r \tk/(\delta \ln k)$, and recall that the sketch size  is $m=P \nBins$.

\subsubsection{Accumulated noise variance in each bin}
The recovery of $\bx$ crucially depends on  the variance of the noise vector $\bS\bw_j^*$ in each sketch bin  $\by_j$ (see \eqref{eq:expression_for_yj_noisy}). Since $\bH$ has $d$ ones per column, each entry of $\bx$ contributes to $d$ bins, chosen uniformly at random from the $\nBins$ bins in the sketch. Thus we have  $\abs{\mc{N}(j)} \sim \bin{ \tn}{\frac{d}{\nBins}}$. Noting that  $\E\{\abs{\mc{N}(j)} \} = \frac{\tn d}{\nBins} = ({n \choose 2}+n)\frac{ \delta \ln k}{r\tk}$, we  apply the tail bound in Lemma \ref{lem:binom_tail_bound} to obtain
\begin{equation} \label{eq:Nj_tail_bound}
    \prob{ \abs{\mc{N}(j)} \le \frac{\delta}{4} \frac{n^2 \ln k}{rk^2}} \le \exp\left( - \frac{\delta}{7} \frac{n^2 \ln k}{r k^2}\right).
\end{equation}

Writing $\tbw_{j}^* :=\bS\bw^*_{j}$ for   the noise vector in \eqref{eq:expression_for_yj_noisy}, the entries in $\tbw_{j}^*$ are:
\beq\label{eq:accumulated_noise}
 (\tilde{w}^*_{j})_{\kappa}=\sum_{\ell\in \mc{N}(j)}{S}_{\kappa, \ell} \, w_{\ell}, \quad \text{for } \kappa \in[P].
\eeq
Recalling that $w_\ell$ is zero-mean with variance $\sigma^2$ and i.i.d. across $\ell$, conditioned on $\bS$ and $\mc{N}(j)$, the vector $\tbw_{j}^*$ is zero-mean  with covariance matrix 
$\bSigma \in\mathbb{R}^{P\times P}$, where
\begin{align}
\Sigma_{\kappa, \kappa} =\sigma^2\sum_{\ell\in  \mc{N}(j)} S_{\kappa,\ell}^2, \quad 
\Sigma_{\kappa,\iota} =\sigma^2\sum_{\ell\in  \mc{N}(j)} S_{\kappa,\ell}S_{\iota,\ell}
\quad \text{for }\kappa,\iota \in[P]\text{ and }\kappa\neq \iota. 
\end{align}

Since $\bS$ has independent standard Gaussian entries, $\E[S_{\kappa, \ell}^2] = 1$ and $\E[S_{\kappa, \ell} S_{\iota, \ell}] =0$. Moreover, \eqref{eq:Nj_tail_bound} implies that $\abs{\mc{N}(j)} \to \infty$ almost surely as $n \to \infty$. Therefore, by the weak law of large numbers, 
$\frac{\Sigma_{\kappa, \kappa}}{|\mc{N}(j)|} \to  \sigma^2$ and  
$\frac{\Sigma_{\kappa, \iota}}{|\mc{N}(j)|} \to 0$ as $n\to \infty$. It follows that the  distribution of the scaled accumulated noise vector 
$\frac{\tbw_{j}^*}{\sqrt{|\mc{N}(j)|}}$ converges to the   Gaussian $\normal(\bzero, \sigma^2\bI_P)$ as $n\rightarrow\infty$.  We will use this limiting distribution to devise zeroton and singleton tests in stage A of the algorithm.

\subsubsection{Stage A of the algorithm}

\paragraph{Zeroton test.} When $\by_j$ is a zeroton, i.e., when $\supp(\bx_0)\cap \mc{N}(j) =\emptyset$, from \eqref{eq:expression_for_yj_noisy} and \eqref{eq:def_sparsified_x_vector_w_vector_noisy} we have
\beq\label{eq:zeroton_expression_noisy}
\by_j=\bS\bw^*_{j}= \tbw_{j}^*, \quad \text{with } \frac{\tbw_{j}^*}{\sqrt{|\mc{N}(j)|}} \ \text{ approximately} \sim \normal\left(\bzero, \sigma^2\mathbf{I}_P\right).
\eeq
Based on this, a given bin $\by_j$ is declared a zeroton if
\begin{equation}\label{eq:robust_zeroton_test}
\frac{1}{P\abs{\mc{N}(j)}}\left\|\by_j\right\|^2\leq \gamma_0 \sigma^2,
\end{equation}
for a suitably chosen constant $\gamma_0$. 
Empirically, $99.7\%$ of the samples of a Gaussian distribution lie within $3$ standard deviations of the mean, so we set  $\gamma_0 \ge 3$. In our experiments, $\gamma_0=5$ gives reasonable performance.

\paragraph{Singleton test.} 
When $\by_j$ is a singleton with $\supp(\bx_0)\cap \mc{N}(j) =\ell$, it takes the form 
\beq\label{eq:singleton_expression_noisy}
 \by_j= x_{0\ell}\bs_\ell +\tbw^*_{j},
 \eeq
 where $\bs_\ell$ is the $\ell$-th column of $\bS.$ 
Given any bin $\by_j$, assuming it is a singleton of the form \eqref{eq:singleton_expression_noisy},   the algorithm first  estimates 
 the index $\ell$ and the value $x_{0\ell}$ as follows:
\begin{equation}\label{eq:ML_projection_coeff}
\hat{\ell}=\argmin_{\ell \in \mc{N}(j)} \left\|\by_j- \hx_{0\ell} \bs_\ell\right\|^2, \ \text{ where } \ \hx_{0\ell}=\frac{\bs_\ell^T\by_j}{\left\|\bs_\ell\right\|^2}.
\end{equation}
The bin $j$ is declared a singleton if
\begin{equation}\label{eq:robust_singleton_test}
\frac{1}{P\abs{\mc{N}(j)}}\left\| \by_j- \hat{x}_{0\hat{\ell}}\, \bs_{\hat{\ell}} \right\|^2\le \gamma_1 \sigma^2,
\end{equation}
for some constant $\gamma_1 \le \gamma_0$. 
The estimated index-value pair $(\hat{\ell}, \hx_{0\hat{\ell}})$ in \eqref{eq:ML_projection_coeff} is the maximum-likelihood estimate of $(\ell, x_{0\ell})$ corresponding to the singleton bin in \eqref{eq:singleton_expression_noisy}, given the  limiting  distribution of the  noise 
$\frac{\tbw_{j}^*}{\sqrt{|\mc{N}(j)|}}   $ is $\normal(\bzero, \sigma^2\mathbf{I}_P)$. We note that  $(\hat{\ell}, \hx_{0\hat{\ell}})$ will be an accurate estimate of the true index-value pair $(\ell, x_{0\ell})$ only when the accumulated noise level $\sqrt{\abs{\mc{N}(j)}} \sigma$ (i.e., the standard deviation of $\tbw^*_j$) is much smaller than the signal magnitude $\abs{x_{0\ell}}$.

 The algorithm declares $\by_j$ a multiton if it fails both zeroton and singleton tests.

 At $t=0$, the algorithm  identifies the initial set of zerotons and singletons among  the $\nBins$ bins using the tests above. Let $\mc{S}(t)$ denote the set of singletons at time $t$. Then at each $t \ge 1$,   the algorithm picks a singleton uniformly at random from $\mc{S}(t-1)$, and recovers the nonzero entry of $\bx_0$ underlying the singleton using the  index-value pair $(\hat{\ell}, \hat{x}_{0\hat{\ell}})$ in  \eqref{eq:ML_projection_coeff}.
  The algorithm then peels off the contribution of the $\hat{\ell}$-th entry of $\bx_0$ from each of the $d$ bins it is involved in, 
 re-categorizes these bins, and updates the set of singletons to include any new singletons created by the peeling of the $\hat{\ell}$-th entry. The updated set of singletons is $\mc{S}(t)$.  
 Stage A continues until singletons run out, i.e., when $\mc{S}(t) =  \emptyset$. \edit{This algorithm is similar to Algorithm \ref{alg:nl_stage_A} with additional input $\bS$, $\gamma_0$, $\gamma_1$ and $\sigma^2$, and different zeroton and singleton tests. We omit the pseudocode for brevity.} If the supports of $\{ \bv_1, \ldots, \bv_r \}$ are known to be disjoint, we proceed to stage B, described in the next subsection.  
 
 When the supports of the $r$ eigenvectors have overlaps, we take a large enough sketch (with $R=\bigo(r k^2)$) to identify all the nonzero entries in $\bX_0$ in stage A, and then perform an eigendecomposition on the recovered submatrix of nonzero entries to obtain the nonzero entries of $ \{ \bv_1, \ldots, \bv_r \}$.

 \emph{Picking the threshold parameters  $\gamma_0, \gamma_1$}: Decreasing the value of $\gamma_0$ makes the zeroton test in \eqref{eq:robust_zeroton_test}  more stringent. As $\gamma_0$ becomes smaller, zerotons are more likely to be  mistaken as singletons with  small signal values. This kind of error results in some zero entries of $\bX_0$ being mistakenly recovered as small nonzeros. Since our recovery algorithm crucially relies on the sparsity of $\bX_0$ to maintain its low complexity, empirically we find that such errors greatly increase the algorithm's running time.  Moreover, in the case where $\{\bv_i\}$ have disjoint supports, such errors 
can lead to stability issues in stage B when  peeling steps involve taking ratios between small signal values. For these reasons, we set $\gamma_0 = 5$ in our numerical simulations. We set a more stringent threshold  for the singleton test by choosing  $\gamma_1\le \gamma_0$ to prevent multitons from being mistaken as singletons.

\subsubsection{Stage B of the algorithm}

 In the stage B graph (Fig. \ref{subfig:graph_1B_t=-1} or Fig. \ref{subfig:graph_1B_non_sym_t=-1}),  the right nodes 
are the nonzero matrix entries (i.e.,  pairwise products) recovered in stage A.  The values of these recovered entries are no longer exact  in the noisy setting.  Hence, 
a left node connected to multiple degree-1 right nodes will receive different suggested values from these right nodes. Therefore, instead of peeling off edges like in Section \ref{subsubsec:stageB}, we use a message passing algorithm on the stage B graph.  In each iteration, all the left nodes with known values send these values to their neighboring right nodes. Recall that each right node in the stage B graph has two edges, and corresponds to the  product of the two left nodes it is connected to.  The message a right node sends along each edge is the ratio of its pairwise product  with the incoming message along the other edge.   Each left node then updates its value by taking the average of all messages it receives. This averaging operation makes the algorithm more robust to noise.

\subsection{Sample complexity and computational cost}\label{subsec:sample_time_complexity_noisy}

\paragraph{Disjoint supports}
We have  $R=\bigo(rk^2/\ln k)$ bins and  $P=\bigo(\log (n/k))$ components in each bin, therefore the sample complexity is  $RP=\bigo( rk^2 \log (n/k)/\ln k)$, which is  higher than the noiseless case by a factor of $P=\bigo(\log (n/k))$. 

The computational cost of the recovery algorithm is
$\bigo(n^2 \log(n/k))$,  significantly higher than the $\bigo( rk^2/ \ln k)$ in the noiseless setting. The increase in complexity is due to the singleton test in \eqref{eq:ML_projection_coeff}--\eqref{eq:robust_singleton_test}, which requires finding a minimum  over the  $\abs{\mc{N}(j)}$  left nodes that connect to bin $\by_j$, for each $j \in [R]$. From \eqref{eq:Nj_tail_bound} and a similar upper tail bound, we know that with high probability $\abs{\mc{N}(j)}$ is of the order $n^2 \ln k/(rk^2)$. Therefore, the total computational cost of the singleton test for the $R$ bins is $ R\abs{\mc{N}(j)} P=\bigo( n^2\log (n/k))$, which dominates the cost of stage A. When the noise level is small (so that in stage A, zero matrix entries are not mistakenly recovered as nonzeros), stage B has the same  $\bigo( rk^{2-\delta})$ complexity as in the noiseless case.  Thus the computational cost of the algorithm is dominated by stage A.

\vspace{-0.2cm}
 \paragraph{General case} The sample complexity is
 $\bigo( rk^2\log(n/k) )$, again  higher than the noiseless setting by a factor of $P=\bigo (\log(n/k))$. The computational cost of  recovering all the nonzero entries of the signal matrix is $\bigo(n^2\log(n/k))$ and that of the eigendecomposition of the recovered nonzero submatrix is $\bigo((rk)^3)$. Thus, the total  cost of the algorithm is $\bigo(\max\{n^2\log(n/k), (rk)^3\})$.
 
 In the noisy setting, the nonzero matrix entries  recovered in stage A will not be exact, which makes it hard to analyze the subsequent step (either stage B or eigendecomposition) and obtain a rigorous performance guarantee. This can be addressed by making suitable assumptions on the signal alphabet. The analysis is conceptually similar to the noiseless setting but with additional technical challenges, so it is omitted. In the following subsection, we demonstrate the performance of the scheme in the noisy case via numerical simulations.

 \subsection{Numerical results} 
 \label{subsec:numerical_results_noisy}
 
\begin{figure}[!t]
\begin{minipage}{0.31\textwidth}
    \centering
    \begin{subfigure}{\textwidth}
        \includegraphics[width=\textwidth]{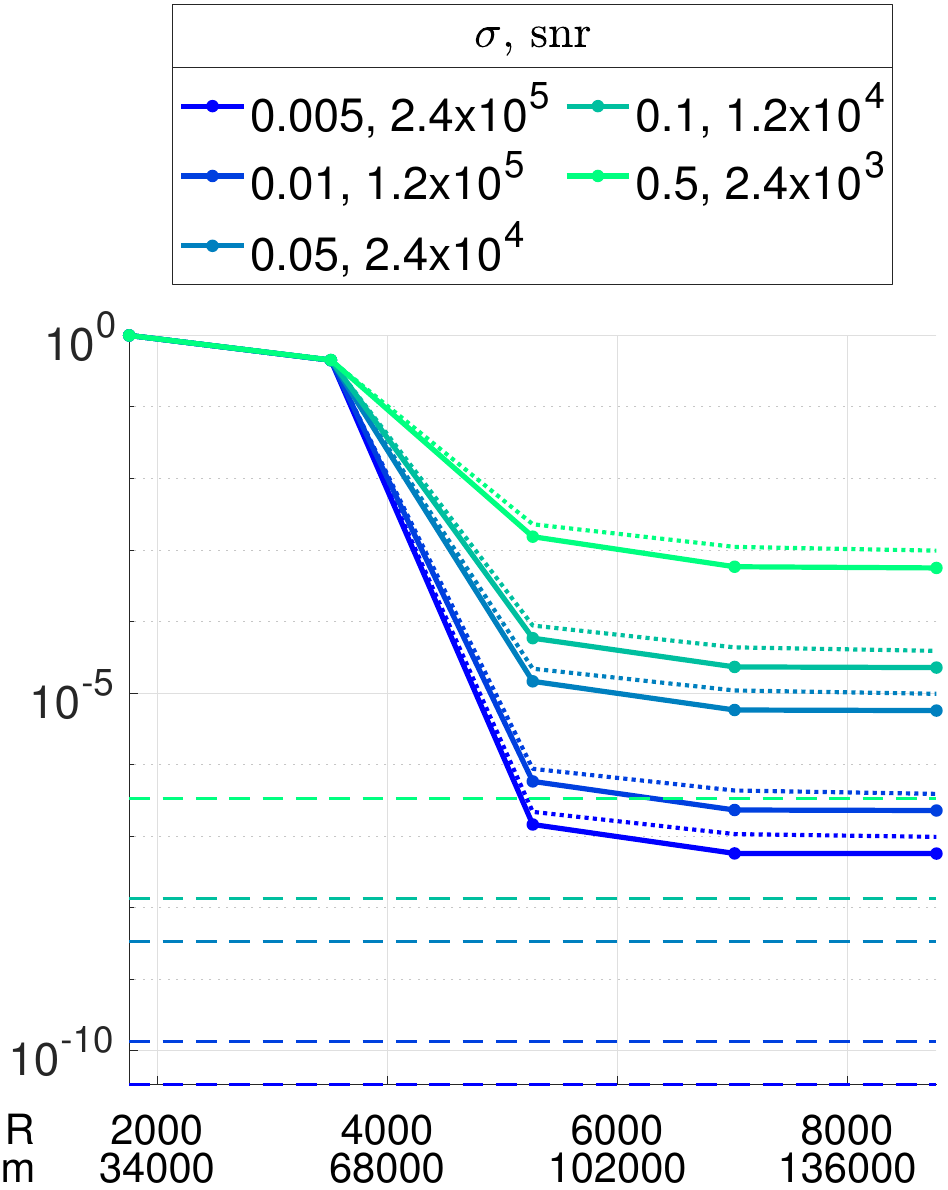}
    \caption*{(a1)}
    \end{subfigure}
    \begin{subfigure}{\textwidth}
        \includegraphics[width=\textwidth]{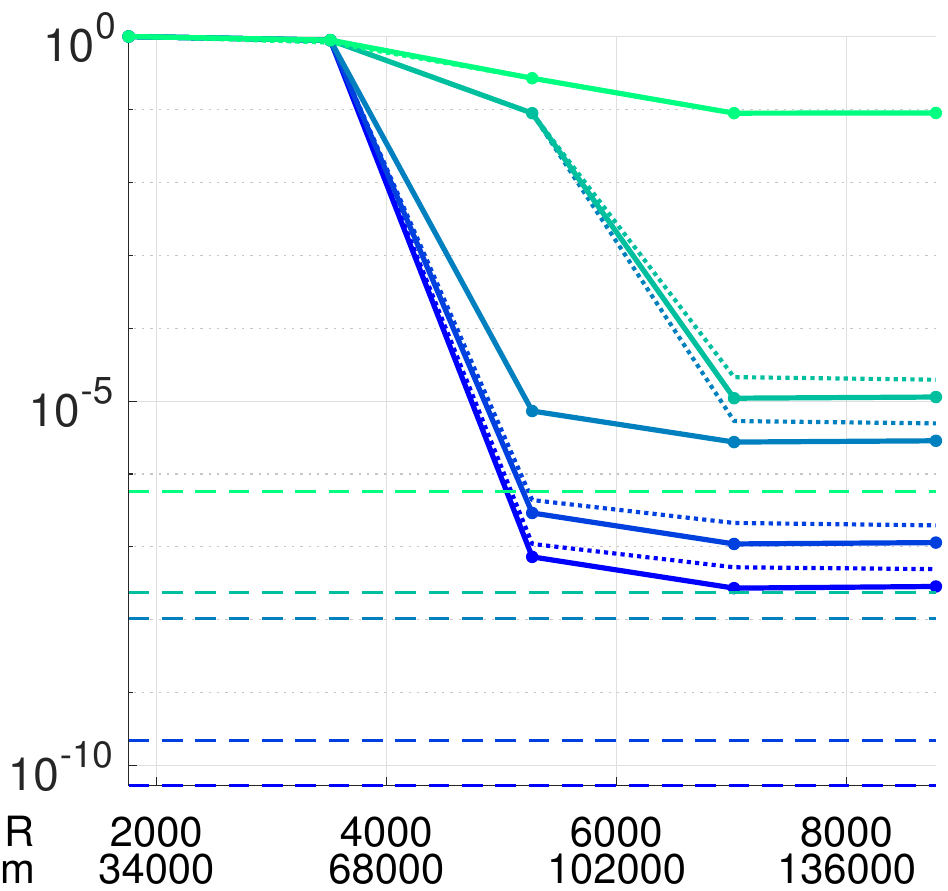}
    \caption*{(a2)}
    \end{subfigure}
    \vspace{-0.5cm}
\caption*{(a) \small $\{\bv_i\}$ with  disjoint supports. Vary $m$ via $R$;   \edit{$P=\ceil{4\log(n/k)}$}.}
    \label{subfig:noisy_varyR_disjoint}
\end{minipage}\hspace{1em}
\begin{minipage}{0.32\textwidth}
    \centering
    \begin{subfigure}{\textwidth}
        \includegraphics[width=\textwidth]{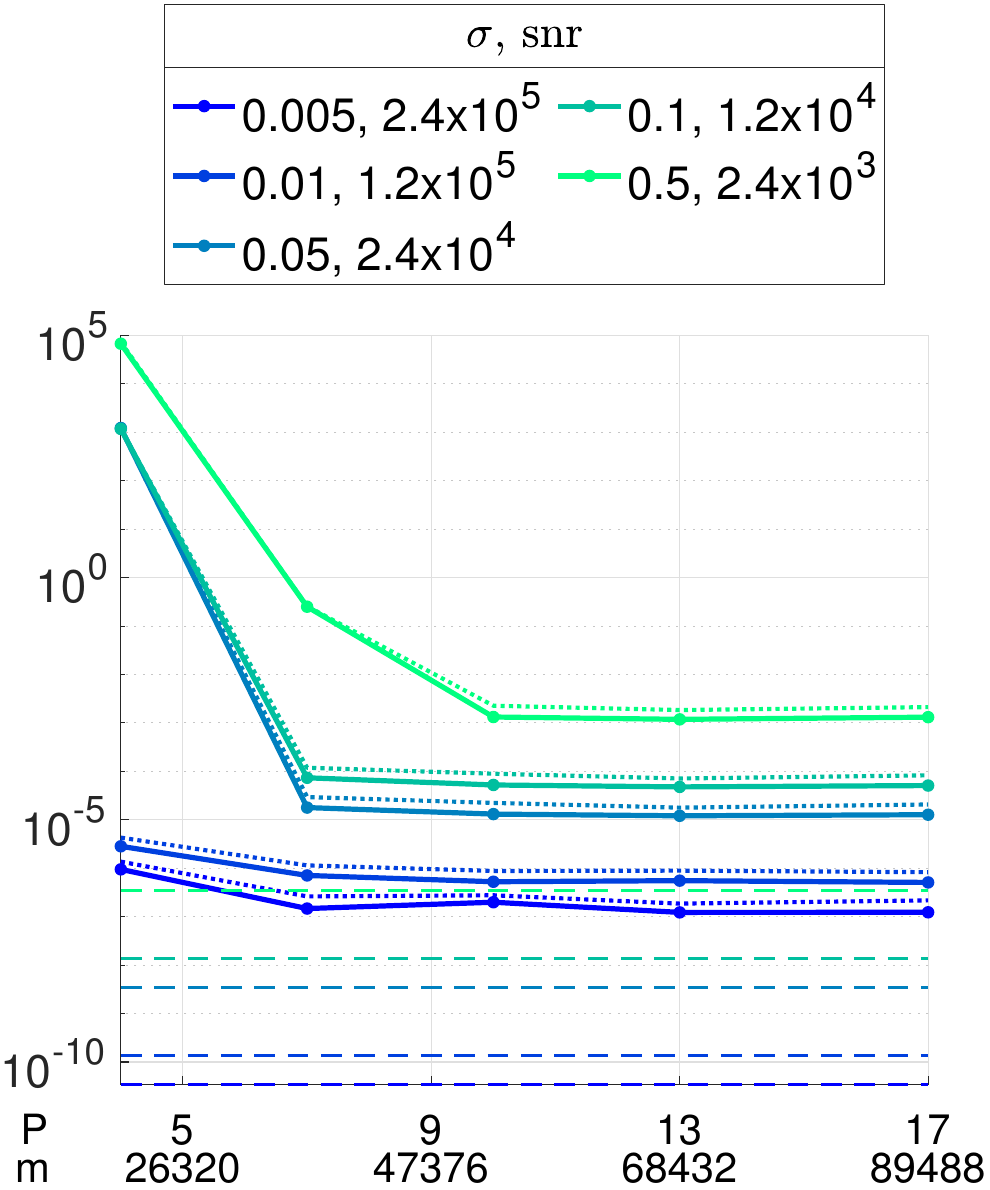}
        \caption*{(b1)}
    \end{subfigure}
    \begin{subfigure}{\textwidth}
        \includegraphics[width=\textwidth]{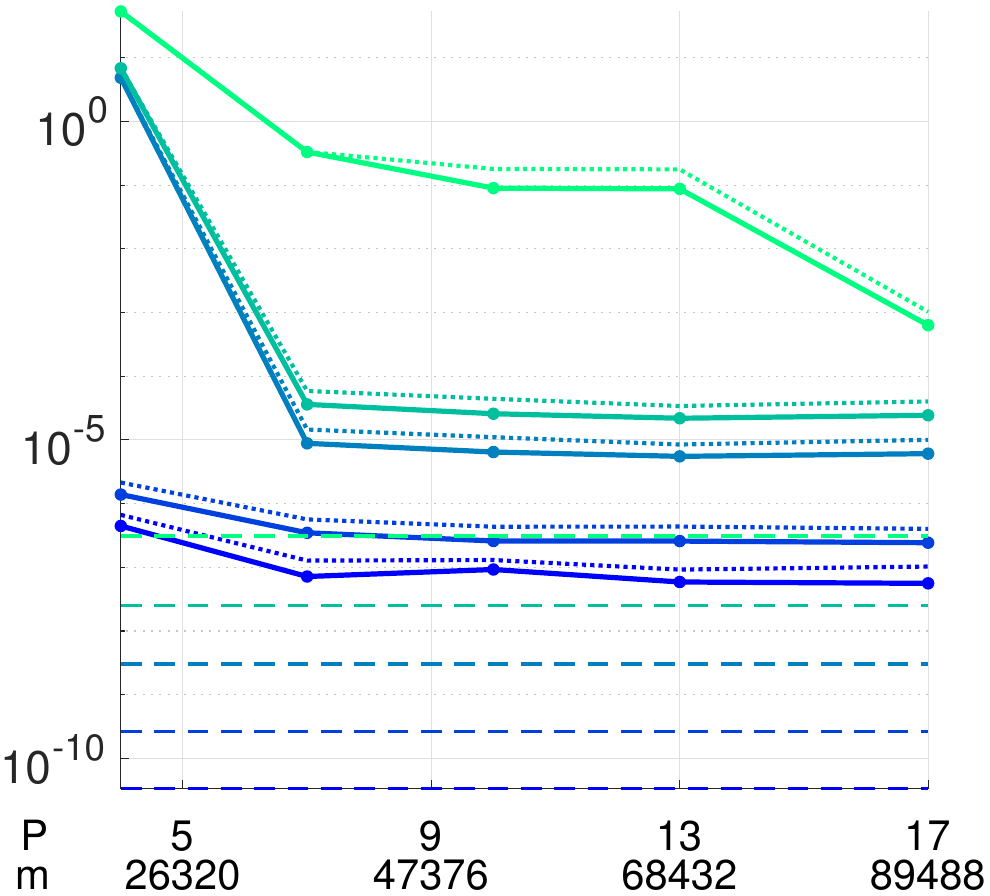}
        \caption*{(b2)}
    \end{subfigure}
    \vspace{-0.5cm}
    \caption*{(b) \small$\{\bv_i\}$ with disjoint supports. Vary $m$ via $P$; \edit{$R = 3r\tk/ (\ln k)$}.}\label{subfig:noisy_varyP_disjoint}
\end{minipage}\hspace{1em}
\begin{minipage}{0.33\textwidth}
    \centering
    \begin{subfigure}{\textwidth}
        \includegraphics[width=\textwidth]{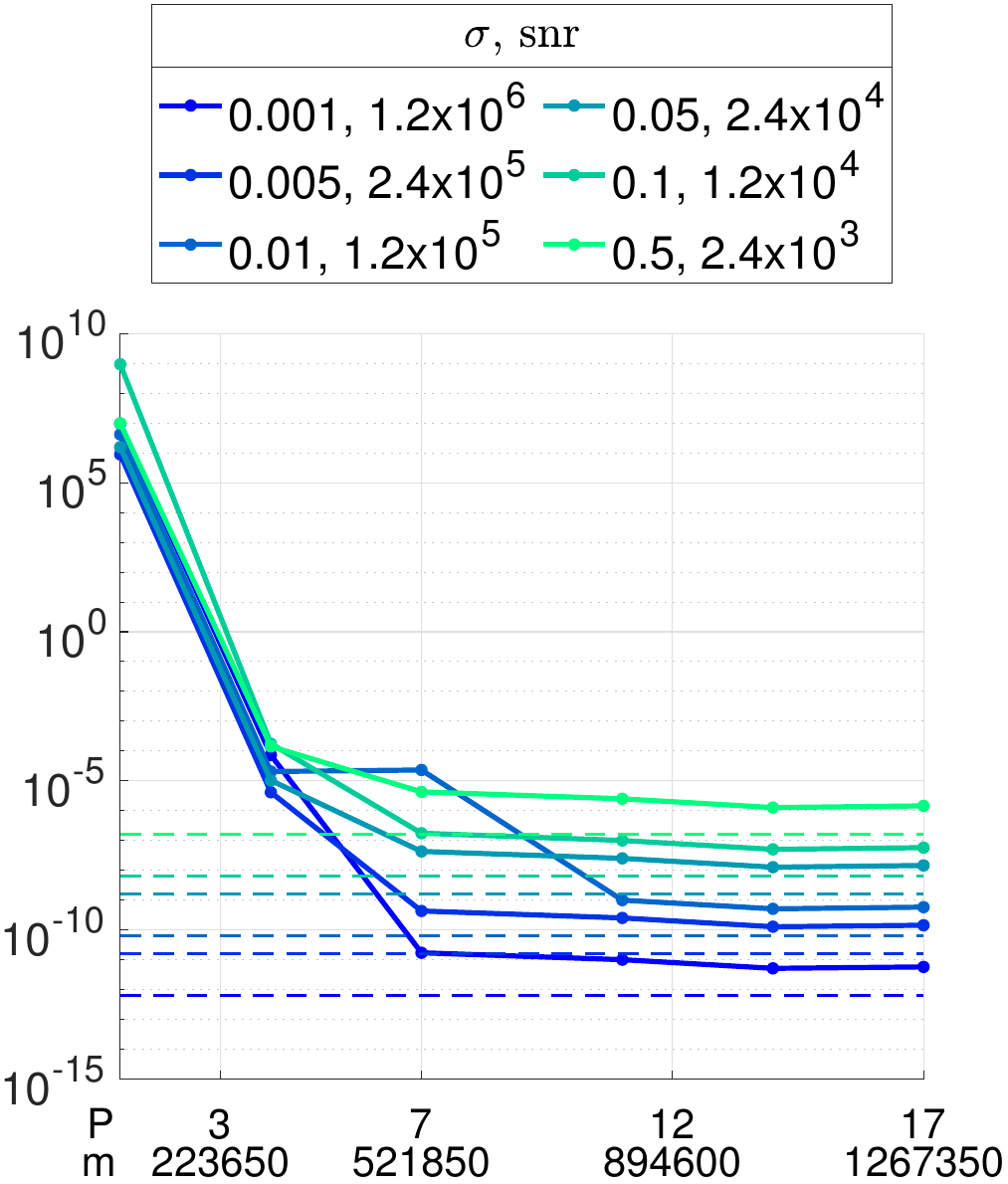}
        \caption*{(c1)}
    \end{subfigure}
    \begin{subfigure}{\textwidth}
        \includegraphics[width=\textwidth]{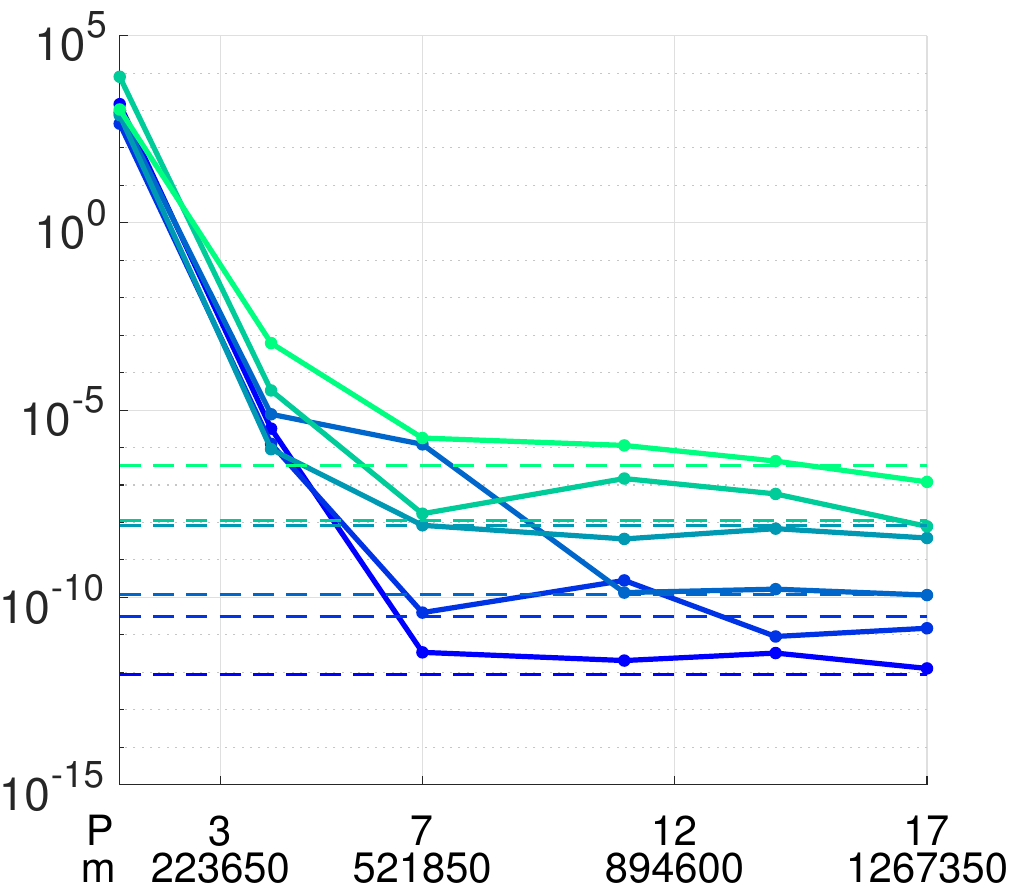}
        \caption*{(c2)}
    \end{subfigure}
    \vspace{-0.5cm}
    \caption*{(c) \small$\{\bv_i\}$ with overlapping supports. Vary $m$ via $P$; $R=10r\tk$.}
    \label{subfig:noisy_varyP_overlap}
\end{minipage}
\vspace{-0.1cm}
\caption{\small \edit{Normalized squared error of the matrix  $\|\hbX_0-\bX_0\|^2_F/{\|\bX_0\|^2_F}$ (top row) and average normalized square error of the eigenvectors $\frac{1}{r}\sum_{i=1}^r \|\hat{\bv}_i-\bv_i\|^2_2/{\|\bv_i\|^2_2}$  (bottom row) plotted against the sketch size $\nMeas$, for different SNRs $\lambda/(\sqrt{n} \sigma)$.   
Solid lines: proposed algorithm with message passing   in stage B.
Dotted lines: proposed algorithm with peeling decoding  in stage B. 
Dashed lines: direct eigendecomposition of the noisy matrix and retain the contribution of the top $r$ eigenvectors. 
In all cases,  \edit{$n=4000$, $k=70$} and $r=3$.   The parameters in the zeroton and singleton tests are $\gamma_0=5$ and \edit{$\gamma_1=2.5$ in (a)--(b) and $\gamma_1=3.5$} in (c).  Results are averaged over 15 trials.}}\label{fig:noisy}
\end{figure}
We apply the robust sketching scheme and recovery algorithm to symmetric matrices $\bX = \bX_0 +\bW$, with $\bX_0 = \sum_{i=1}^r  \bv_i\bv_i^T$ and $\bW$ a random  symmetric \edit{noise} matrix. 
 \edit{Here, we have dropped the  $ \, \tilde{} \, $ notation  for simplicity to let $\bv_1, \bv_2, \dots, \bv_r$ denote the unnormalized eigenvectors. Throughout this section, the column weight of the parity check matrix $\bH$  is chosen to be $d=2$.} 

Each  $\bv_i$ is 
 $k$-sparse with its nonzero entries drawn uniformly at random  from a discrete alphabet $\{\pm 10, \pm 20, \dots, \pm 50\}$. When the supports of $\{\bv_i\}$ overlap, we adjust in each $\bv_i$ a small number of nonzero entries  to ensure orthogonality among the $\{\bv_i\}$. 
 The resulting eigenvalues of $\bX_0$ are 
 $\lambda_i = \|\bv_i\|^2$ for $i\in [r]$. We measure the signal strength by the expected average eigenvalue $\lambda: = \E\left[\frac{1}{r} \sum_{i=1}^r \lambda_i\right]$. The entries in the noise matrix $\bW$ are drawn i.i.d. up to symmetry from $\normal(0, \sigma^2)$.  Let $\lambda_{W,1}\ge \lambda_{W,2}\ge \dots\ge \lambda_{W,n}$ denote the eigenvalues of $\bW$.  From Wigner's semicircle law \cite{anderson2009introduction}, we have that $\max_{i\in [r]} \frac{\lambda_{W,i}}{\sqrt{n}}$ concentrates around $2 \sigma$ and that 
$
\lim_{n\to \infty}\E\left[\frac{1}{n^2} \sum_{i=1}^n \lambda_{W,i}^2\right] =  \sigma^2.
$
Thus, to ensure the right scaling in our signal to noise ratio (SNR), we define our SNR to be $\lambda / (\sqrt{n} \sigma)$.

\edit{Let $\hat{\bv}_1, \hat{\bv}_2, \dots, \hat{\bv}_r$ denote the  unnormalized eigenvectors recovered by the algorithm, and $\hbX_0$ the estimated signal matrix computed using  $\{\hat{\bv}_i\}$. Fig. \ref{fig:noisy} shows the normalized squared error (NSE) of the recovered matrix  (defined as $\|\hbX_0-\bX_0\|^2_F/{\|\bX_0\|^2_F}$)  and the average NSE of the recovered eigenvectors (defined as $\frac{1}{r}\sum_{i=1}^r \|\hat{\bv}_i-\bv_i\|^2_2/{\|\bv_i\|^2_2}$)  against the sketch size $m$, for different SNR values.  The three columns in Fig. \ref{fig:noisy} correspond to three different settings, as indicated by the captions. In each column, the top subfigure plots the matrix NSE and the bottom subfigure plots the eigenvector NSE. }

\edit{In each subfigure, the solid lines correspond to the proposed algorithm with message passing decoding in stage B. The dotted lines trailing along the solid lines correspond to the  algorithm with peeling decoding in stage B. As we can see, message passing  consistently outperforms peeling decoding, offering slightly better noise robustness. Moreover, the horizontal dashed lines indicate the NSE of the baseline method without sketching -- this directly computes the eigendecomposition of the  matrix $\bX$ and estimates the signals using  the top $r$ eigenvectors. (The sample complexity of the baseline method is the number of distinct entries in  $\bX$, i.e., $ \tilde{n} ={n \choose 2} + n$.) Our NSEs approach the baseline NSEs as $m$ increases, especially in the overlapping support case (Fig. \hyperref[subfig:noisy_varyP_overlap]{8c}). For accurate recovery, that is, in the settings corresponding to the first data points after the transition in the plots, the required sketch size $m$ is at least an order of magnitude smaller than $\tilde{n}$; the running time of the recovery algorithm is also noticeably smaller than the baseline method.
}




Recall that the sketch size is $m=RP$, where $R$ is the number of bins and $P$ is the number of components in each bin. 
In Fig. \hyperref[subfig:noisy_varyR_disjoint]{8a}, the sketch size is varied via $R$, with $P$ held constant. As $R$ increases,  more  nonzero matrix entries are recovered in the first stage of the algorithm. This increases the number of right nodes in the graph for stage B, allowing more nonzero eigenvector entries to be recovered which leads to   smaller NSE. In Figs. \hyperref[subfig:noisy_varyP_disjoint]{8b} and \hyperref[subfig:noisy_varyP_overlap]{8c}, the sketch size is varied via $P$, with the number of bins $R$  held constant. The value of $P$ is varied as $\ceil{c\log (n/k)}$ for 
\edit{$c\in [0.1, 4.1]$}. The accuracy of the recovery improves as $P$ increases due to more reliable zeroton and singleton detection. 
In all subfigures, we observe that the NSEs
decrease with the SNR as expected.

\subsubsection{\edit{Comparison with alternative sketching schemes}}
\label{sec:noisy_comparison}
\edit{
We compare our scheme to two other existing methods. The first one uses  an i.i.d. Gaussian sketching operator $\mc{A}$ and a convex program  similar to \eqref{eq:nl_oymak_obj} for recovery,  with an inequality constraint to account for the noise. In particular, to recover a PSD matrix $\bX_0$ from the sketch $ \mc{A}(\bX)$ where $\bX= \bX_0 + \bW$, we solve the following optimization problem:
\beq
\hat{\bX}_0 \in \argmin_{\bar{\bX}\succeq 0} \|\bar{\bX}\|_* + \lambda\|\bar{\bX}\|_1\quad \text{subject to}\quad \|   \mc{A}(\bar{\bX}) - \mc{A}(\bX) \|_2 \le \varepsilon_{n, m}\;\;,\label{eq:ny_oymak_obj}
\eeq
where $\varepsilon_{n, m}$ is an upper bound on the $\ell_2$-norm of the noise vector $\bz := \mc{A}(\bW)$. Recalling that $W_{ij}\stackrel{\text{iid}}{\sim} \normal(0, \sigma^2)$ for $i\le j$, we can deduce that $\E[\|\bz\|_2^2]= mn^2\sigma^2.$ Thus we set $\varepsilon_{n,m}^2 =  1.2mn^2\sigma^2$ in our experiments. Like in the noiseless case in Section \ref{subsec:numerical_results} (Fig. \ref{fig:noiseless_oymak}), we use  $\lambda=0.25$, and solve \eqref{eq:ny_oymak_obj} using CVX which calls the solver MOSEK.}

\begin{figure}[!b]
  	\begin{subfigure}[H]{\linewidth}
  	\centering
    \includegraphics[width=0.38\linewidth]{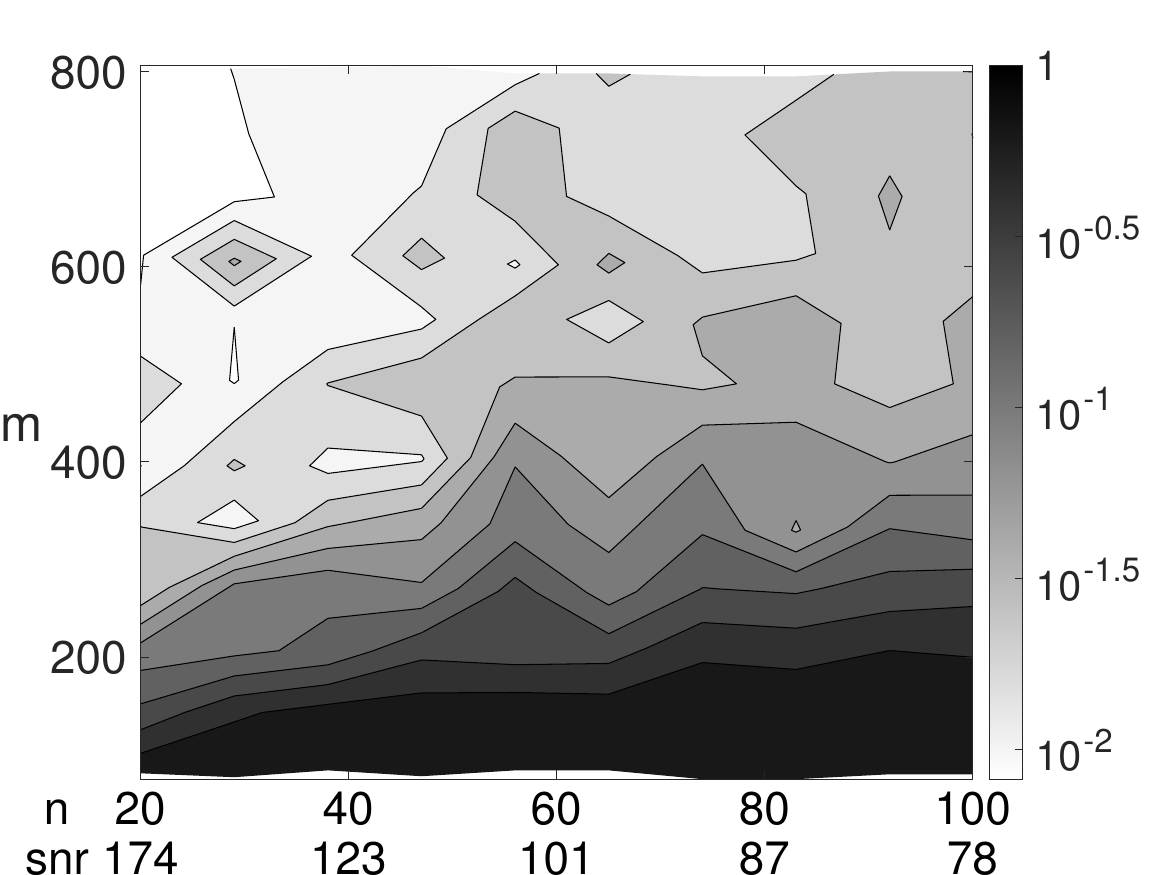}\hspace{0.5cm}
    \includegraphics[width=0.38\linewidth]{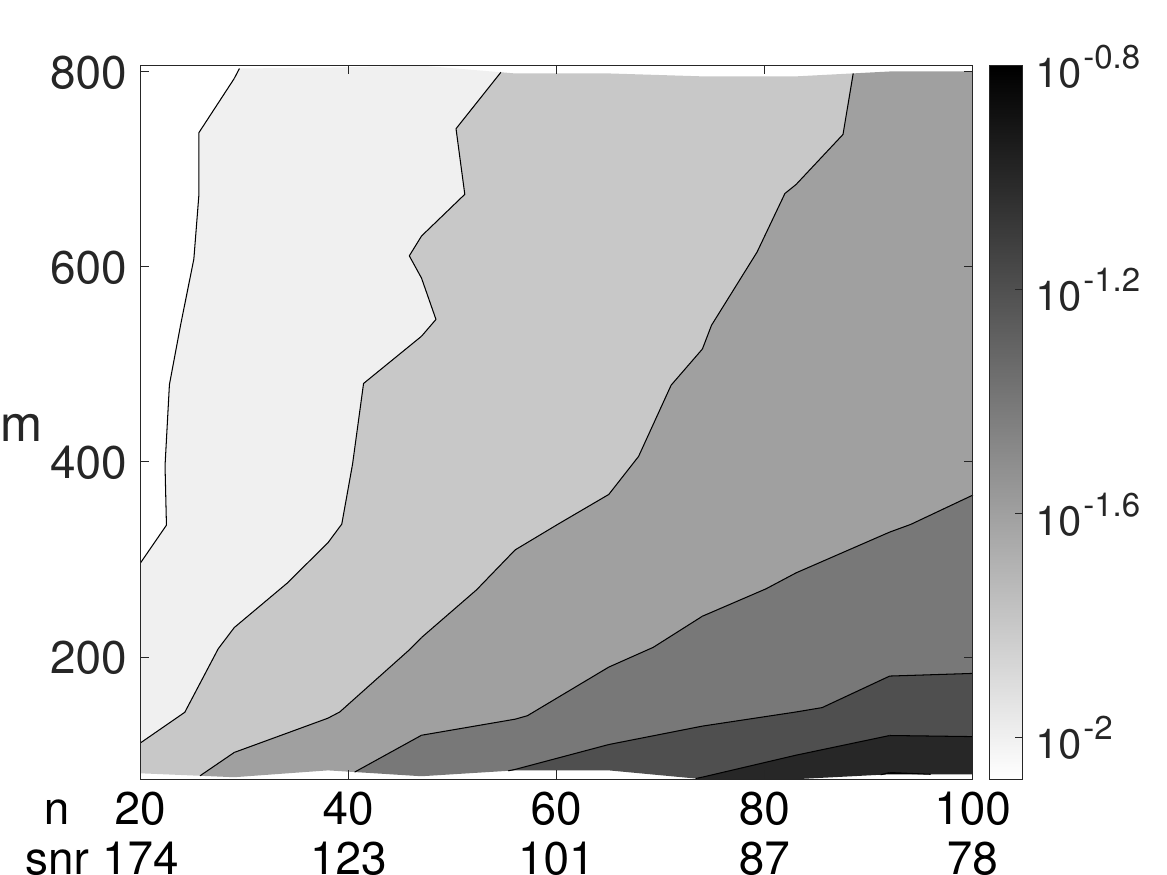}
    \vspace{-0.1cm}
    \caption{$\{\bv_i\}$ with disjoint supports.} \label{subfig:ny_oymak_disjoint_gaussian}
\end{subfigure}
\begin{subfigure}[H]{\linewidth}
\centering
    \includegraphics[width=0.38\linewidth]{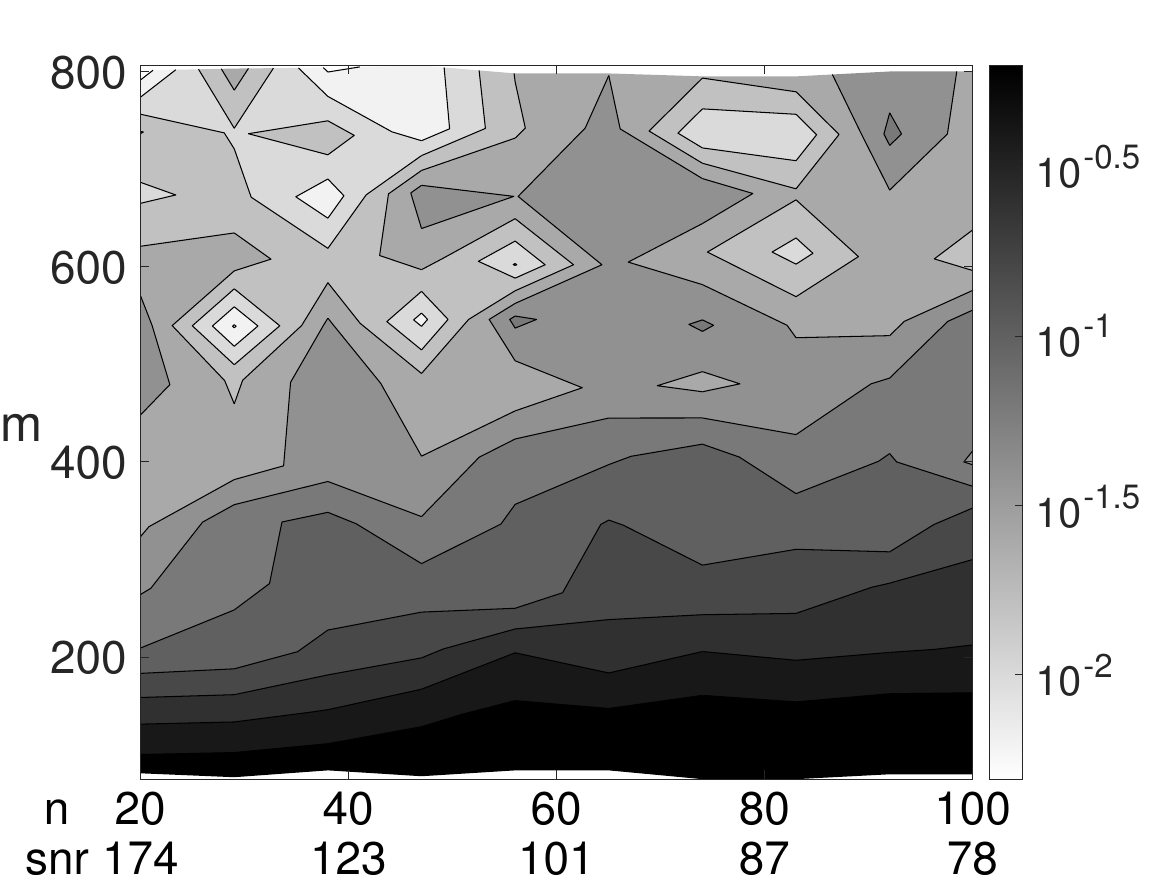}\hspace{0.5cm}
    \includegraphics[width=0.38\linewidth]{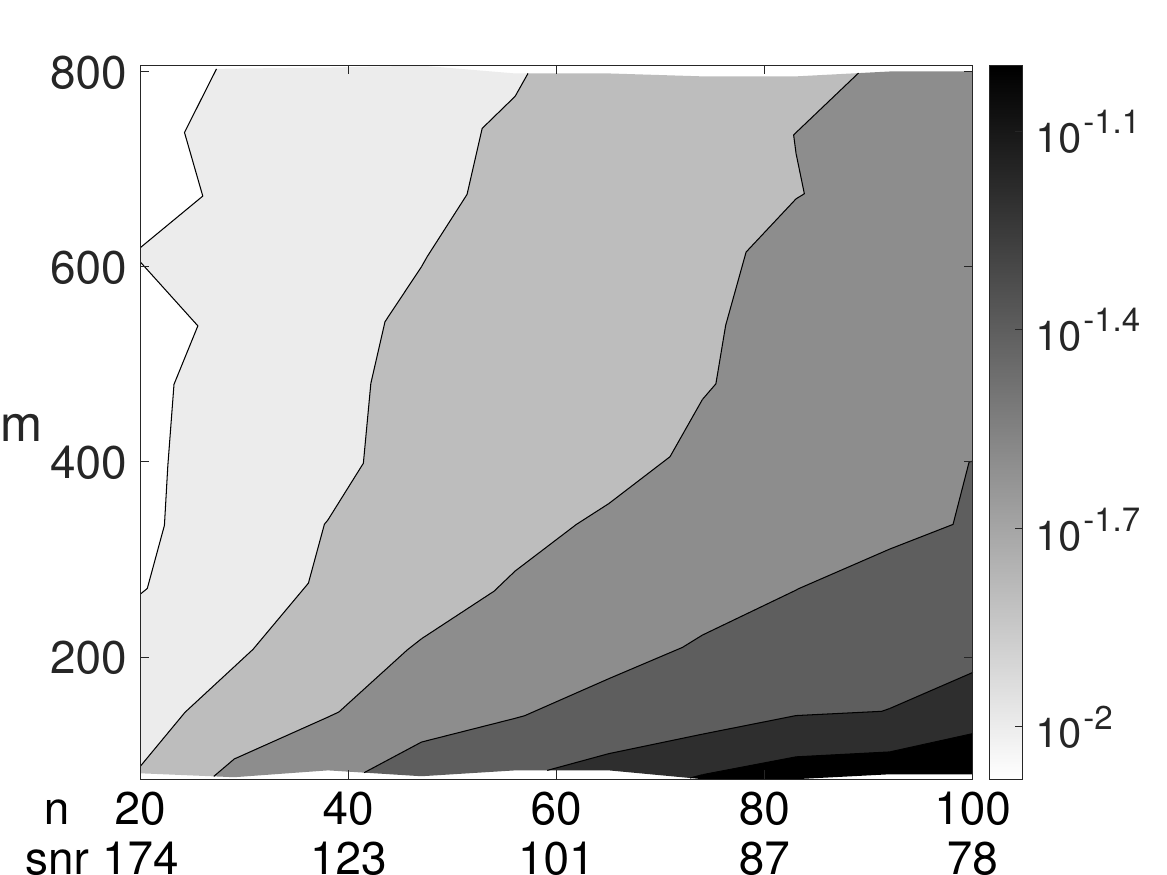}
     \vspace{-0.1cm}
	\caption{$\{\bv_i\}$ with overlapping supports.}
   \label{subfig:ny_oymak_overlap_gaussian}
  \end{subfigure} 
 %
 %
\vspace{-0.1cm}
  	\caption{\small \edit{Normalized error $\|\hat{\bX}_0-\bX_0\|_F/\|\bX_0\|_F$   on $\log_{10}$ scale. Left column: our scheme.  Right column: conventional scheme  using a  Gaussian sketching operator and recovery via the convex program in  \eqref{eq:ny_oymak_obj}. 
Each heatmap shows the normalized error across different  values of ambient dimension $n$ ($x$-axis) and  number of measurements $m$ ($y$-axis).  The matrices $\bX_0$ are  rank-$2$ PSD matrices, with fixed sparsity level $k=3$. The noise standard deviation is $\sigma=0.1$ and SNR is $\lambda/(\sqrt{n}\sigma)$. $P=\ceil{4.5\log(n/k)}$, $\gamma_0=5$ and $\gamma_1=2.5.$ Results are averaged over 30  trials.}}\label{fig:ny_oymak}
\end{figure}

\edit{In Fig. \ref{fig:ny_oymak}, we use  rank-$2$ sparse PSD matrices $\bX_0 = \sum_{i=1}^2 \bv_i\bv_i^T$. The   nonzero entries of each $\bv_i$ are independently drawn from a mixture of Gaussians  $\frac{1}{2}\normal(-5,1)+\frac{1}{2}\normal(5,1)$. 
The noise standard deviation is  $\sigma=0.1$ and the SNR is calculated as $\lambda/(\sqrt{n} \sigma)$ as before. For a fixed sparsity level $k=3$, the heatmaps indicate the normalized error across different values of ambient dimension $n$ and sketch size $m$. We vary $m$ via $R$, fixing $P$ as $\ceil{4.5\log(n/k)}$ for each  $n$ value.}
\edit{Our experiments are restricted to small matrices, as the CVX solver for the convex program \eqref{eq:ny_oymak_obj}  is infeasibly slow for matrices bigger than $100\times 100$.}

\begin{figure}[!t]
\centering
\begin{subfigure}[H]{0.38\linewidth}
    \includegraphics[width=\linewidth]{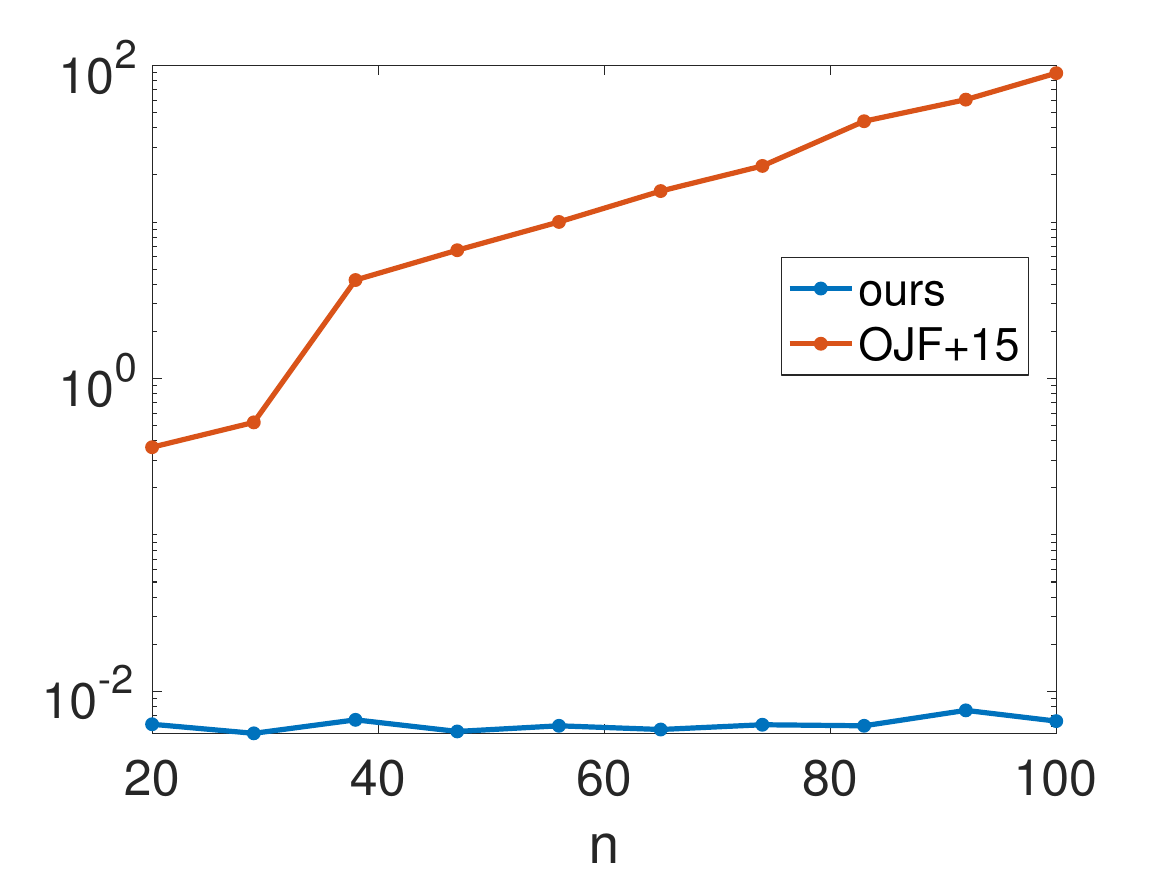}
    \caption{Same setting as in Fig. \ref{subfig:ny_oymak_disjoint_gaussian}. Fix $R=42$, vary $n$ and  set $P=\ceil{4.5\log(n/k)}$. }
\end{subfigure}
    \hspace{0.5cm}
\begin{subfigure}[H]{0.38\linewidth}
    \includegraphics[width=\linewidth]{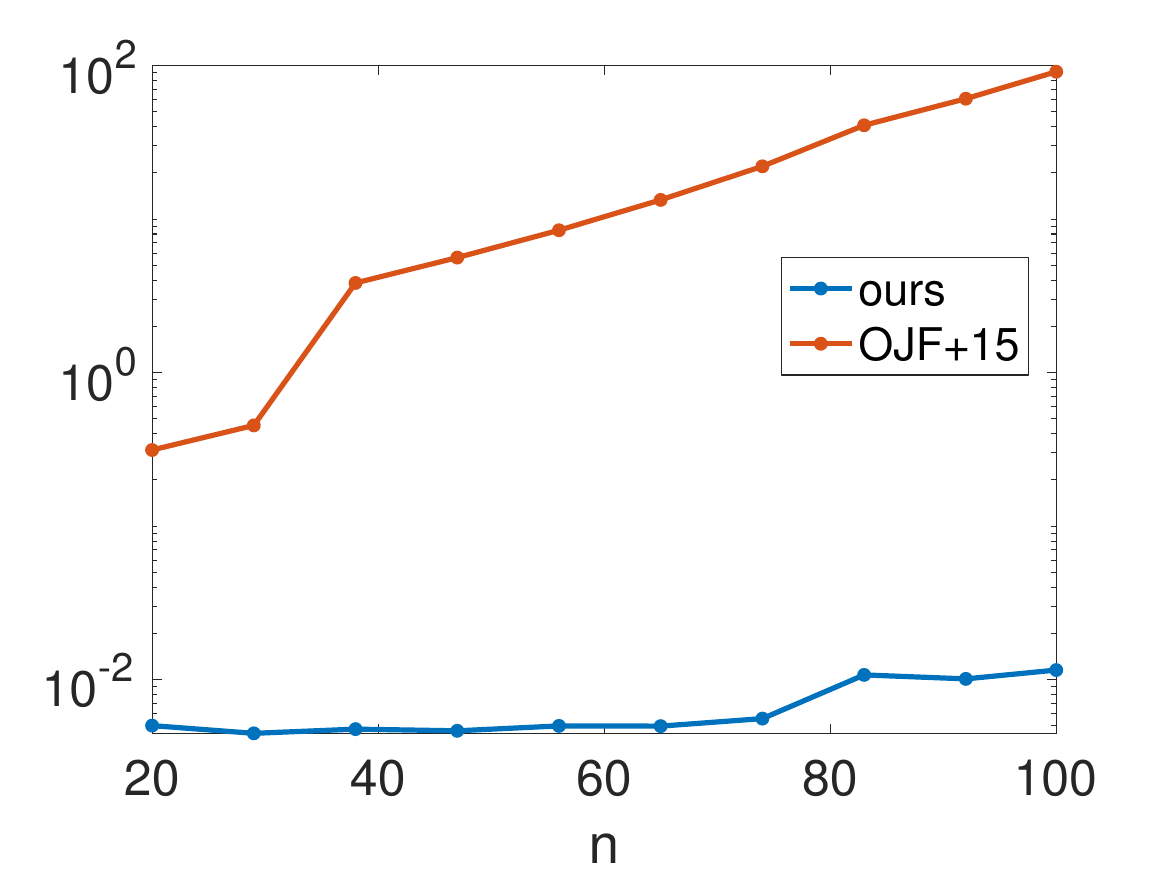}
    \caption{Same setting as in Fig. \ref{subfig:ny_oymak_overlap_gaussian}. Fix $R=42$, vary $n$ and set $P=\ceil{4.5\log(n/k)}$.}
\end{subfigure}
\caption{\small \edit{Running time of our scheme and the conventional scheme using Gaussian sketching operator and recovery via the convex program \eqref{eq:ny_oymak_obj}. Same setting as in Fig. \ref{fig:ny_oymak}. Each plot shows the running time against $n$, for fixed $k=3$ and $R=42.$ }} \label{fig:ny_oymak_runtime}
\end{figure}



\edit{The left column of Fig. \ref{fig:ny_oymak} corresponds to our scheme and the right  to  the optimization approach. We observe a similar trade-off as in the noiseless case (Fig. \ref{fig:noiseless_oymak}):  the conventional scheme is more sample efficient in the relatively dense regime, but is outperformed by our scheme in the sparse regime given a sufficient number of measurements. Specifically, the normalized error in the top right region of the heatmaps is smaller under our scheme. 
It is worth noting that the contours in the heatmaps  are significantly flatter under our scheme, indicating better scalability. The shallow  gradient of the contours under our scheme is due to the $\log (n/k)$ factor in our sample complexity $m=\bigo(rk^2\log (n/k))$. } 

\edit{
Moreover, similar to the noiseless case,  Fig. \ref{fig:ny_oymak_runtime} shows that the running time of our algorithm grows much more slowly than the convex optimization method, and is several orders of magnitude faster across different $n$ values.}


\edit{The second scheme that we compare with is the nested  scheme proposed in \cite{bahmani2015sketching}, \cite[Sec. 3.3]{bahmani2016nearopt} for sketching PSD matrices. The sketching operator takes the form 
$\mc{A}: \bX \to \left[\langle \ba_i\ba_i^T, \bX\rangle\right]_{i=1}^m$ where $\ba_i = \bPsi^T \be_i\in \reals^{n}$, with $\bPsi \in \reals^{L\times n} $ populated with entries  $\stackrel{\text{iid}}{\sim}\normal(0, 1/L)$, and $\be_1, \dots, \be_m\in \reals^L\stackrel{\text{iid}}{\sim} \normal(\bzero, \bI_{L\times L})$. 
That is, each measurement takes the form
 \begin{align}
 y_i & = \langle \ba_i\ba_i^T, \bX\rangle
 = \ba_i^T \bX \ba_i 
 = \be_i^T \bPsi \bX \bPsi^T \be_i
 = \langle \be_i \be_i^T, \bPsi \bX\bPsi^T \rangle
 =: \mc{E}(\bPsi \bX\bPsi^T), \quad i\in [m]. \label{eq:bah_y_i}
 \end{align}
Note that $\mc{A}$ has a nested structure, consisting of a linear operator $\mc{E}$  and a matrix $\bPsi$. With appropriately chosen $L$ and $m$, $\mc{E}$
 forms a restricted isometry for low-rank matrices, and $\bPsi$ forms a restricted isometry for sparse matrices.  This structure enables the following two-stage  algorithm to estimate $\bX_0$  \cite{bahmani2016nearopt, bahmani2015sketching}:
 \begin{align}
\text{Low-rank estimation stage:}& \quad    
\hat{\bB} \in \argmin_{\bB\succeq 0} \trace(\bB) \quad
\text{subject to}\quad \sum_{i=1}^m (\be_i^T\bB\be_i -y_i)^2 \le \varepsilon_{n,m}^2. \label{eq:bah_low_rank}\\
\text{Sparse estimation stage:}& \quad 
\hat{\bX}_0 \in \argmin_{\bar{\bX}} \|\bar{\bX}\|_1 \quad
\text{subject to} \quad 
\|\bPsi \bar{\bX}\bPsi^T - \hat{\bB}\|_F 
\le \frac{C\varepsilon_{n,m}}{\sqrt{m}} .
\label{eq:bah_sparse}
 \end{align}
}

\begin{figure}[!b]
\centering
\begin{subfigure}[H]{0.38\linewidth}
    \centering
    \includegraphics[width=\linewidth]{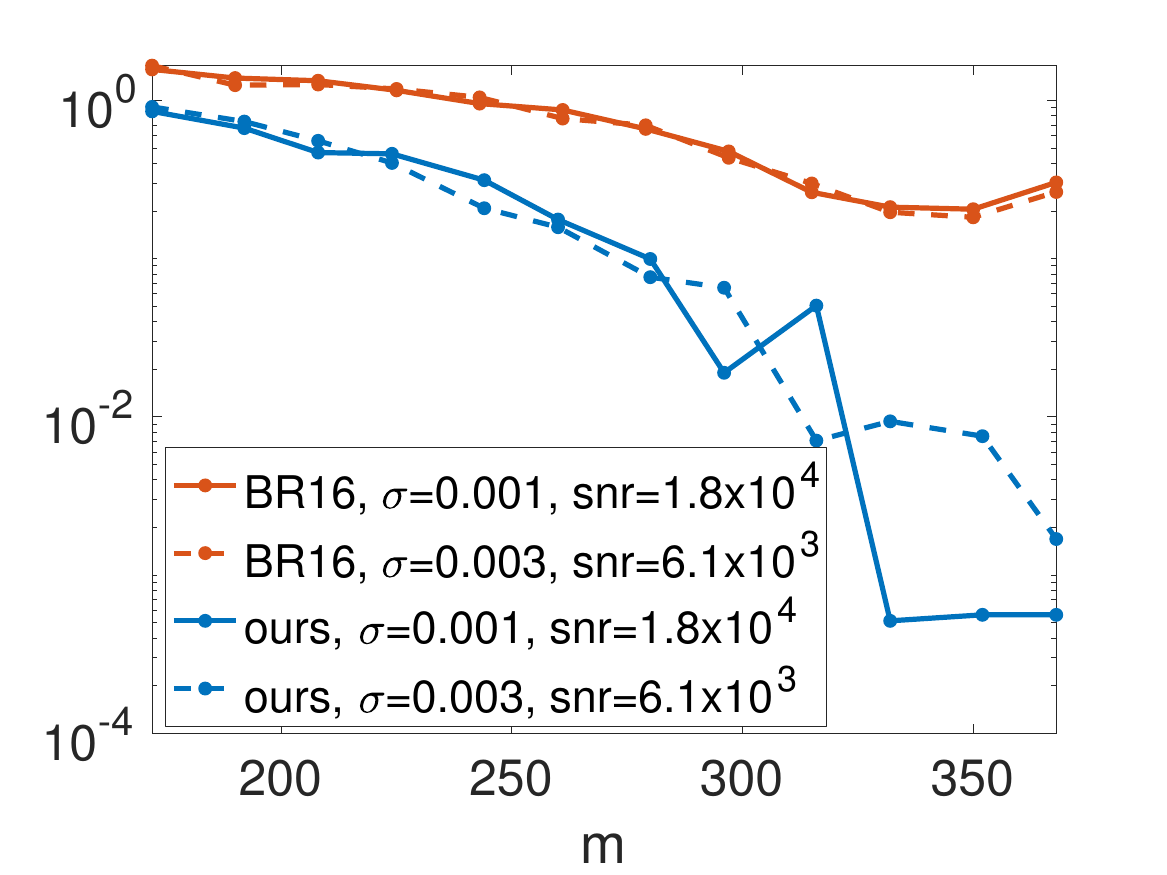}
    \vspace{-0.5cm}
    \caption{Gaussian alphabet; $k=7$. $P=\ceil{1.5\log(n/k)}$, $L=\ceil{3k(1+\log(n/k))}$.}
    \label{subfig:ny_bah_K7_gaussian}
\end{subfigure}
\hspace{0.5cm}
\begin{subfigure}[H]{0.38\linewidth}
    \centering
    \includegraphics[width=\linewidth]{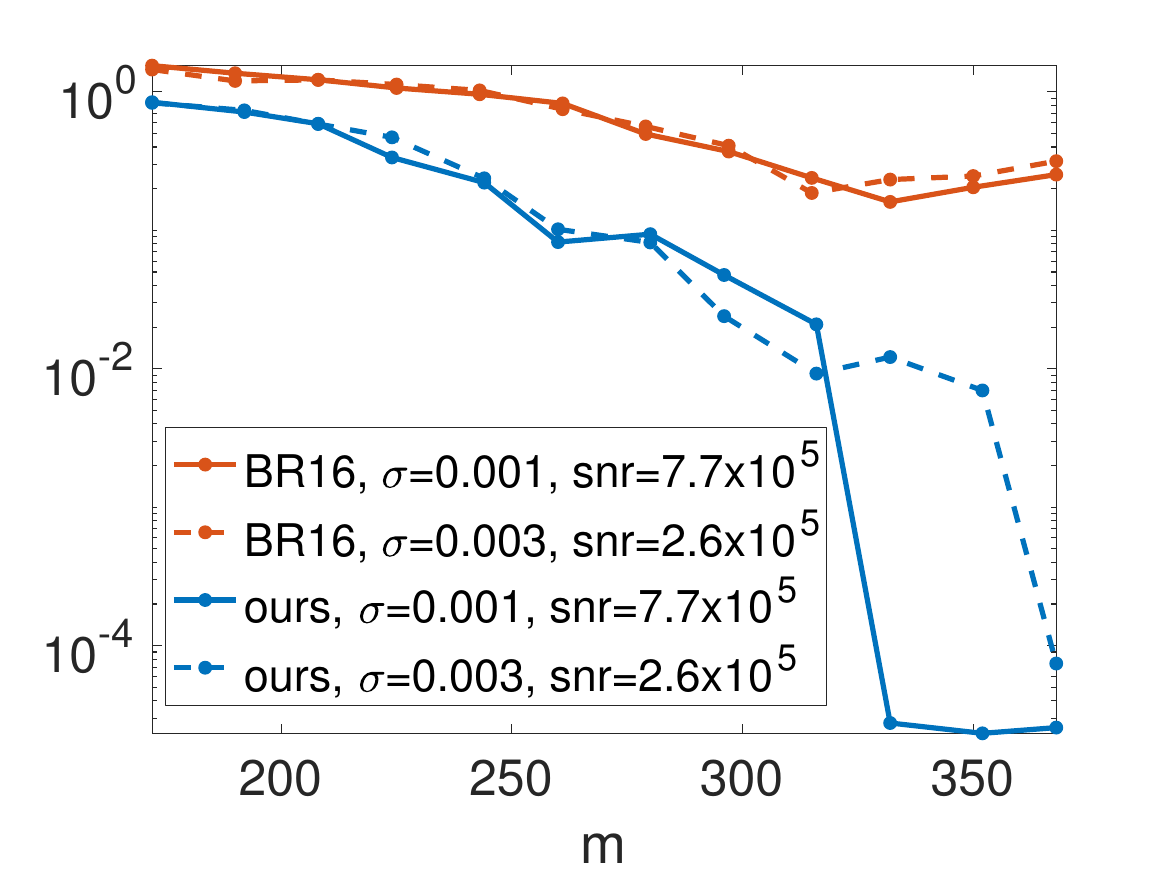}
    \vspace{-0.5cm}
    \caption{Discrete alphabet; $k=7$. $P=\ceil{1.5\log(n/k)}$, $L=\ceil{3k(1+\log(n/k))}$.}
    \label{subfig:ny_bah_K7_rhoB}
\end{subfigure}
\begin{subfigure}[H]{0.38\linewidth}
    \centering
    \includegraphics[width=\linewidth]{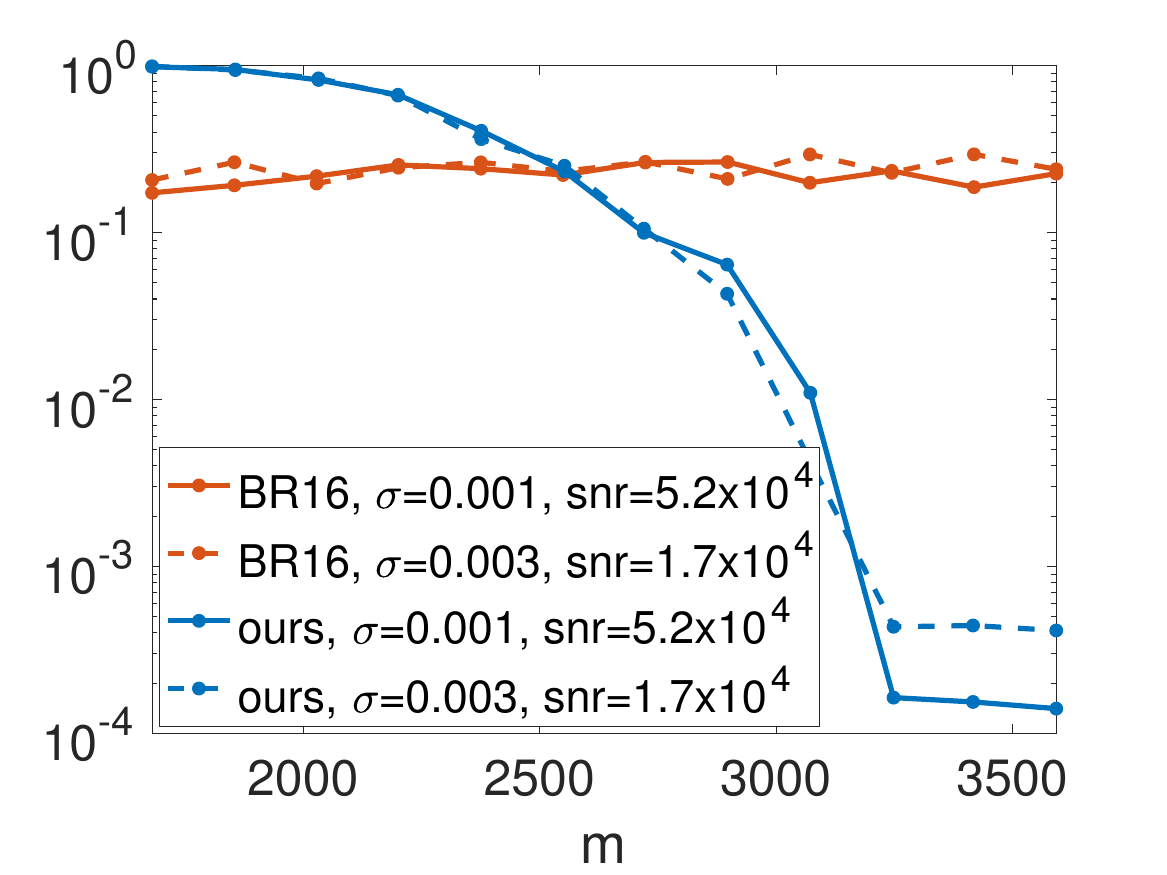}
    \vspace{-0.5cm}
    \caption{Gaussian alphabet; $k=20$. $P=\ceil{4.5\log(n/k)}$, $L=\ceil{1.8k(1+\log(n/k))}$.}
    \label{subfig:ny_bah_K20_gaussian}
\end{subfigure}
\hspace{0.5cm}
\begin{subfigure}[H]{0.38\linewidth}
    \centering
    \includegraphics[width=\linewidth]{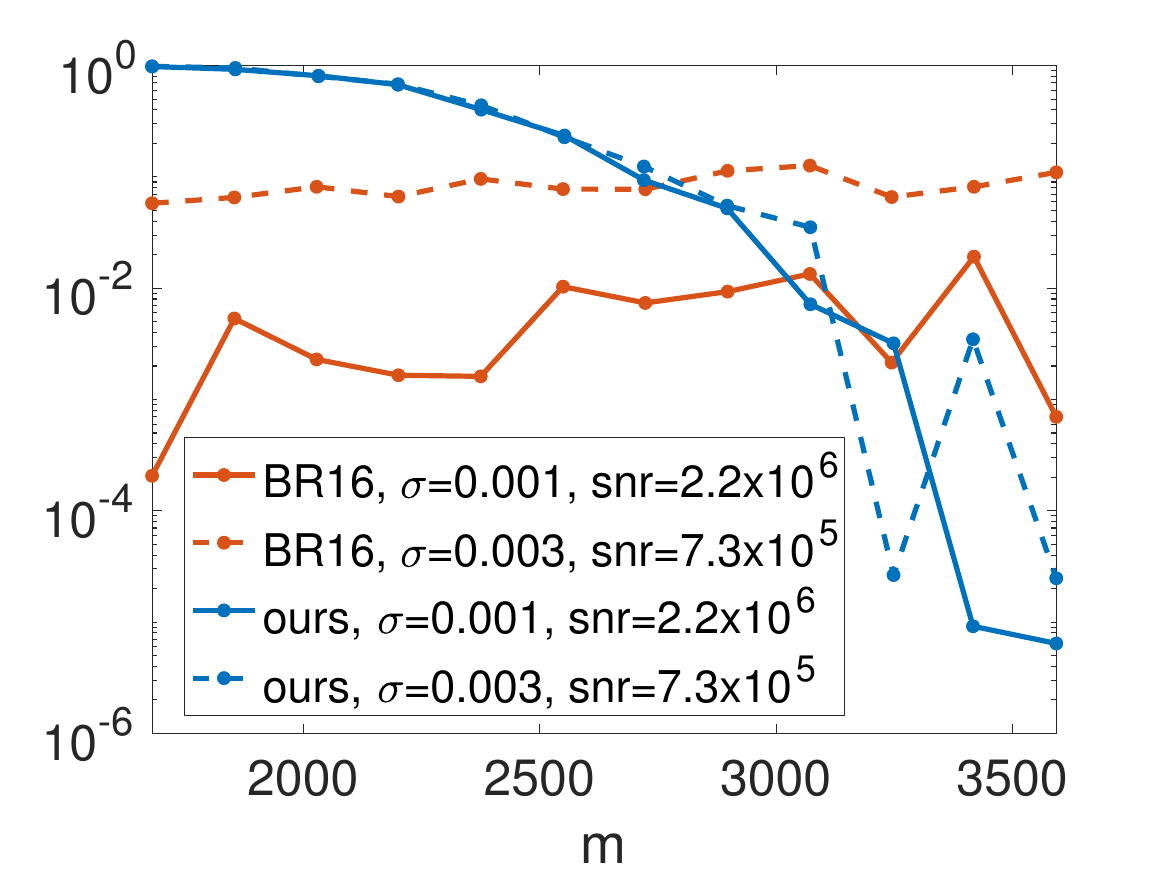}
    \vspace{-0.5cm}
    \caption{Discrete alphabet; $k=20$. $P=\ceil{4.5\log(n/k)}$, $L=\ceil{1.8k(1+\log(n/k))}$.}
    \label{subfig:ny_bah_K20_rhoB}
\end{subfigure}  
\caption{\small \edit{Normalized error $\|\hat{\bX}_0-\bX_0\|_F/\|\bX_0\|_F$ of our scheme (red lines) and the nested sketching scheme proposed in \cite{bahmani2015sketching, bahmani2016nearopt} (blue lines) plotted against $m$. In each setting,   $\bX_0$ is a $100\times 100$ rank-$2$ PSD matrix, whose eigenvectors have disjoint supports.  The solid and dashed lines correspond to $\sigma=0.001$ and $\sigma=0.003$. SNR is $\lambda/(\sqrt{n}\sigma)$. $\gamma_0=5, \gamma_1=2.5.$
Results are averaged over $60$ trials. } }
\label{fig:ny_bah}
	\vspace{-2pt} 	
\end{figure}

\begin{figure}[!t]
\centering
\begin{subfigure}[H]{0.38\linewidth}
    \centering
    \includegraphics[width=\linewidth]{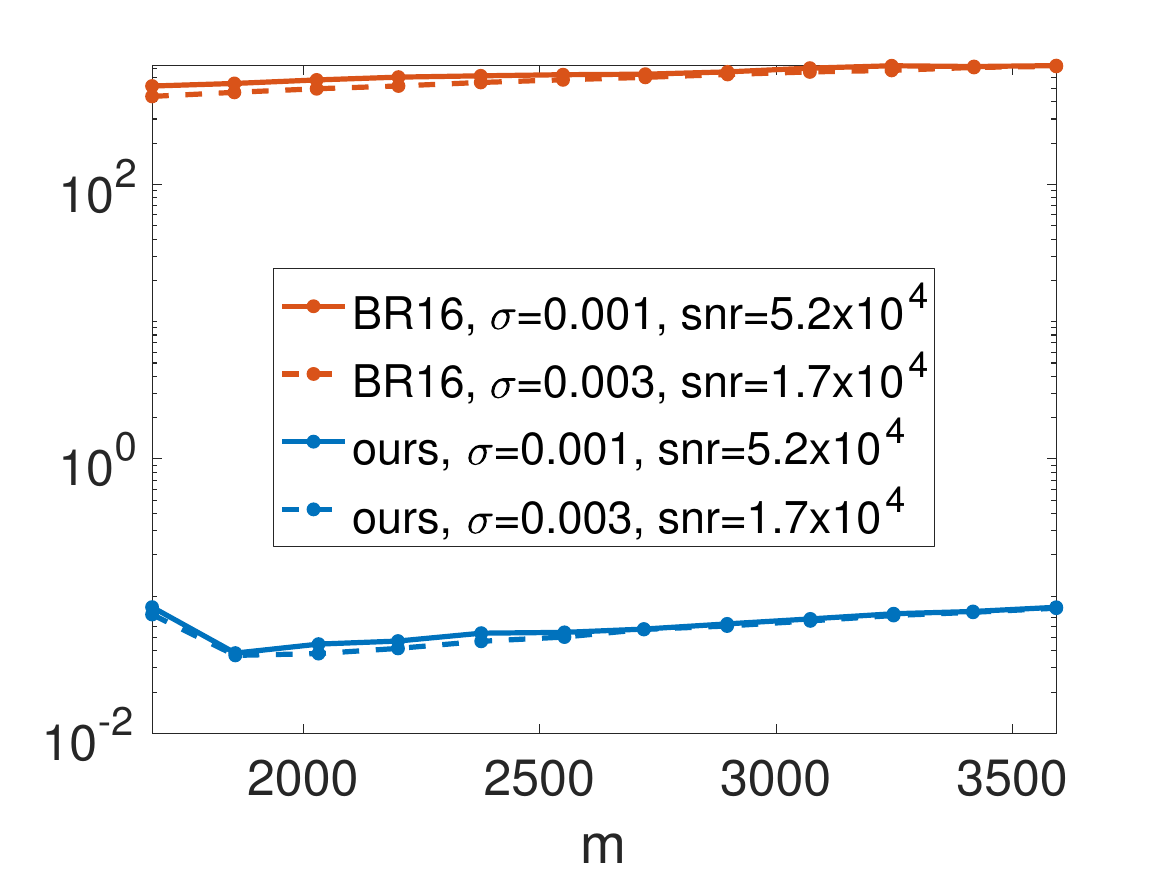}
    \vspace{-0.5cm}
    \caption{Gaussian alphabet; $k=20$.}
\end{subfigure}
    \hspace{0.5cm}
\begin{subfigure}[H]{0.38\linewidth}
    \includegraphics[width=\linewidth]{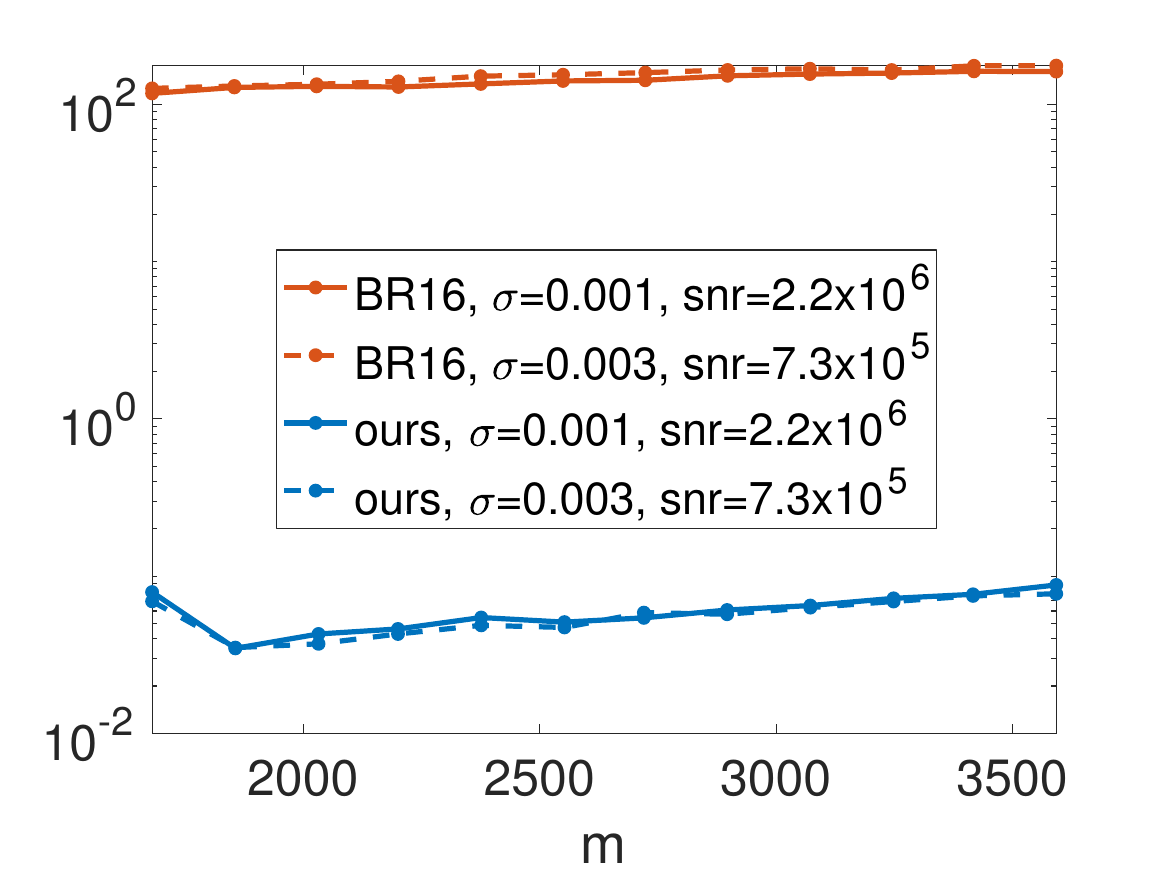}
\vspace{-0.5cm}
    \caption{Discrete alphabet; $k=20$.}
\end{subfigure}
\caption{\small \edit{Running time of our scheme (red lines)  and the nested sketching scheme (blue lines), recorded in the experiments for  Figs. \ref{subfig:ny_bah_K20_gaussian}--\ref{subfig:ny_bah_K20_rhoB}. 
Results are  averaged over 60 trials. }}
\label{fig:ny_bah_runtime}
	\vspace{-2pt} 	
\end{figure}

\edit{Here $C$ is a suitable constant, and $\varepsilon_{n,m}$ is an upper bound on the $\ell_2$-norm of the noise vector $\bz$, whose entries take the form $z_i =\be_i^T \bPsi \bW \bPsi^T \be_i$ for $ i\in [m]$. Omitting the details,  we use $\varepsilon_{n,m}=1.5 \sqrt{\E[\|\bz\|_2^2]} = 1.5\sigma \sqrt{m n(n+2)(1+2/L)}$ in our experiments. }

 \edit{Fig. \ref{fig:ny_bah} compares the performance of our scheme with the nested sketching scheme for two different choices for the nonzero entries of the eigenvectors. These entries are drawn independently from either: (i) the Gaussian mixture distribution $\frac{1}{2}\normal(-5,1)+\frac{1}{2}\normal(5,1)$ or (ii) the uniform distribution over the discrete alphabet $\{\pm 10, \pm 20, \dots, \pm 50\}$.  The matrix $\bX_0$ is  $100\times 100$ rank-$2$ PSD, with eigenvectors having disjoint supports.  We observe that our scheme consistently outperforms the nested scheme of \cite{bahmani2016nearopt, bahmani2015sketching} in the  sparse setting (Figs. \ref{subfig:ny_bah_K7_gaussian}--\ref{subfig:ny_bah_K7_rhoB}).  In the denser setting (Figs. \ref{subfig:ny_bah_K20_gaussian}--\ref{subfig:ny_bah_K20_rhoB}), our scheme needs more measurements to go through phase transition in normalized error. Nevertheless, it achieves an error that's $2$ to $4$ orders of magnitude smaller than the latter after the transition. We note that the normalized error of the nested scheme is on the order of $10^{-1}$, which agrees with the simulation results in \cite[Sec. III]{bahmani2015sketching}. }

 \edit{The nested sketching scheme was implemented according to the description in  \cite[Sec. 4]{bahmani2016nearopt}. We used TFOCS \cite{tfocs} for the low-rank estimation stage \eqref{eq:bah_low_rank} and a variant of the Alternating Direction Method of Multipliers adapted from \cite{yang2011alternating,yall1} for the sparse estimation stage \eqref{eq:bah_sparse}. This implementation is significantly faster than that via CVX, but as shown in  Fig. \ref{fig:ny_bah_runtime}, is still on average $3$ to $4$ orders of magnitude slower than our algorithm.} 


\section{Proofs of main lemmas} \label{sec:main_lem_proofs}

In this section, we give the proofs of Lemmas \ref{lem:1A_rank_r}--\ref{lem:non_sym_Stage_B}.

\subsection{Preliminaries} \label{subsec:prelim}
We begin with the following tail bound for binomial random variables, which follows from  Sanov's theorem \cite[Theorem 12.4.1]{cover2006elements}.
  \begin{lem}[Tail bound for Binomial]
\label{lem:binom_tail_bound} For $X\sim \bin{n}{p}$,
\begin{align}
\prob{X \leq x} & \leq \exp\left(- nD_e\left(\frac{x}{n}\middle\|p\right)\right)\ \text{ for } 0\leq x\leq np,  \\
\prob{X \geq x} &\leq 
 \exp\left(- nD_e\left(\frac{x}{n}\middle\|p\right)\right)\  \text{ for } np\leq x\leq n,
\end{align}
 where $D_e(q\|p) := q \ln \frac{q}{p} + (1-q) \ln \frac{1-q}{1-p}$.
 \end{lem}

We will use the notion of \emph{negative association} (NA) \cite{joagDev1983negative, dubhashi1998balls} to handle the dependence between the degrees of the bipartite graphs we analyze. Roughly speaking,  a collection of random variables $X_1, X_2, \ldots, X_n$  are negatively associated if whenever a subset of them is high, then a disjoint subset of them must be low.
\begin{definition}[Negative Association (NA)] \label{def:NA}
The random variables $X_1, X_2, \dots  , X_n$  are said to be negatively associated if for any two disjoint index sets $\mc{I, J} \subseteq [n]$ and two functions $f: \reals^{ \abs{\mc{I}} }\to \reals$, 
$g: \reals^{ \abs{\mc{J}} }\to \reals$ both monotonically increasing or monotonically decreasing, 
\beq
\E\left[ f(X_i, i\in \mc{I}) g(X_j, j\in \mc{J})\right]
\le \E\left[f(X_i, i\in \mc{I})\right] \E\left[ g(X_j, j\in \mc{J})\right].
\eeq
\end{definition} 

\begin{lem}[Useful properties of NA \cite{joagDev1983negative, dubhashi1998balls}]\text{ } \\
\vspace{-0.5cm}
\begin{enumerate}[label={(\roman*)}]
\item \label{lem:NA_closure_union_of_indep_sets} The union of independent sets of NA random variables is NA.
\item \label{lem:NA_closure_mono_functions} Concordant monotone  functions (i.e., all monotonically increasing or all monotonically decreasing functions) defined on disjoint subsets of a set of NA random variables are NA.
\item \label{lem:perm_dist_are_NA} Let $\bx = \left[ x_1, x_2, \dots, x_n\right]$ be a real-valued vector and let $\bX = \left[ X_1, X_2,\dots, X_n\right]$ be a random vector which takes as values all the $n!$ permutations of $\bx$ with equal probabilities. Then $X_1, X_2, \dots, X_n$ are NA.
\end{enumerate}
\label{lem:NA_closure}
\end{lem}

\begin{lem}[Chernoff bound for NA Bernoulli variables \cite{dubhashi1998balls}]\label{lem:chern_hoeff_bound_NA} Let $X_1,  X_2, \dots, X_n$ be NA random variables with $X_i\in \{0,1 \}$ for $i \in [n]$. Then, $Y=\sum_{i=1}^n X_i$ satisfies:
\begin{align}
&\prob{Y > (1+\epsilon) \E[Y]} 
\le \exp\left( - \frac{\epsilon^2 \, \E[Y]}{2 +\epsilon}\right)\quad \text{ for any }\epsilon \ge 0,\label{eq:chernoff_upper_tail_bound_NA}\\
&\prob{Y < (1-\epsilon) \E[Y]} 
\le \exp\left( - \frac{\epsilon^2 \, \E[Y]}{2 - \epsilon}\right)\quad \text{ for any }\epsilon \in [0,1]. \label{eq:chernoff_lower_tail_bound_NA}
\end{align}

\end{lem}

\subsection{Proof of Lemma \ref{lem:1A_rank_r}} \label{subsec:proof_lem1A}

The pruned graph at the start of stage A  has $r \tk = r(\binom{k}{2} + k)$ left nodes, each with $d$ edges. Fig. \ref{subfig:pruned_graph_1A_t=0} is an example for $r=1$. The number of right nodes is $\nBins  = dr \tk/ (\delta \ln k)$.  Let $Z_j$ denote the degree of the $j$-th right node for $j\in [\nBins]$.  Since the total number of edges in the graph is $d r \tk$,  we have  $\sum_{j=1}^\nBins Z_j = dr \tk$.  

\begin{lem}[Initial right degrees are Binomial and NA]\label{lem:initial_right_deg_diag}
For $j \in [\nBins]$, we have $Z_j \sim \bin{r\tk}{\frac{d}{\nBins} }$. Furthermore, $Z_1, Z_2, \ldots, Z_{\nBins}$ are NA. 
\end{lem}
\begin{proof}
Recall that  the bipartite graph at the start of stage A for a rank-$r$ matrix $\bX$ has $r\tk$ left nodes and $\nBins$ right nodes. 
For $l\in [r\tk]$ and $j\in [\nBins]$, let
\begin{equation}\label{eq:def_Y_l,r_for_diagonal_graph_1A}
Y_{l, j}:=
\begin{cases}
    1 & \text{if the $l$-th left node connects to the $j$-th right node,}\\
    0 & \text{otherwise}.
  \end{cases}
\end{equation}
Observe that 
$Y_{l, j} \sim \bern{\frac{d}{\nBins}}$ and that $Y_{l_1, j_1} \perp Y_{l_2, j_2}$ for $l_1\neq l_2$. Therefore,
\begin{equation}\label{eq:Zr_as_sum_of_Yl,r}
Z_j= \sum_{l=1}^{r\tk} Y_{l, j}  \sim \bin{r\tk}{\frac{d}{\nBins}} \quad \text{for }j\in [\nBins].
\end{equation}

By the construction of the bipartite graph, the vector $\bY_l:= \left[Y_{l,1}, Y_{l,2}, \dots, Y_{l,\nBins}\right]$ contains $d$ ones, distributed uniformly at random among its $\nBins$ entries; the remaining $(\nBins -d)$ entries are zeros.  That is, the joint distribution of 
the entries of $\bY_l$  is a permutation distribution.  
Hence, $\left[Y_{l,1}, Y_{l,2}, \dots, Y_{l,\nBins}\right] $ is NA by Lemma \ref{lem:NA_closure}\ref{lem:perm_dist_are_NA}.  
Since the vectors  $\bY_l$ are mutually independent for $ l\in [r\tk ]$, the concatenated vector $[ Y_{l,j} \ \forall l \in [r\tk], \ \forall j \in [\nBins]]$
is NA by Lemma \ref{lem:NA_closure}\ref{lem:NA_closure_union_of_indep_sets}.
 Furthermore, from \eqref{eq:Zr_as_sum_of_Yl,r}, we note that  $Z_1, \dots, Z_{\nBins}$ are increasing functions defined on disjoint subsets of the  $Y_{l,j}$'s.  Thus, by  Lemma \ref{lem:NA_closure}\ref{lem:NA_closure_mono_functions},  $Z_1, \dots, Z_{\nBins}$ are NA.
\end{proof}

The number of left nodes (i.e. nonzero matrix entries) recovered in stage A is at least as large as  the number of singletons at the start of stage A, which equals $\sum_{j=1}^{\nBins} \mathbbm{1} {\left\lbrace Z_j =1\right\rbrace }$. We now use Lemma \ref{lem:initial_right_deg_diag} to obtain a high probability bound on this sum.
\begin{lem}[Bound on number of singletons]
\label{lem:bound_on_num_singletons_in_initial_graph_1A}
Let $Z\sim \bin{r\tk}{\frac{d}{\nBins} }$. Then, for sufficiently large $k$,   the number of  singletons $\sum_{j=1}^\nBins \mathbbm{1} {\left\lbrace Z_j=1\right\rbrace}$ satisfies: 
\beq\label{eq:bound_num_singletons_1A}
 \prob{ \abs{\frac{1}{\nBins}\sum_{j=1}^{\nBins} \mathbbm{1} {\left\lbrace Z_j =1\right\rbrace } - \prob{Z=1}  }
\geq  k^{-\delta} } 
< 4\exp\left(- \frac{dr}{\nConst \delta^2} \frac{k^{2-\delta }}{\ln^2 k}\right).
\eeq
\end{lem}
\begin{proof}
We first note
\begin{equation*}
   \frac{1}{\nBins}\sum_{j=1}^{\nBins} \mathbbm{1}{\left\lbrace Z_j =1\right\rbrace } -  \prob{Z=1} 
    =   \frac{1}{\nBins}\sum_{j=1}^{\nBins} \mathbbm{1}{\left\lbrace Z_j \le 1\right\rbrace } -  \prob{Z \le 1}  \,  + \,    \prob{Z = 0} - \frac{1}{\nBins}\sum_{j=1}^{\nBins} \mathbbm{1}{\left\lbrace Z_j =0 \right\rbrace } \, .
\end{equation*}
Applying the triangle inequality and a union bound, 
we have  
\begin{align}
& \prob{ \abs{\frac{1}{\nBins}\sum_{j=1}^{\nBins} \mathbbm{1}{\left\lbrace Z_j =1\right\rbrace } -  \prob{Z=1}  }
\geq k^{-\delta}} \nonumber\\
& \le
 \prob{ \abs{ \frac{1}{\nBins}\sum_{j=1}^{\nBins} \mathbbm{1}{\left\lbrace Z_j\leq 1 \right\rbrace}  -\prob{Z \leq 1 } }\geq  \frac{k^{-\delta}}{2} }  + 
 \prob{\abs{ \frac{1}{\nBins}\sum_{j=1}^{\nBins} \mathbbm{1}{\left\lbrace Z_j=0 \right\rbrace} -\prob{Z=0} } \geq    \frac{k^{-\delta}}{2} }. \label{eq:bound_pmf_deviation_term_by_the_sum_of_two_cmf_deviation_terms_diag_1A} 
\end{align}
For brevity, denote the two terms in \eqref{eq:bound_pmf_deviation_term_by_the_sum_of_two_cmf_deviation_terms_diag_1A} by $T_1$ and $T_2$, respectively.  Since  $Z_1, \dots, Z_{\nBins}$ are NA (by Lemma \ref{lem:initial_right_deg_diag}) and   $\mathbbm{1}{\left\lbrace Z_j\leq 1 \right\rbrace}$ is monotonic in $Z_j$, by Lemma \ref{lem:NA_closure}\ref{lem:NA_closure_mono_functions}, the random variables $\mathbbm{1}{\left\lbrace Z_j\leq 1 \right\rbrace}$ for $j\in [\nBins]$ are NA.  
Thus, applying the Chernoff bound  in Lemma \ref{lem:chern_hoeff_bound_NA} and recalling that $\nBins= d r \tk/ (\delta \ln k)$, we obtain that for sufficiently large $k$  
\edit{
\begin{align}
 T_1 \le 2\exp\left( - \frac{\nBins}{2(2k^{\delta})^2 \prob{Z\le 1} (1+o(1))}\right) < 2\exp\left( - \frac{d r}{\nConst  \delta^2} \frac{k^{2-\delta}}{\ln^2k }\right),
\end{align}
}
where we have used the fact  that $\prob{Z=0}= k^{-\delta}(1- o(1))$ and $\prob{Z=1}= (\delta \ln k) k^{-\delta}(1- o(1))$. Similarly, using that $\mathbbm{1}{\left\lbrace Z_j =0 \right\rbrace}$ is monotonic in $Z_j$, we deduce that 
\edit{
\begin{align}
 T_2 \le 2\exp\left( - \frac{\nBins}{2(2k^{\delta})^2 \prob{Z = 0} (1+o(1))}\right) < 2\exp\left( - \frac{d r}{\nConst \delta} \frac{k^{2-\delta}}{\ln k }\right).
 \end{align}
 }
Combining the upper bounds on $T_1$ and $T_2$ yields \eqref{eq:bound_num_singletons_1A}. 
\end{proof}

We now prove \eqref{eq:bound_on_A_rank_r} from \eqref{eq:bound_num_singletons_1A}. 
Using $\nBins  = dr\tk/(\delta \ln k)$  and  $\prob{Z=1} =  (\delta  \ln k) k^{-\delta} (1-o(1))$, 
we deduce from \eqref{eq:bound_num_singletons_1A} that there exists a non-negative sequence $c_{k, \delta} = o(1)$ such that for sufficiently large $k$, 
\beq \label{eq:bound_num_singletons_Stage_A}
\prob{ \sum_{j=1}^\nBins \mathbbm{1}{\left\lbrace Z_j=1 \right\rbrace}
\le dr\tk \cdot k^{-\delta} (1 - c_{k, \delta})} 
< 4\exp\left(- \frac{dr}{\nConst \delta^2} \frac{k^{2-\delta }}{\ln^2 k}\right).
\eeq
Since each left node connects to at most $d$ singletons, the number of left nodes recoverable from just the singletons in the initial graph is at least $\frac{1}{d} \sum_{j=1}^\nBins \mathbbm{1}{\lbrace Z_j=1 \rbrace} $.  This number is  smaller than the total number  of  left nodes  recoverable in  stage A because  additional singletons may be created during the peeling process. Thus, with $\sfA $ denoting the total fraction of  left nodes (out of $r \tk$) recovered  in stage A, using \eqref{eq:bound_num_singletons_Stage_A} we have:
\begin{align}\label{eq:bound_num_lnodes_recovered_Stage_A}
 \prob{ \sfA r \tk \le
r \tk  k^{-\delta} (1 - c_{k, \delta})}
& \le 
\prob{ \frac{1}{d} \sum_{j=1}^\nBins \mathbbm{1}{\left\lbrace Z_j=1 \right\rbrace}  \le
r \tk k^{-\delta} (1 - c_{k, \delta}) } \nonumber\\
& <
4\exp\left(- \frac{dr}{\nConst \delta^2} \frac{k^{2-\delta }}{\ln^2 k}\right).
\end{align}
Rearranging \eqref{eq:bound_num_lnodes_recovered_Stage_A} yields \eqref{eq:bound_on_A_rank_r}.

We now prove \eqref{eq:bound_on_all_Ai_rank_r}. For $l\in [r\tk ]$, let 
\beq
\barV_l : =
\begin{cases}
    1 & \text{if the $l$-th left node is recovered by the end of stage A},\\
    0 & \text{otherwise}.
  \end{cases} 
\eeq
Since the recovered left nodes are distributed uniformly at random among all left nodes, each $\barV_l\sim \bern{\sfA}$. 
Moreover, conditioned on $\sum_{l=1}^{r \tk}\barV_l = \sfA r\tk$,  the $\{\barV_l\}$ are NA by Lemma \ref{lem:NA_closure}\ref{lem:perm_dist_are_NA}.
 Let $\mc{S}_{Di}$ and $\mc{S}_i$ be the set of indices of the left nodes that represent the diagonal and above-diagonal entries in the $i$-th nonzero submatrix (corresponding to $\lambda_i \bv_i\bv_i^T$), respectively, with $\abs{\mc{S}_{Di}} = k$ and $\abs{\mc{S}_i} = \binom{k}{2}$. 
 
 The number of diagonal entries (out of $k$) recovered in the $i$-th nonzero submatrix is  $N_{Di} = \sum_{l\in \mc{S}_{Di}} \barV_l$, which can be bounded as follows using the Chernoff bound in Lemma \ref{lem:chern_hoeff_bound_NA}. Recalling that 
  $\alphaUB=k^{-\delta} - o(k^{-\delta})$, we have for  any fixed $\epsilon \in (0,\frac{1}{5})$ and sufficiently large $k$, 
\begin{align}
\prob{N_{Di} < 1  \bigmid \sfA=\alpha^*}  
& \le \prob{N_{Di} < k\alpha^* \epsilon  \bigmid \sfA=\alpha^*}  \nonumber\\
& \le \prob{ \sum_{l\in \mc{S}_{Di}} \barV_l <   k\E[\barV_l \mid \sfA] \epsilon \bigmid \sfA=\alpha^*}  \nonumber\\
& \le \exp\left( - \frac{ (1-\epsilon)^2 k\alpha^* }{1+ \epsilon}\right)  \le \exp\left(- \frac{1}{2} k^{1-\delta}\right). \label{eq:bound_on_num_diag_recovered_in_each_block}
\end{align}

The fraction of above-diagonal entries (out of $\binom{k}{2}$)  recovered in the $i$-th nonzero submatrix is $\sfA_i = \frac{1}{\binom{k}{2} }\sum_{l\in \mc{S}_i}\barV_l$.   
Let $\alphaUB_i = \alphaUB(1-k^{-\frac{1}{4}})$.  Then using the  Chernoff bound in Lemma \ref{lem:chern_hoeff_bound_NA},  we have for sufficiently large $k$,
\begin{align} 
\prob{\sfA_i < \alphaUB_{i}  \bigmid \sfA =  \alphaUB} 
& = \prob{  \frac{1}{\binom{k}{2} }\sum_{l\in \mc{S}_i}\barV_l < \alphaUB \left(1- k^{-\frac{1}{4}} \right)\bigmid \sfA =  \alphaUB}\nonumber\\
& = \prob{  \frac{1}{\binom{k}{2} }\sum_{l\in \mc{S}_i}\barV_l <  \E[\barV_l \mid \sfA]\left(1- k^{-\frac{1}{4}} \right)   \bigmid \sfA =  \alphaUB}\nonumber\\
& \le \exp\left( - \frac{\left( k^{- \frac{1}{4}}\right)^2\binom{k}{2} \alpha^* }{2}\right)
 \le \exp\left(- \frac{1}{8} k^{\frac{3}{2} - \delta}\right). \label{eq:bound_on_num_above_diag_recovered_in_each_block}
\end{align}

Finally, we obtain \eqref{eq:bound_on_all_Ai_rank_r} using \eqref{eq:bound_on_num_diag_recovered_in_each_block} and \eqref{eq:bound_on_num_above_diag_recovered_in_each_block} as follows. For sufficiently large $k$, we have
\begin{align}\label{eq:bound_on_Ai_and_NDi}
&\prob{N_{Di} \ge 1 \ \text{and}\ \sfA_i \ge \alphaUB_{i}, \ \forall \ i\in [r]} \nonumber\\
& \ge \prob{\sfA\ge \alpha^*} \prob{ N_{Di} \ge 1\text{ and } \sfA_i \ge \alphaUB_{i},\ \forall \ i\in [r]\mid\sfA = \alpha^* } \nonumber\\
& \stackrel{(\rm{i})}{\ge} \prob{\sfA\ge \alpha^*}
\left[1-\sum_{i=1}^r \prob{N_{Di} < 1  \bigmid \sfA=\alpha^*} 
- \sum_{i=1}^r \prob{\sfA_i< \alpha_i^* \bigmid \sfA=\alpha^*} \right] \nonumber\\
&\stackrel{(\rm{ii})}{\ge} \left[ 1-
4\exp\left(- \frac{dr}{\nConst \delta^2} \frac{k^{2-\delta }}{\ln^2 k}\right) \right]
\left[ 1-r\exp \left(-\frac{1}{2}  k^{1-\delta}\right) - r\exp \left(-\frac{1}{8}  k^{\frac{3}{2}-\delta}\right)\right] \nonumber\\
& \ge 1-2r\exp \left(-\frac{1}{2}  k^{1-\delta}\right) ,
\end{align}  
 where  (\rm{i}) is obtained using a union bound, and (\rm{ii})  by applying \eqref{eq:bound_on_A_rank_r}, \eqref{eq:bound_on_num_diag_recovered_in_each_block} and \eqref{eq:bound_on_num_above_diag_recovered_in_each_block}.
\qed

\subsection{Proof of Lemma \ref{lem:1B_rank_r}} \label{subsec:proof_lem1B}

Since the eigenvectors $\{\bv_1, \ldots, \bv_r\}$ have non-overlapping supports, the nonzero entries of $\bX = \sum_{i=1}^r \lambda_i \bv_i \bv_i^{T}$ form $r$ disjoint submatrices, each of size $k \times k$. Therefore, the graph for stage B  consists of $r$ disjoint bipartite  subgraphs, one corresponding to each submatrix. Specifically, the $i$-th  subgraph has  $k$ left nodes representing the unknown nonzeros in the $i$-th eigenvector and  $\sfA_i\binom{k}{2}$ right nodes representing the nonzero pairwise products recovered in the $i$-th submatrix in stage A. (See Fig. \ref{subfig:graph_1B_t=-1} for an example subgraph.) The peeling algorithm described in Section \ref{subsubsec:stageB}  is applied to each of these subgraphs.

\begin{figure}
\captionsetup[subfigure]{justification=centering}
\centering
\linespread{1}
 \begin{subfigure}[b]{0.24\linewidth} 
 \centering 
 \begin{adjustwidth}{-0.5cm}{0cm}
\begin{tikzpicture}
\node[vnode](vnodev1){$\tv_1$};
\node[vnode, below of=vnodev1, node distance=1cm](vnodev2){$\tv_2$};
\node[vnode, below of=vnodev1, node distance=2cm](vnodev3){$\tv_3$};
\node[vnode, below of=vnodev1, node distance=3cm] (vnodev4){$\tv_4$};
\node[vnode, below of=vnodev1, node distance=4cm](vnodev5){$\tv_5$};

\node[cnode, right of= vnodev1, yshift = -0.5cm, node distance=2.5cm](cnodev1v4){$\tv_1\tv_4$};
\node[cnode, right of= vnodev3, node distance=2.5cm, yshift=0.5cm](cnodev2v3){$\tv_2\tv_3$};
\node[cnode, right of= vnodev3, node distance=2.5cm, yshift=-0.5cm](cnodev2v3-1){$\tv_2\tv_3$};
\node[cnode, right of= vnodev5, yshift = +0.5cm, node distance=2.5cm](cnodev4v5){$\tv_4\tv_5$};

\draw (vnodev1)--(cnodev1v4);
\draw (vnodev4)--(cnodev1v4);
\draw (vnodev2)--(cnodev2v3);
\draw (vnodev3)--(cnodev2v3);
\draw  (vnodev2)--(cnodev2v3-1);
\draw (vnodev3)--(cnodev2v3-1);
\draw  (vnodev4)--(cnodev4v5);
\draw (vnodev5)--(cnodev4v5);

\node[above of=vnodev1, node distance=0.8cm, text width=2.5cm, align=center, scale=0.8](numLnodeLabel){$k=5$ nonzero entries of $\tbv_i$};
\node[above of=cnodev1v4, node distance=0.8cm, text width=2cm, align=center, scale=0.8](numRnodeLabel){$\sfA_i \binom{k}{2}=4$ right nodes};
\node[below of=cnodev4v5, node distance=0.8cm, text width=3cm, align=center, scale=0.8](rightDist1){Degree of each right node = 2};
\end{tikzpicture}
\end{adjustwidth}
\vspace{0.2cm}
\caption{Prior to $t=0$.}
\label{subfig:alternative_graph_1B_t=-1}
\end{subfigure}
\begin{subfigure}[b]{0.24\linewidth} 
\centering
 \begin{tikzpicture}
\node[vnode](vnodev1){$\tv_1$};
\node[vnode, below of=vnodev1, node distance=1cm](vnodev2){$\tv_2$};
\node[vnode, below of=vnodev1, node distance=2cm](vnodev3){$\tv_3$};
\node[vnodeFaded,  below of=vnodev1, node distance=3cm] (vnodev4){$\tv_4$};
\node[vnode, below of=vnodev1, node distance=4cm](vnodev5){$\tv_5$};
\node[cnode, right of= vnodev1, yshift = -0.5cm, node distance=2.5cm](cnodev1v4){$\tv_1\tv_4$};
\node[cnode, right of= vnodev3, node distance=2.5cm, yshift=0.5cm](cnodev2v3){$\tv_2\tv_3$};
\node[cnode, right of= vnodev3, node distance=2.5cm, yshift=-0.5cm](cnodev2v3-1){$\tv_2\tv_3$};
\node[cnode, right of= vnodev5, yshift = +0.5cm, node distance=2.5cm](cnodev4v5){$\tv_4\tv_5$};
\draw (vnodev1)--(cnodev1v4);
\draw [gray!50](vnodev4)--(cnodev1v4);
\draw (vnodev2)--(cnodev2v3);
\draw (vnodev3)--(cnodev2v3);
\draw  (vnodev2)--(cnodev2v3-1);
\draw (vnodev3)--(cnodev2v3-1);
\draw  (vnodev5)--(cnodev4v5);
\draw [gray!50](vnodev4)--(cnodev4v5);
\node[above of=vnodev1, node distance=0.8cm, text width=1.5cm, align=center, scale=0.8, text=white](numLnodeLabel){$k=5$ supports of $\tbv$}; 
\node[below of=cnodev4v5, xshift = 0.2cm, node distance=1cm, text width=4cm, align=center, scale=0.8](rightDist1){Right nodes connected to $\tv_4$ have degree 1; others have degree 2};
\end{tikzpicture}
\caption{$t=0$.}
\label{subfig:alternative_graph_1B_t=0}
\end{subfigure}
\begin{subfigure}[b]{0.24\linewidth} 
\centering
 \begin{tikzpicture}
\node[vnodeFaded](vnodev1){$\tv_1$};
\node[vnode, below of=vnodev1, node distance=1cm](vnodev2){$\tv_2$};
\node[vnode, below of=vnodev1, node distance=2cm](vnodev3){$\tv_3$};
\node[vnodeFaded,  below of=vnodev1, node distance=3cm] (vnodev4){$\tv_4$};
\node[vnode, below of=vnodev1, node distance=4cm](vnodev5){$\tv_5$};
\node[cnode, right of= vnodev1, yshift = -0.5cm, node distance=2.5cm](cnodev1v4){$\tv_1\tv_4$};
\node[cnode, right of= vnodev3, node distance=2.5cm, yshift=0.5cm](cnodev2v3){$\tv_2\tv_3$};
\node[cnode, right of= vnodev3, node distance=2.5cm, yshift=-0.5cm](cnodev2v3-1){$\tv_2\tv_3$};
\node[cnode, right of= vnodev5, yshift = +0.5cm, node distance=2.5cm](cnodev4v5){$\tv_4\tv_5$};
\draw [gray!50](vnodev1)--(cnodev1v4);
\draw [gray!50](vnodev4)--(cnodev1v4);
\draw (vnodev2)--(cnodev2v3);
\draw (vnodev3)--(cnodev2v3);
\draw  (vnodev2)--(cnodev2v3-1);
\draw (vnodev3)--(cnodev2v3-1);
\draw [gray!50] (vnodev4)--(cnodev4v5);
\draw  (vnodev5)--(cnodev4v5);
\node[above of=vnodev1, node distance=0.8cm, text width=1.5cm, align=center, scale=0.8, text=white](numLnodeLabel){$k=5$ supports of $\tbv$}; 
\node[below of=cnodev4v5, xshift = 0.2cm, node distance=1cm, text width=4cm, align=center, scale=0.8, text=white](rightDist1){Right nodes connected to $\tv_4$ have degree 1; others have degree 2};
\end{tikzpicture}
\caption{$t=1$.}
\label{subfig:alternative_graph_1B_t=1}
\end{subfigure}
\begin{subfigure}[b]{0.24\linewidth} 
\centering
 \begin{tikzpicture}
\node[vnodeFaded](vnodev1){$\tv_1$};
\node[vnode, below of=vnodev1, node distance=1cm](vnodev2){$\tv_2$};
\node[vnode, below of=vnodev1, node distance=2cm](vnodev3){$\tv_3$};
\node[vnodeFaded,  below of=vnodev1, node distance=3cm] (vnodev4){$\tv_4$};
\node[vnodeFaded, below of=vnodev1, node distance=4cm](vnodev5){$\tv_5$};

\node[cnode, right of= vnodev1, yshift = -0.5cm, node distance=2.5cm](cnodev1v4){$\tv_1\tv_4$};
\node[cnode, right of= vnodev3, node distance=2.5cm, yshift=0.5cm](cnodev2v3){$\tv_2\tv_3$};
\node[cnode, right of= vnodev3, node distance=2.5cm, yshift=-0.5cm](cnodev2v3-1){$\tv_2\tv_3$};
\node[cnode, right of= vnodev5, yshift = +0.5cm, node distance=2.5cm](cnodev4v5){$\tv_4\tv_5$};
\draw [gray!50](vnodev1)--(cnodev1v4);
\draw [gray!50](vnodev4)--(cnodev1v4);
\draw (vnodev2)--(cnodev2v3);
\draw (vnodev3)--(cnodev2v3);
\draw (vnodev2)--(cnodev2v3-1);
\draw (vnodev3)--(cnodev2v3-1);
\draw [gray!50] (vnodev4)--(cnodev4v5);
\draw [gray!50](vnodev5)--(cnodev4v5);
\node[above of=vnodev1, node distance=0.8cm, text width=1.5cm, align=center, scale=0.8, text=white](numLnodeLabel){$k=5$ supports of $\tbv$}; 
\node[below of=cnodev4v5, xshift = 0.2cm, node distance=1cm, text width=4cm, align=center, scale=0.8, text=white](rightDist1){Right nodes connected to $\tv_4$ have degree 1; others have degree 2};
\end{tikzpicture}
\caption{$t=2$.}
\label{subfig:alternative_graph_1B_t=2}
\end{subfigure}
\caption{\small Alternative  graph process showing the recovery of $\tbv_i$ for $\bX=\sum_{i=1}^r \tbv_i\tbv_i^T$ 
given  $\sfA_i$ and ${\NDiag}_i\geq 1$ 
(as the counterpart of Fig. \ref{fig:actual_graph_process_1B}). 
The double occurrence of $\tv_2\tv_3$ as right nodes highlights that the right nodes are sampled \emph{with replacement} from the $\binom{k}{2}$ nonzero pairwise products in the nonzero submatrix corresponding to  $\tbv_i\tbv_i^T$. 
The faded nodes and edges are those that have been peeled off.}
\label{fig:alternative_graph_process_1B}
\end{figure}
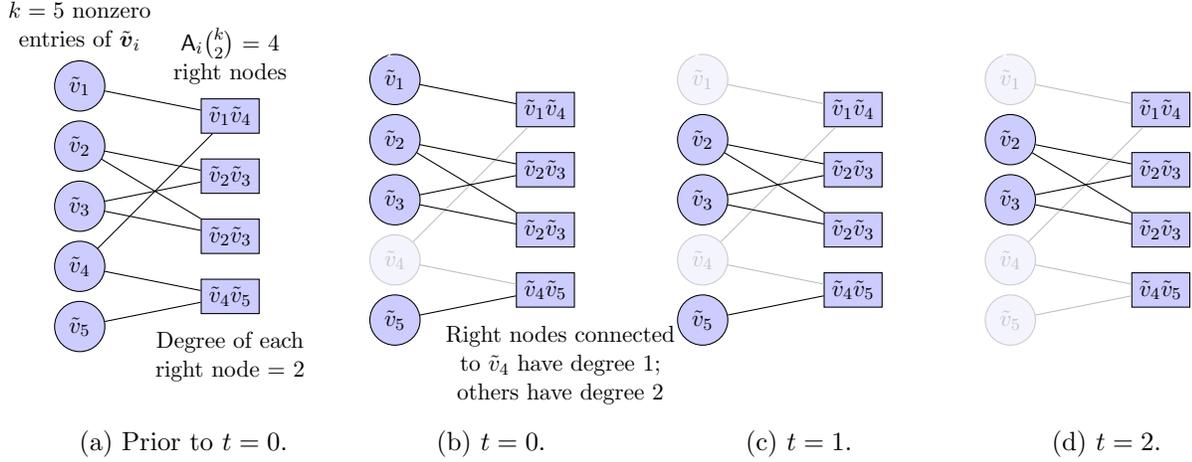

Consider the $i$-th subgraph, given $\sfA_i$ and the event that ${\NDiag}_i \ge 1$.   The  $\sfA_i\binom{k}{2}$ right nodes  can be seen as drawn uniformly at random without replacement from the $\binom{k}{2}$ nonzero pairwise products in the $i$-th submatrix. This creates  dependence between the degrees of the left nodes, which makes the evolution of the random graph process in Fig. \ref{fig:actual_graph_process_1B} difficult to characterize. We therefore consider an alternative graph process, in which the $\sfA_i\binom{k}{2}$ right nodes at the start of stage B are drawn uniformly at random \emph{with replacement} from the $\binom{k}{2}$ nonzero pairwise products in the $i$-th submatrix; see Fig. \ref{subfig:alternative_graph_1B_t=-1} for an illustration.    Stage B of the algorithm proceeds in the same way on the alternative graph as the original. At $t=0$, a left node is recovered based on one of the nonzero diagonal entries recovered in stage A (Fig. \ref{subfig:alternative_graph_1B_t=0}). Then, for each $t \ge 1$,  a degree-1 right node is picked uniformly at random from the available ones and its connecting left node recovered and peeled off to update the set of degree-1 right nodes. This continues till there are no more degree-1 right nodes (Figs. \ref{subfig:alternative_graph_1B_t=1} and \ref{subfig:alternative_graph_1B_t=2}). 

Since the right nodes in the alternative graph are sampled with replacement, the number of distinct right nodes in the alternative graph at the start of stage B can be no larger than $\sfA_i\binom{k}{2}$. Any repeated  right nodes do not play a role in the recovery algorithm. Let  $\Balti$ denote the fraction of left nodes recovered when there remain no degree-1 right nodes in the alternative graph process. Since the number of distinct right nodes in the initial alternative graph is no larger than that in the original one, we have:
\begin{equation} 
    \prob{ \Balti  < 1\bigmid N_{Di}, \, \sfA_i } \ge 
    \prob{ \sfB_i  < 1\bigmid N_{Di}, \, \sfA_i }, \quad  i \in [r],
    \qquad \forall \, N_{Di} \ge 1,  \  \sfA_i \ge \alphaUB_i.
    \label{eq:alt_orig_bnd}
\end{equation}
We will prove \eqref{eq:bound_on_Bi_rank_r} by showing that 
\begin{equation}\label{eq:bound_on_Balti_rank_r}
\prob{\Balti < 1 \bigmid N_{Di}, \,  \sfA_i } < \exp\left(-\frac{1}{30} k^{1-\delta}\right),  \quad  i \in [r], \qquad 
\forall \,  N_{Di} \ge 1,  \  \sfA_i \ge \alphaUB_i.
\end{equation}

For brevity, we shall refer to `prior to $t=0$' as `at $t=-1$'.   Given $\sfA_i$,  let $\tZ_l$ be the degree of the $l$-th left node in the alternative graph at $t=-1$ (i.e. Fig. \ref{subfig:alternative_graph_1B_t=-1}), for $l\in [k]$. Since each right node has degree-2 at $t=-1$, the total number of edges at $t=-1$ is  $\sum_{l=1}^{k} \tZ_l= 2\sfA_i\binom{k}{2}$.

\begin{lem}[Initial left degrees are Binomial]
\label{lem:initial_left_degree_bound_1B}
For $l\in [k]$, 
\begin{equation}\label{eq:1B_left_deg_binomial}
\tZ_l\sim \bin{\sfA_i \binom{k}{2} }{\frac{2}{k}}.
\end{equation}
Let
  \begin{equation}
\mu :=\sfA_i\binom{k}{2} \frac{2}{k}=\sfA_i(k-1).
\label{eq:def_mu_Stage_B}
  \end{equation}
Then, for $\sfA_i \ge \alpha^*_i$, we have for sufficiently large $k$,
\begin{align}
\prob{ \tZ_l \le \frac{\mu}{2} }
 < \exp\left( -\frac{1}{10} k^{1-\delta} \right)
  ,  \qquad 
 \prob{  \tZ_l \ge \frac{3\mu}{2}} < \exp\left( -\frac{1}{10} k^{1-\delta} \right).
\label{eq:tail_bounds_left_degrees_Stage_B} 
\end{align}
\end{lem}
\begin{proof}
The result \eqref{eq:1B_left_deg_binomial} follows from a standard counting argument, noting that the right nodes are picked uniformly at random with replacement  and that each right node has degree 2 at $t=-1$. Result \eqref{eq:tail_bounds_left_degrees_Stage_B} is obtained using the Binomial tail bound in Lemma \ref{lem:binom_tail_bound} and simplifying the resulting expression.
\end{proof}

Recall that one left node is recovered in each iteration $t \ge 0$ of the peeling algorithm, until there are no more degree-1 right nodes left.
Let $\pi(t)$  denote the index of the left node recovered in iteration $t$. Thus, $\pi: \left\lbrace 0, 1, \dots, k-1 \right\rbrace  \rightarrow \left\lbrace 1,2, \dots, k \right\rbrace$. If the decoding process has not terminated after iteration $(\nlnodes-1)$, for $1 \le \nlnodes \le (k-1)$, then $\nlnodes$ left nodes out of $k$ have been peeled off. For the algorithm to proceed, we need at least one of the remaining $(k - \nlnodes)$ left nodes to be connected to at least one degree-1 right node. The following lemma obtains the distribution of the number of degree-1 right nodes connected to each remaining left node, after each iteration of the peeling  algorithm.

\begin{lem}[Number of degree-1 right nodes connected to each unpeeled left node]
\label{lem:num_singletons}
Consider the alternative graph process on subgraph $i \in [r]$, given $\sfA_i$ and ${\NDiag}_i $.  Suppose that  the peeling algorithm on this  subgraph has not terminated after  iteration $(\nlnodes-1)$, for $1\leq \nlnodes \leq (k-1)$.  Let $S_l^{(\nlnodes)}$ be the number of degree-1 right nodes connected to the $l$-th remaining left node at the start of iteration $\nlnodes$, for $l\in [k]\setminus \left\lbrace\pi(0), \ldots, \pi(\nlnodes-1) \right\rbrace$. Then
\begin{equation}\label{eq:dist_Sl}
S_l^{(\nlnodes)} \sim \bin{\sfA_i \binom{k}{2}}{\frac{\nlnodes}{\binom{k}{2}}}.
\end{equation}
Moreover, for $\sfA_i \ge  \alphaUB_i$ and sufficiently large $k$, the following holds for $1\le \nlnodes \le \frac{k-1}{20}$:
\beq\label{eq:Slt_LB}
\prob{ S_{l}^{(\nlnodes)} \ge \frac{\mu}{5} \mid S_{l}^{(\nlnodes)} \geq 1}
  <  \exp\left(- \frac{1}{16}k^{1-\delta} \right),
\eeq
where $\mu = \sfA_i (k-1)$ as defined in \eqref{eq:def_mu_Stage_B}.
\end{lem}
\begin{proof}
For $l\in [k]$ and $j \in [\sfA_i \binom{k}{2}]$, let
\begin{equation}
Y_{l,j}:=
\begin{cases}
    1 & \text{if the $l$-th left node  connects to the $j$-th right node  at $t=-1$},\\
    0 & \text{otherwise}.
  \end{cases}
\end{equation}
At $t=-1$, since each right node has degree $2$  (see Fig. \ref{subfig:alternative_graph_1B_t=-1}), the vector $ \left[Y_{1,j}, Y_{2,j}, \dots, Y_{k,j}\right]$ contains exactly two ones distributed uniformly at random, and the remaining entries are zero. It follows that $\sum_{l=1}^k Y_{l,j} = 2$, $Y_{l,j} \sim \bern{\frac{2}{k}}$ and $Y_{l_1,j} \cdot Y_{l_2,j} \sim \text{Bern}(1/ \binom{k}{2})$ for $l_1\neq l_2$. Moreover, by construction of the alternative graph,  the edges of each right node are independent from those of the others, so $Y_{l_1, j_1} \perp Y_{l_2, j_2}$ for $j_1\neq j_2$.

 In iteration $t=0$, the first left node $v_{\pi(0)}$ is peeled off. 
Consider one of the $(k-1)$  remaining left nodes after iteration $t=0$, say $v_l$.
The number of degree-1 right nodes connected to $v_l$  is equal to the number of right nodes  which were connected to $v_l$ and $v_{\pi(0)}$ at $t=-1$. That is, at the start of iteration $t=1$, we have 
 \begin{align}
S_{l}^{(1)} =  \sum_{j=1}^{\sfA_i \binom{k}{2}} Y_{l,j} \cdot Y_{\pi(0),j}   
\stackrel{\rm{(i)}}{\sim} \bin{\sfA_i \binom{k}{2}}{\frac{1}{\binom{k}{2}}},  \qquad l\in [k] \setminus \pi(0),
 \end{align}
 where (\rm{i}) holds because $Y_{l,j} \cdot Y_{\pi(0),j} \stackrel{\text{iid}}{\sim} \text{Bern}(1/\binom{k}{2} )$ for $l\neq \pi(0)$ and 
 $ j\in [\sfA_i\binom{k}{2}]$.

 More generally,  when $\nlnodes$ left nodes have been peeled off by the end of  iteration $(\nlnodes-1)$, consider one of the remaining left nodes,  say $v_l$. Then the number of degree-1 right nodes connected to $v_l$   is the number of right nodes which were originally connected to $v_l$ and some recovered $v_\iota$, i.e., $v_\iota$ for $\iota \in \left\lbrace\pi(0), \ldots, \pi(\nlnodes-1) \right\rbrace$. Hence, we deduce that
\begin{equation}
S_l^{(\nlnodes)}  =\sum_{j=1}^{\sfA_i \binom{k}{2}}  \sum_{\kappa=0}^{\nlnodes-1} Y_{l,j} \cdot Y_{\pi(\kappa),j}
 \stackrel{(\rm{i})}{\sim} \bin{\sfA_i\binom{k}{2}}{\frac{\nlnodes}{\binom{k}{2}}}, 
 \qquad l\in [k]\setminus \lbrace\pi(0), \ldots, \pi(\nlnodes-1) \rbrace,
\end{equation}
where (\rm{i}) is obtained by noting that for  $l\notin \left\lbrace\pi(0), \ldots, \pi(\nlnodes -1) \right\rbrace$ and each $j$, the sum $ \sum_{\kappa=0}^{\nlnodes -1}Y_{l,j} Y_{\pi(\kappa),j} \sim \text{Bern}(\nlnodes /\binom{k}{2})$, and that these  sums are mutually independent across $j \in [\sfA_i \binom{k}{2}]$.
 This gives the result in \eqref{eq:dist_Sl}. 
 
 We next prove \eqref{eq:Slt_LB} for $\sfA_i \ge \alphaUB_i$ using the distribution in \eqref{eq:dist_Sl}. For $1\leq \nlnodes\leq \frac{k-1}{20}$  and sufficiently large $k$,  applying the Binomial tail bound in Lemma \ref{lem:binom_tail_bound}, we obtain
\begin{align}\label{eq:P(S>nmax | S>=1)}
\prob{ S_{l}^{(\nlnodes)} \ge \frac{\mu}{5} \mid S_{l}^{(\nlnodes)} \geq 1}
& \le  \frac{ \exp\left(-\sfA_i\binom{k}{2} D_e\Big(\frac{\mu/5}{\sfA_i\binom{k}{2}}\middle\| \frac{\nlnodes}{\binom{k}{2}}\Big) \right) }{ 1-\left(1- \nlnodes/\binom{k}{2} \right)^{\sfA_i\binom{k}{2}}} \nonumber\\
& \le \frac{ \exp\left(-\sfA_i\binom{k}{2} D_e\Big(\frac{\mu/5}{\sfA_i\binom{k}{2}}\middle\| \frac{(k-1)/20}{\binom{k}{2}}\Big) \right) }{ 1-\left(1- 1/\binom{k}{2} \right)^{\sfA_i\binom{k}{2} }} \nonumber\\
& \stackrel{(\rm{i})}{<} \frac{ \exp\left(-k^{1-\delta}/8 \right) }{ \alphaUB_i/2 }    <\exp\left(-\frac{1}{16}k^{1-\delta}\right).
\end{align}
For the inequality  (\rm{i}), the numerator is obtained by simplifying the relative entropy term  and noting that $\sfA_i \ge \alphaUB_i$; the denominator is obtained using
$(1- 1/\binom{k}{2})^{\sfA_i\binom{k}{2} }\le (1- 1/\binom{k}{2})^{\alphaUB_i\binom{k}{2} } \leq e^{-\alphaUB_i}
<1- \alphaUB_i/ 2$ (where the last step holds since $\alphaUB_i \in (0,1)$). 
\end{proof}
\begin{lem}[Degree-1 right nodes don't run out in  initial iterations]\label{lem:deg_1_rnodes_upto_t0}
Consider the alternative graph process on subgraph $i \in [r]$, given $\sfA_i \ge \alphaUB_i$ and ${\NDiag}_i \ge 1$. Let $C_1(t) \in \nonNegInt$ denote the  number of degree-1 right nodes in the residual subgraph after iteration $t$ of the peeling algorithm.
Let $t_0: = \frac{k-1}{20}$, and define the event
\begin{equation}\label{eq:Edef}
\mathcal{E} := \left\lbrace C_1(t)>0, \ \text{ for } \ 0\le t\le  t_0 \right\rbrace.
\end{equation}
Then, for sufficiently large $k$,
\begin{align}
&\prob{\mathcal{E} } >  1-\exp\left(-\frac{1}{20} k^{1-\delta}\right).\label{eq:hp_bnd_E}
\end{align}
\end{lem}

\begin{proof}
   Recall that the initial left degrees (at $t=-1$), denoted by $\{ \tZ_l \}_{l \in [k]}$, each has a Binomial distribution given in \eqref{eq:1B_left_deg_binomial}. 
Recall also that $\pi(t)$ is the index of the left node recovered in iteration $t$. At $t=0$, the first left node $v_{\pi(0)}$ is peeled off, creating $\tZ_{\pi(0)}$ degree-1 right nodes. Thus, we have
\begin{equation} \label{eq:C1(0)}
C_1(0) = \tZ_{\pi(0)}.
\end{equation}
More generally,  after iteration $(\nlnodes-1)$, when $\nlnodes$ left nodes have been peeled off, consider one of the $(k-\nlnodes )$ remaining left nodes, say $v_l$, for some $l\in[k] \setminus \left\lbrace \pi(0), \dots, \pi( \nlnodes -1)\right\rbrace$.  At the start of iteration $\nlnodes$, among the  $\tZ_l$ edges of $v_l$, the number of  edges connected to degree-1 right nodes is $S_{l}^{(t)}$ whose distribution is given in \eqref{eq:dist_Sl}. 
When one of the $(k-\nlnodes)$ remaining left nodes, denoted by $v_{\pi(t)}$, is peeled off in  iteration  $t$, the number of degree-1 right nodes in the residual graph can be expressed as:
\begin{align}
C_1(\nlnodes) &= C_1(\nlnodes-1)+\#\text{ degree-$2$ right nodes reduced to degree-$1$ due to the removal of }v_{\pi(\nlnodes)} \nonumber\\
&\quad - \#\text{ degree-$1$ right nodes reduced to degree-$0$ due to the removal of }v_{\pi(\nlnodes)}\nonumber\\
 &=   C_1(\nlnodes-1)+\left(\tZ_{\pi(\nlnodes)} - S_{\pi(\nlnodes)}^{(\nlnodes)}\right) - S_{\pi(\nlnodes)}^{(\nlnodes)}\nonumber\\
 &=  C_1(\nlnodes-1)+\left(\tZ_{\pi(\nlnodes)} - 2S_{\pi(\nlnodes)}^{(\nlnodes)}\right) .\label{eq:C1n_C1n-1}
\end{align}
The marginal distributions of $\tZ_{\pi(\nlnodes)}$ and $S_{\pi(\nlnodes)}^{(\nlnodes)}$ are given by 
 Lemmas \ref{lem:initial_left_degree_bound_1B} and \ref{lem:num_singletons},
which also guarantee that with high probability, $\tZ_{\pi(\nlnodes)} > \frac{\mu}{2}$ and $S_{\pi(\nlnodes)}^{(\nlnodes)} < \frac{\mu}{5}$ for $ \nlnodes \le t_0$.

By the same reasoning as above, we can write $C_1(\nlnodes-1)$ in terms of $C_1(\nlnodes-2)$ and so on to obtain
\begin{align}
C_1(\nlnodes) &=  C_1(\nlnodes-1)+\left(\tZ_{\pi(\nlnodes)} - 2S_{\pi(\nlnodes)}^{(\nlnodes)}\right) \nonumber\\
 & = C_1(0) + \sum_{\kappa=1}^{\nlnodes} \left(\tZ_{\pi(\kappa)} - 2S_{\pi(\kappa)}^{(\kappa)}\right)\nonumber\\
 & = \tZ_{\pi(0)} + \sum_{\kappa=1}^{\nlnodes} \left(\tZ_{\pi(\kappa)} - 2S_{\pi(\kappa)}^{(\kappa)}\right), 
 \label{eq:C1t_Z}
\end{align}
where the last step follows from \eqref{eq:C1(0)}. We next show that  $C_1(\nlnodes)>0$ with high probability for $0\le t \le t_0$ by bounding each term on the RHS of \eqref{eq:C1t_Z}. From \eqref{eq:C1t_Z}, we note that 
\begin{align}
\left\lbrace \tZ_{\pi(t)}> \frac{\mu}{2} \; \text{and} \; S_{\pi(t)}^{(t)} < \frac{\mu}{5},  \ \forall \ t\le t_0\right\rbrace 
& \Rightarrow 
\left\lbrace C_1(t) > \frac{\mu}{2} + t\left( \frac{\mu}{2} -
\frac{2\mu}{5}\right), \ \forall \ t \in[0, t_0] \right\rbrace\nonumber\\
& \Rightarrow  
\left\lbrace C_1(t) >0, \ \forall \ t \in[0, t_0] \right\rbrace.
\label{eq:ZS_bounds}
\end{align}
It follows that 
\begin{align}
\prob{\mathcal{E} } &= \prob{C_1(t)>0, \ \forall \ t \in[0, t_0]} \nonumber\\
& \geq \prob{ \tZ_{\pi(t)} > \frac{\mu}{2}  \; \text{and}\; S_{\pi(t)}^{(t)} < \frac{\mu}{5},  \ \forall \ t \le t_0 \bigmid S_{\pi(t)}^{(t)} \geq 1, \ \forall \ t\in [1, t_0] }, 
\end{align}
where the condition $S_{\pi(t)}^{(t)} \geq 1,\ \forall \ t\in [1, t_0]$ captures the fact that the left node recovered in each iteration $t$ was connected to at least one degree-1 right node at $t$. Applying a union bound gives
\begin{equation}\label{eq:PE_bnd1}
\prob{\mc{E}} \geq 
1-\sum_{t=0}^{t_0}  \prob{\tZ_{\pi(t)} \le  \frac{\mu}{2} }- \sum_{t=1}^{t_0} \prob{S_{\pi(t)}^{(t)} \ge  \frac{\mu}{5} \mid S_{\pi(t)}^{(t)}\geq 1}.
\end{equation}
Each summand in the second term can be bounded using \eqref{eq:tail_bounds_left_degrees_Stage_B} and each summand in the third term can be bounded using \eqref{eq:Slt_LB} for $t \le t_0=\frac{k-1}{20}$. We therefore obtain 
\begin{align}
\prob{\mc{E}}
 & > 1- (t_0+1) \exp\left(-\frac{1}{10} k^{1-\delta}\right) 
 -  t_0 \exp\left(-\frac{1}{16} k^{1-\delta} \right) \nonumber\\
& > 1-\left( \frac{k-1}{20}+1\right)  \exp\left(-\frac{1}{10} k^{1-\delta}\right) 
-\frac{k-1}{20} \exp\left(-\frac{1}{16} k^{1-\delta} \right) \nonumber\\
& > 1-\exp\left(-\frac{1}{20} k^{1-\delta}\right),
\end{align}
where the last inequality holds for sufficiently large $k$. 
\end{proof}

We now prove Lemma \ref{lem:1B_rank_r} by proving \eqref{eq:bound_on_Balti_rank_r} based on Lemmas \ref{lem:num_singletons} and \ref{lem:deg_1_rnodes_upto_t0}. In particular, we show that after $t_0=\frac{k-1}{20}$ iterations,  every remaining left node in the residual subgraph is connected to at least one degree-1 right node with high probability and so can be recovered by the peeling algorithm.
\begin{proof}[Proof of Lemma \ref{lem:1B_rank_r}]
Consider the alternative graph process on subgraph $i \in [r]$.  Conditioned on the event  
\[ 
\mathcal{E} = \left\lbrace C_1(t)>0, \ \forall \ t\in [0, t_0]\right\rbrace, 
\]
 $(t_0+1)$ left nodes (out of $k$)  have been recovered and peeled off by  iteration $t_0$. Consider one of the remaining left nodes after iteration $t_0$, i.e., $v_l$  with $l\in[k]\setminus \left\lbrace \pi(0), \dots, \pi(t_0)\right\rbrace$. We  show that with high probability, $v_l$ is  connected to at least one degree-1 right node  and is  therefore recoverable in the forthcoming steps.  From Lemma \ref{lem:num_singletons},  the number of degree-1 right nodes connected to  $v_l$ after iteration $t_0$, denoted by $S_l^{(t_0+1)}$, is distributed as 
$S_l^{(t_0+1)} \sim \bin{\sfA_i \binom{k}{2} }{(t_0+1)/\binom{k}{2}}$. 
Thus, using $t_0 = \frac{k-1}{20}$ and $\sfA_i \ge \alpha_i^*$, we have the bound
\begin{align}\label{eq:p0_single}
\prob{ S_l^{(t_0+1)} =0 \mid \mathcal{E} } =\left(1- \frac{t_0+1}{\binom{k}{2}}\right)^{ \sfA_i\binom{k}{2} }
\le\exp\left(- (t_0+1) \alphaUB_i\right)
< \exp\left( - \frac{1}{25} k^{1-\delta}\right).
\end{align}
Applying a union bound  gives
\begin{align}\label{eq:prob_ge1_single}
\prob{ S_l^{(t_0+1)} \geq 1, \ \forall \ l\in [k]\setminus \left\lbrace \pi(0), \ldots, \pi(t_0) \right\rbrace \mid \mathcal{E} } 
&> 
1- \left(k-(t_0+1)\right)\exp\left( - \frac{1}{25} k^{1-\delta}\right) \nonumber\\
&> 1- k \exp\left( - \frac{1}{25} k^{1-\delta}\right) .
\end{align}

We recall that $\Balti$ is the fraction of left nodes recovered when  degree-1 right nodes run out in the alternative graph process on subgraph 
$i\in [r]$. Then, using \eqref{eq:hp_bnd_E} and \eqref{eq:prob_ge1_single}, for all  $N_{Di} \ge 1$  and $\sfA_i \ge \alphaUB_i$, we have for sufficiently large $k$:
\begin{align}
 \prob{ \Balti =1 \bigmid N_{Di}, \,  \sfA_i} 
& \geq \prob{\mathcal{E}} \prob{ S_l^{(t_0+1)} \geq 1, \ \forall \ l\in [k]\setminus \left\lbrace \pi(0), \ldots, 
\pi(t_0) \right\rbrace \mid \mathcal{E} } \nonumber\\
& > \left[1-\exp\left(-\frac{1}{20} k^{1-\delta}\right)\right]
\left[1- k\exp\left( - \frac{1}{25} k^{1-\delta}\right)\right] \nonumber\\
& > 1- \exp\left(-\frac{1}{30} k^{1-\delta}\right) .
\end{align}
This proves \eqref{eq:bound_on_Balti_rank_r}, and the result of Lemma \ref{lem:1B_rank_r} follows via  \eqref{eq:alt_orig_bnd}.
\end{proof}

\subsection{Proof of Lemma \ref{lem:non_sym_Stage_A} }\label{sec:proof_Stage_A_main_lemma_nonsym}

At the start of stage A, the pruned bipartite graph for a non-symmetric matrix $\bX$ as defined in Lemma \ref{lem:non_sym_Stage_A} has $r \beta k^2$ left nodes and $\nBins  = d r \beta k^2/(\delta \ln k)$ right nodes. The proof of  \eqref{eq:Ai_LB_nonsymm} is  along the same lines as that of  \eqref{eq:bound_on_all_Ai_rank_r}. Specifically, we first obtain a concentration inequality similar to  \eqref{eq:bound_num_singletons_1A} on the number of singleton right nodes in the initial graph. The subsequent steps are very similar to \eqref{eq:bound_num_singletons_Stage_A}--\eqref{eq:bound_num_lnodes_recovered_Stage_A} and \eqref{eq:bound_on_num_above_diag_recovered_in_each_block}--\eqref{eq:bound_on_Ai_and_NDi} and are omitted for brevity. 

\subsection{Proof of Lemma \ref{lem:non_sym_Stage_B} }\label{sec:proof_Stage_B_main_lemma_nonsym}

Recall that the vectors in each of the sets $\{\bu_1, \ldots, \bu_r\}$ and $\{\bv_1, \ldots, \bv_r\}$ have $k$ and $\beta k$ disjoint supports, respectively, for some constant $\beta\in (0,1]$. Thus, the nonzero entries of $\bX = \sum_{i=1}^r \sigma_i \bu_i \bv_i^{T} = \sum_{i=1}^r \tbu_i \tbv_i^{T} $ form $r$ disjoint submatrices, each of size $k \times \beta k$. 
The initial graph for stage B  consists of $r$ disjoint bipartite  subgraphs. The peeling algorithm in stage B  is applied to each of these subgraphs. 
In the $i$-th  subgraph, the  $(k + \beta k)$ left nodes represent the unknown nonzeros in $\tbu_i$ and $\tbv_i$. The $\sfA_i \beta k^2$ right nodes represent the nonzero pairwise products recovered in the $i$-th nonzero submatrix in stage A. 

As in the proof of the symmetric case, we will consider an alternative initial graph  in which for each subgraph $i$, the $\sfA_i \beta k^2$ right nodes  are drawn uniformly at random \emph{with replacement} from the set of $\beta k^2$ nonzero pairwise products in the $i$-th nonzero submatrix.  See Fig. \ref{subfig:alt_graph_1B_non_sym_t=-1} for an illustration. The number of distinct right nodes  for the alternative initial graph can be no larger than that for the original one.  Moreover, as illustrated in Fig.  \ref{fig:alt_graph_1B_non_sym}, following the recovery of the first nonzero entry of $\tbu_i$, the  process on the alternative graph peels off all the recoverable entries of $\tbv_i$ before peeling any nonzero entries of $\tbu_i$ that have become recoverable along the way. The order of peeling does not affect the total number of nonzeros recoverable in $\tbu_i$ and $\tbv_i$ before degree-1 right nodes run out.

Writing $\usfB_{\text{alt},i}$ and $\vsfB_{\text{alt},i}$ for the fractions of nonzeros  recoverable in $\tbu_i$ and $\tbv_i$ in the alternative process, we therefore have 
\beq\label{eq:failure_prob_original_vs_alt_model}
\prob{ \{\usfB_{\text{alt},i}<1\} \,  \cup \, \{ \vsfB_{\text{alt},i}<1 \} \mid \sfA_i }
\ge \prob{ \{ \usfB_i<1 \} \, \cup \, \{ \vsfB_i<1 \} \mid \sfA_i  }.
\eeq
To prove \eqref{eq:non_sym_Stage_B_main_result}, we will show that for $\sfA_i \ge \alpha_i^*$ and sufficiently large $k$:
\beq\label{eq:failure_prob_block_i_Stage_B_nonsym}
\prob{ \{ \usfB_{\text{alt},i} <1 \} \, \cup \, \{ \vsfB_{\text{alt},i} <1 \} \mid\sfA_i }  \le \exp\left(-\frac{\beta}{8} k^{1-2\delta}\right),\quad i\in[r].
\eeq

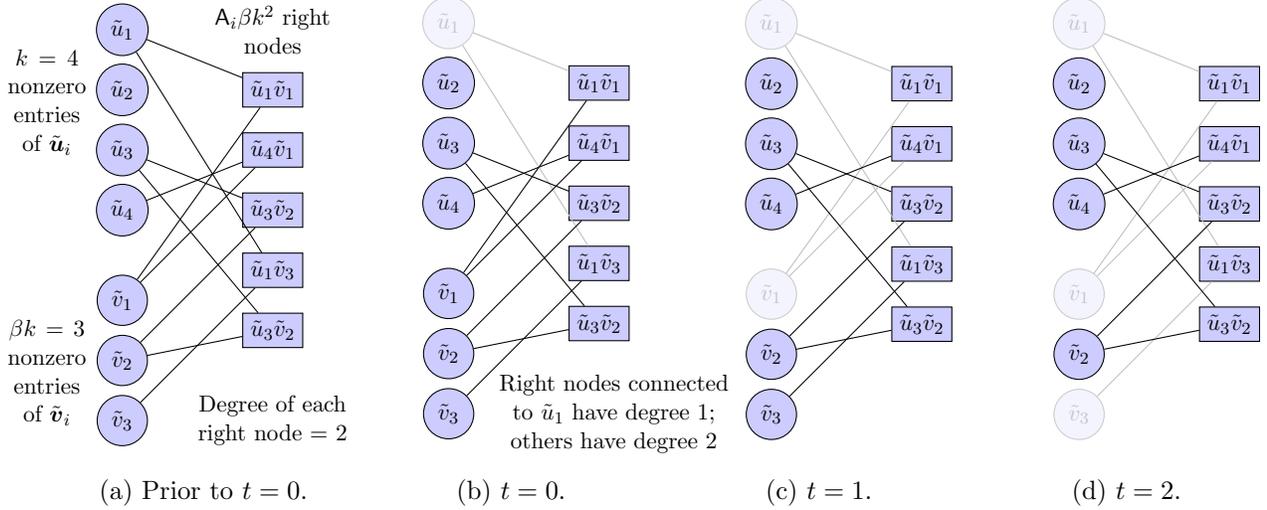
\begin{figure}
\captionsetup[subfigure]{justification=centering}
\centering
\linespread{1}
 \begin{subfigure}[B]{0.24\linewidth} 
 \centering 
 \begin{adjustwidth}{-1cm}{0cm}
\begin{tikzpicture}
\node[vnode](vnodeu1){$\tu_{1}$};
\node[vnode, below of=vnodeu1, node distance=1cm](vnodeu2){$\tu_{2}$};
\node[vnode, below of=vnodeu1, node distance=2cm](vnodeu3){$\tu_3$};
\node[vnode, below of=vnodeu1, node distance=3cm](vnodeu4){$\tu_4$};

\node[vnode, below of =vnodeu4, node distance=1.5cm](vnodev1){$\tv_1$};
\node[vnode, below of=vnodev1, node distance=1cm](vnodev2){$\tv_2$};
\node[vnode, below of=vnodev1, node distance=2cm](vnodev3){$\tv_3$};

\node[cnode, right of= vnodeu1, yshift=-1cm, node distance=2.5cm](cnodeu1v1){$\tu_1\tv_1$};
\node[cnode, below of= cnodeu1v1, node distance=1cm](cnodeu4v1){$\tu_4 \tv_1$};
\node[cnode, below of= cnodeu1v1, node distance=2cm](cnodeu3v2){$\tu_3 \tv_2$};
\node[cnode, below of= cnodeu1v1, node distance=3cm](cnodeu1v3){$\tu_1\tv_3$};
\node[cnode, below of= cnodeu1v1, node distance=4cm](cnodeu3v2-1){$\tu_3\tv_2$};

\draw (vnodeu1)--(cnodeu1v1);
\draw (vnodev1)--(cnodeu1v1);
\draw (vnodeu1)--(cnodeu1v3);
\draw (vnodev3)--(cnodeu1v3);
\draw  (vnodeu3)--(cnodeu3v2);
\draw (vnodev2)--(cnodeu3v2);
\draw  (vnodeu4)--(cnodeu4v1);
\draw (vnodev1)--(cnodeu4v1);
\draw  (vnodeu3)--(cnodeu3v2-1);
\draw (vnodev2)--(cnodeu3v2-1);

\node[left of=vnodeu2, node distance=1cm, text width=2cm, align=center, scale=0.8, yshift=-0.2cm](numLnodeLabel){$k =4$ nonzero entries of $\tbu_i$};
\node[left of=vnodev2, node distance=1cm, text width=2cm, align=center, scale=0.8, yshift=-0.2cm](numLnodeLabel){$\beta k=3$ nonzero entries of $\tbv_i$};
\node[above of=cnodeu1v1, node distance=0.8cm, xshift =0cm, text width=3cm, align=center, scale=0.8](numRnodeLabel){$\sfA_i \beta k^2$ right nodes}; 
\node[below of=cnodeu1v3, node distance=2cm, text width=3cm, align=center, scale=0.8](rightDist1){Degree of each right node = 2};
\end{tikzpicture}
\end{adjustwidth}
\caption{Prior to $t=0$.}
\label{subfig:alt_graph_1B_non_sym_t=-1}
\end{subfigure}
\begin{subfigure}[B]{0.24\linewidth} 
\centering 
\begin{adjustwidth}{0.8cm}{0cm}
\begin{tikzpicture}
\node[vnodeFaded](vnodeu1){$\tu_1$};
\node[vnode, below of=vnodeu1, node distance=1cm](vnodeu2){$\tu_2$};
\node[vnode, below of=vnodeu1, node distance=2cm](vnodeu3){$\tu_3$};
\node[vnode, below of=vnodeu1, node distance=3cm](vnodeu4){$\tu_4$};

\node[vnode, below of =vnodeu4, node distance=1.5cm](vnodev1){$\tv_1$};
\node[vnode, below of=vnodev1, node distance=1cm](vnodev2){$\tv_2$};
\node[vnode, below of=vnodev1, node distance=2cm](vnodev3){$\tv_3$};

\node[cnode, right of= vnodeu1, yshift=-1cm, node distance=2.5cm](cnodeu1v1){$\tu_1 \tv_1$};
\node[cnode, below of= cnodeu1v1, node distance=1cm](cnodeu4v1){$\tu_4 \tv_1$};
\node[cnode, below of= cnodeu1v1, node distance=2cm](cnodeu3v2){$\tu_3 \tv_2$};
\node[cnode, below of= cnodeu1v1, node distance=3cm](cnodeu1v3){$\tu_1 \tv_3$};
\node[cnode, below of= cnodeu1v1, node distance=4cm](cnodeu3v2-1){$\tu_3 \tv_2$};

\draw [gray!50] (vnodeu1)--(cnodeu1v1);
\draw (vnodev1)--(cnodeu1v1);
\draw [gray!50] (vnodeu1)--(cnodeu1v3);
\draw (vnodev3)--(cnodeu1v3);
\draw  (vnodeu3)--(cnodeu3v2);
\draw (vnodev2)--(cnodeu3v2);
\draw  (vnodeu4)--(cnodeu4v1);
\draw (vnodev1)--(cnodeu4v1);
\draw  (vnodeu3)--(cnodeu3v2-1);
\draw (vnodev2)--(cnodeu3v2-1);

\node[above of=cnodeu1v1, node distance=0.8cm, xshift =0.2cm, text width=4cm, align=center, scale=0.8, text=white](numRnodeLabel){};
\node[below of=cnodeu1v3, node distance=2cm, xshift =0.2cm, text width=3.9cm, align=center, scale=0.8](rightDist1){Right nodes connected to $\tu_1$ have degree 1; others have degree 2};
\end{tikzpicture}
\end{adjustwidth}
\vspace{-0.1cm}
\caption{$t=0$.}
\label{subfig:alt_graph_1B_non_sym_t=0}
\end{subfigure}
\begin{subfigure}[B]{0.24\linewidth} 
\centering 
\begin{adjustwidth}{1cm}{0cm}
\begin{tikzpicture}
\node[vnodeFaded](vnodeu1){$\tu_1$};
\node[vnode, below of=vnodeu1, node distance=1cm](vnodeu2){$\tu_2$};
\node[vnode, below of=vnodeu1, node distance=2cm](vnodeu3){$\tu_3$};
\node[vnode, below of=vnodeu1, node distance=3cm](vnodeu4){$\tu_4$};

\node[vnodeFaded, below of =vnodeu4, node distance=1.5cm](vnodev1){$\tv_1$};
\node[vnode, below of=vnodev1, node distance=1cm](vnodev2){$\tv_2$};
\node[vnode, below of=vnodev1, node distance=2cm](vnodev3){$\tv_3$};

\node[cnode, right of= vnodeu1, yshift=-1cm, node distance=2.5cm](cnodeu1v1){$\tu_1 \tv_1$};
\node[cnode, below of= cnodeu1v1, node distance=1cm](cnodeu4v1){$\tu_4 \tv_1$};
\node[cnode, below of= cnodeu1v1, node distance=2cm](cnodeu3v2){$\tu_3 \tv_2$};
\node[cnode, below of= cnodeu1v1, node distance=3cm](cnodeu1v3){$\tu_1 \tv_3$};
\node[cnode, below of= cnodeu1v1, node distance=4cm](cnodeu3v2-1){$\tu_3 \tv_2$};

\draw [gray!50] (vnodeu1)--(cnodeu1v1);
\draw [gray!50] (vnodev1)--(cnodeu1v1);
\draw [gray!50] (vnodeu1)--(cnodeu1v3);
\draw (vnodev3)--(cnodeu1v3);
\draw (vnodeu3)--(cnodeu3v2);
\draw (vnodev2)--(cnodeu3v2);
\draw (vnodeu4)--(cnodeu4v1);
\draw [gray!50](vnodev1)--(cnodeu4v1);
\draw (vnodeu3)--(cnodeu3v2-1);
\draw (vnodev2)--(cnodeu3v2-1);
\node[above of=cnodeu1v1, node distance=0.8cm, xshift =0.2cm, text width=4cm, align=center, scale=0.8, text=white](numRnodeLabel){};
\node[below of=cnodeu1v3, node distance=1.8cm, text width=3.5cm, align=center, scale=0.8, text=white](rightDist1){Right nodes connected to $\tu_1$ have degree 1; others have degree 2};
\end{tikzpicture}
\end{adjustwidth}
\vspace{-0.1cm}
\caption{$t=1$.}
\label{subfig:alt_graph_1B_non_sym_t=1}
\end{subfigure}
\begin{subfigure}[B]{0.24\linewidth} 
\centering 
\begin{adjustwidth}{1cm}{0cm}
\begin{tikzpicture}
\node[vnodeFaded](vnodeu1){$\tu_1$};
\node[vnode, below of=vnodeu1, node distance=1cm](vnodeu2){$\tu_2$};
\node[vnode, below of=vnodeu1, node distance=2cm](vnodeu3){$\tu_3$};
\node[vnode, below of=vnodeu1, node distance=3cm](vnodeu4){$\tu_4$};

\node[vnodeFaded, below of =vnodeu4, node distance=1.5cm](vnodev1){$\tv_1$};
\node[vnode, below of=vnodev1, node distance=1cm](vnodev2){$\tv_2$};
\node[vnodeFaded, below of=vnodev1, node distance=2cm](vnodev3){$\tv_3$};

\node[cnode, right of= vnodeu1, yshift=-1cm, node distance=2.5cm](cnodeu1v1){$\tu_1 \tv_1$};
\node[cnode, below of= cnodeu1v1, node distance=1cm](cnodeu4v1){$\tu_4 \tv_1$};
\node[cnode, below of= cnodeu1v1, node distance=2cm](cnodeu3v2){$\tu_3 \tv_2$};
\node[cnode, below of= cnodeu1v1, node distance=3cm](cnodeu1v3){$\tu_1 \tv_3$};
\node[cnode, below of= cnodeu1v1, node distance=4cm](cnodeu3v2-1){$\tu_3 \tv_2$};

\draw [gray!50] (vnodeu1)--(cnodeu1v1);
\draw [gray!50] (vnodev1)--(cnodeu1v1);
\draw [gray!50] (vnodeu1)--(cnodeu1v3);
\draw [gray!50] (vnodev3)--(cnodeu1v3);
\draw (vnodeu3)--(cnodeu3v2);
\draw (vnodev2)--(cnodeu3v2);
\draw (vnodeu4)--(cnodeu4v1);
\draw [gray!50](vnodev1)--(cnodeu4v1);
\draw (vnodeu3)--(cnodeu3v2-1);
\draw (vnodev2)--(cnodeu3v2-1);
\node[above of=cnodeu1v1, node distance=0.8cm, xshift =0.2cm, text width=4cm, align=center, scale=0.8, text=white](numRnodeLabel){};
\node[below of=cnodeu1v3, node distance=1.8cm, text width=3.5cm, align=center, scale=0.8, text=white](rightDist1){Right nodes connected to $\tu_1$ have degree 1; others have degree 2};
\end{tikzpicture}
\end{adjustwidth}
\vspace{-0.1cm}
\caption{$t=2$.}
\label{subfig:alt_graph_1B_non_sym_t=2}
\end{subfigure}
\caption{\small Alternative graph process showing the recovery of $\tbu_i$ and $\tbv_i$  given $\sfA_i $ (as the counterpart of Fig. \ref{fig:graph_process_1B_non_sym}). The double occurrence of $\tu_3 \tv_2$ as right nodes highlights that 
the right nodes are sampled  \emph{with replacement} from 
the $\beta k^2$ nonzero pairwise products in the nonzero submatrix corresponding to $\tbu_i \tbv_i^T$. 
(b): At $t=0$, $\tu_1$ is recovered as 1 and  peeled off. (c)--(d): Following the recovery of $\tu_1$, we peel off all the recoverable $\tv_j$'s (i.e., $\tv_1, \tv_3$)  before peeling any  $\tu_l$'s that have become recoverable along the way (i.e., $\tu_4$).}
\label{fig:alt_graph_1B_non_sym}
\end{figure}

The following lemma specifies the degree distribution of each left node at the beginning of stage B. The proof is along the same lines as that of Lemma \ref{lem:initial_left_degree_bound_1B} and is omitted.
\begin{lem}[Initial left degrees are Binomial]
\label{lem:initial_left_degrees_Stage_B_nonsym}
Let $\utZ_l$ be the degree of the  left node representing the $l$-th nonzero of $\tbu_i$, and $\vtZ_j$ be the degree of the left node representing the $j$-th nonzero of $\tbv_i$,  for  $l\in [k]$ and  $j\in [\beta k]$. Then,  
\begin{align}\label{eq:left_degrees_are_binomial_Stage_B_nonsym}
\utZ_l \sim \bin{\sfA_i \beta k^2}{\frac{1}{k}}, \quad 
\vtZ_j \sim \bin{\sfA_i \beta k^2}{\frac{1}{\beta k}}.
\end{align}
Then, for  $\sfA_i\ge \alphaUB_i$ and sufficiently large $k$:
\begin{align}\label{eq:bounds_on_left_degrees_Stage_B_nonsym}
 &\prob{ \utZ_l \notin \left[\frac{\sfA_i \beta k}{2}, \, \frac{3\sfA_i \beta k}{2}\right] \bigmid \sfA_i}
< 2\exp\left( -\frac{ \beta}{10} k^{1-\delta} \right) ,\quad \\
& \prob{ \vtZ_j \notin \left[\frac{\sfA_i k}{2}, \,  \frac{3 \sfA_i k}{2}\right] \bigmid \sfA_i}
< 2\exp\left( -\frac{1 }{10} k^{1-\delta} \right).
\end{align}
Define the event 
\beq
\mc{T}_0: = \left\lbrace \utZ_l \in \left[\frac{\sfA_i \beta k}{2}, \, \frac{3\sfA_i \beta k}{2}\right],  \ \forall\, l\in[k] \right\rbrace
\cap \left\lbrace\vtZ_j \in \left[\frac{\sfA_i k}{2}, \,  \frac{3 \sfA_i k}{2}\right], \ \forall \, j\in[\beta k]\right\rbrace.
\label{eq:T0_def}
\eeq
 Then for  $\sfA_i\ge \alphaUB_i$ and sufficiently large $k$:
\beq\label{eq:bound_left_degrees_B_nonsym}
\prob{\mc{T}_0 \mid \sfA_i} >1- 2k \exp\left( -\frac{ \beta }{10} k^{1-\delta} \right) 
- 2\beta k \exp\left( -\frac{1 }{10} k^{1-\delta} \right) .
\eeq
\end{lem}

The recovery of each nonzero entry of $\tbu_i$ or $\tbv_i$ is counted as one iteration of the  algorithm. The recovery of the first nonzero entry of $\tbu_i$ or $\tbv_i$ 
corresponds to iteration $t=0$. The next lemma specifies the distribution of degree-1 right nodes connected to each remaining left node after each iteration of the algorithm. The proof is similar to that of Lemma \ref{lem:num_singletons} and is omitted for brevity.

\begin{lem}[Number of degree-1 right nodes connected to each unpeeled left node]
\label{lem:num_singleton_edges_Stage_B_nonsym}
Suppose that the peeling algorithm on the $i$-th subgraph  has not terminated after iteration $(t-1)$, for $1\le t\le( k+\beta k-1)$. 
 Let $\Ut{t-1}$ and $\Vt{t-1}$ be the set of indices of the nonzeros recovered in $\tbu_i$ and $\tbv_i$ by the end of iteration $(t-1)$. 
At the start of iteration $t$, let $\uSt{t}_l$ be the number of degree-1 right nodes connected to the remaining left node $\tu_l$ and let $\vSt{t}_j$ be that connected to the remaining left node $\tv_j$, where   $l\in [k]\setminus \Ut{t-1}$ and  $j\in [\beta k]\setminus \Vt{t-1}$. Then,
\begin{equation}\label{eq:dist_num_singleton_edges_Stage_B_nonsym}
\uSt{t}_l \sim \bin{\sfA_i \beta k^2}{\frac{\abs{\Vt{t-1}}}{\beta k^2 }}, \quad 
\vSt{t}_j \sim \bin{\sfA_i \beta k^2}{\frac{\abs{\Ut{t-1}}}{\beta k^2 }}.
\end{equation}
\end{lem}

 In the following, we refer to the left nodes representing the nonzero entries of $\tbu_i$ as $\tu$-left nodes, and those representing  nonzero entries of $\tbv_i$ as $\tv$-left nodes. 
We assume without loss of generality that  the peeling algorithm is initialized by setting  a $\tu$-left node to 1. As  illustrated in Fig.  \ref{subfig:alt_graph_1B_non_sym_t=0}, following the recovery of the first $\tu$-left node, all the right nodes connected to this node reduce to degree-1. 
 The next lemma gives a high probability lower bound on the number of degree-1 right nodes created. 


\begin{lem}[Recovery of the first left node]
\label{lem:bound_num_nonempty_bins_B_nonsym}  
Denote by $N$ the number of distinct $\tv$-left nodes connected to the degree-1 right nodes created by the recovery of the first $\tu$-left node.  
Then, there exists $b_{k,\delta} = o(k^{1-\delta})$ such that for  $\sfA_i\ge \alphaUB_i$ and sufficiently large $k$, we have
\beq\label{eq:bound_num_nonempty_bins_B_nonsym}
\prob{N \ge\frac{\beta}{2} k^{1-\delta} - b_{k, \delta} \bigmid \mc{T}_0, \sfA_i}
\ge 1-2\exp\left(- \frac{\beta }{2} k^{1-\frac{3}{2}\delta} \right) .
\eeq
Here $\mc{T}_0$ is the event defined in \eqref{eq:T0_def}.
\end{lem}
\begin{proof}
Let $\utZ$ be the degree of the  first $\tu$-left node recovered. Then  the recovery of this $\tu$-left node  creates $\utZ$ degree-1 right nodes.  Consider a $\tv$-left node, say $\tv_j$, and let $\vtS_j$  be the number of degree-1 right nodes connected to $\tv_j$ after the first $\tu$-left node is peeled off. Then, conditioned  on $\utZ$, we have
\beq \label{eq:num_balls_in_one_nonempty_bin}
\vtS_j := \sum_{s=1}^{\utZ}\mathbbm{1}{\{ \text{The $s$-th degree-1 right node is connected to } \tv_j \}} \sim\bin{\utZ }{\frac{1}{\beta k}}, \quad \text{for }j\in[\beta k],
\eeq
and  we can express $N$ (defined in the lemma statement) as
\beq\label{eq:def_N_num_nonempty_bins}
N := \sum_{j=1}^{\beta k} \mathbbm{1}{\{\vtS_j\ge 1\}}.
\eeq

 Let $V_s$ be the index of the $\tv$-left node that the $s$-th degree-1 right node is connected to, for $s \in[ \utZ]$. Let $N_0 =\E[N \mid \utZ, \sfA_i]$ and 
\beq
N_s = \E[  N \mid V_1, \dots, V_s, \,  \utZ, \sfA_i ], \qquad   s \in[\utZ].
\eeq 
Conditioned on $\utZ$, the sequence $\{ N_s \}$ is a Doob Martingale, with $N_{\utZ} =N$. Writing $N= f(V_1, \dots, V_{\utZ})$, we note that $f$ is 1-Lipschitz since changing the connection of a single right node cannot change $N$ by more than one. Moreover, by the construction of the alternative graph, $V_1, \ldots, V_{\utZ}$ are independent given  ($\utZ, \sfA_i$). Therefore, by McDiarmid's inequality 
\cite[Sec. 13.5]{mitzenmacher2017probabilityAndComputing}, we have 
\begin{equation}
\label{eq:concen_bound_on_num_nonempty_bins}
\prob{\abs{N-\E[N \mid \utZ, \sfA_i]} \ge k^{-\frac{\delta}{4}} \utZ \, \bigmid \,  \utZ, \, \sfA_i} 
\le 2\exp\left(- 2k^{-\frac{\delta}{2}} \utZ\right).
\end{equation}
Hence
\begin{equation}
    \prob{ N \, \ge \, \E[N \mid \utZ, \sfA_i] - k^{-\frac{\delta}{4}}\utZ
    \, \bigmid \,  \utZ, \, \sfA_i} \ge 1- 2\exp\left(- 2k^{-\frac{\delta}{2}} \utZ\right).
    \label{eq:tail_prob_N0}
\end{equation}

From \eqref{eq:num_balls_in_one_nonempty_bin} and \eqref{eq:def_N_num_nonempty_bins}, we obtain that
\beq\label{eq:expected_num_nonempty_bins}
\E[N\mid \utZ, \sfA_i] = \beta k \,  \prob{ \vtS_j\ge 1 }= \beta k \left[1 -\left(1-\frac{1}{\beta k}\right)^{\utZ} \right] \ge 
\beta k \left( 1 - \exp\Big(-\frac{\utZ}{\beta k} \Big) \right).
\eeq 
Furthermore, given $\sfA_i \ge \alpha_i^* = k^{-\delta}  - o(k^{-\delta})$ and conditioning on $\mc{T}_0$, we have
\begin{equation}
\utZ \ge \frac{\beta}{2} k^{1-\delta} - o(k^{1-\delta}), \qquad \E[N\mid \utZ, \sfA_i] - k^{-\frac{\delta}{4}} \utZ \ge \frac{\beta}{2} k^{1-\delta} - o(k^{1-\delta}).
\label{eq:cond_exp_N}
\end{equation}
Using \eqref{eq:cond_exp_N} in \eqref{eq:tail_prob_N0} gives the result in \eqref{eq:bound_num_nonempty_bins_B_nonsym}.
\end{proof}

\begin{proof}[Proof of Lemma \ref{lem:non_sym_Stage_B}]
With $N$ as defined in Lemma 
\ref{lem:bound_num_nonempty_bins_B_nonsym}, at the end of iteration $N$ of the peeling algorithm, one $\tu$-left node and $N$ $\tv$-left nodes have been recovered. 
Adopting the notation in Lemma \ref{lem:num_singleton_edges_Stage_B_nonsym}, this means $\abs{ \Ut{N}} =1$ and $ \abs{ \Vt{N}} =N$. Lemma \ref{lem:num_singleton_edges_Stage_B_nonsym}  implies that at the start of iteration $(N+1)$, the number of degree-1 right nodes connected to each remaining $\tu$-left node  is
\beq\label{eq:num_singleton_edges_round3_B_nonsym}
\uSt{N+1}_l
\sim \bin{\sfA_i \beta k^2}{\frac{ N}{\beta k^2 }}, \quad\text{for }l\in [k]\setminus \Ut{N}.
\eeq
Recall the definition of $\mc{T}_0$ from \eqref{eq:T0_def} and let \begin{equation}
    \mc{T}_1:= \left\lbrace N \ge\frac{\beta}{2} k^{1-\delta} - b_{k,\delta}\right\rbrace, \qquad \mc{T}_2 :=  \left\{ \uSt{N+1}_l \ge 1, \ \forall \, l\in [k]\setminus \Ut{N} \right\}.
    \label{eq:T1_T2_def}
\end{equation}
Here $b_{k,\delta} = o(k^{1-\delta})$ is the same as in Lemma \ref{lem:bound_num_nonempty_bins_B_nonsym}. Note that $\mc{T}_2$ is the event that at the start of iteration $(N+1)$, each of the $(k-1)$ remaining $\tu$-left nodes is connected to at least one degree-1 right node and can therefore be recovered.  Analogously to \eqref{eq:p0_single} and \eqref{eq:prob_ge1_single}, we can show that for $\sfA_i\ge \alphaUB_i$ and sufficiently large $k$, 
\beq\label{eq:prob_bound_second_last_round_B_nonsym}
\prob{ \mc{T}_2 \mid \mc{T}_0, \,  \mc{T}_1,  \,  \sfA_i } >1- (k-1)\exp\left( -\frac{\beta}{4} k^{1-2\delta}\right).
\eeq
Since the recovery of each left node corresponds
to one iteration, conditioned on the event $\mc{T}_2$, we have $\abs{\Ut{N+k-1}} =k$ and $\abs{\Vt{N+k-1}} =N$. 

At the start of iteration $(N+k)$, all the  left nodes that remain are $\tv$-left nodes, so all the right nodes that remain are degree-1.  Moreover, for $\sfA_i \ge \alphaUB_i$ and sufficiently large $k$,  conditioning on event $\mc{T}_0$ ensures that each $\tv$-left node is connected to at least one right node. Therefore, all the remaining $\tv$-left nodes can be recovered. If the events $\mc{T}_0,  \mc{T}_1, \mc{T}_2$ all hold, then the algorithm successfully recovers all the nonzeros of  $\tbu_i$ and $\tbv_i$. 
Indeed, we have for $\sfA_i\ge \alphaUB_i$ and sufficiently large $k$:
\begin{align}\label{eq:success_prob_ul_is_first_recovered}
&\prob{ \{ \usfB_{\text{alt},i} =1 \} \, \cap\, \{ \vsfB_{\text{alt},i}  =1 \} \, \mid \, \sfA_i }  \nonumber\\
& \ge \prob{\mc{T}_0, \mc{T}_1, \mc{T}_2 \mid \sfA_i } \nonumber\\
& = \prob{\mc{T}_0 \mid \sfA_i } \prob{\mc{T}_1 \mid \mc{T}_0, \sfA_i }
\prob{\mc{T}_2 \mid \mc{T}_0, \mc{T}_1, \sfA_i} \nonumber\\
& \stackrel{\text{(\romannumeral 1)}}{>} 
\left[1- 2k \exp\left( -\frac{\beta}{10} k^{1-\delta} \right) 
- 2\beta k \exp\left( -\frac{1}{10} k^{1-\delta} \right) \right] 
\left[1-2\exp\left(- \frac{\beta}{2} k^{1-\frac{3}{2}\delta} \right) \right]  \nonumber  \\
 & \hspace{10pt} \cdot 
\left[1- (k-1)\exp\left( -\frac{\beta }{4} k^{1-2\delta}\right)\right]
 \nonumber\\
& >1- \exp\left(-\frac{\beta}{8} k^{1-2\delta}\right),
\end{align}
where the inequality $\rm{(i)}$ is obtained using the bounds in \eqref{eq:bound_left_degrees_B_nonsym}, \eqref{eq:bound_num_nonempty_bins_B_nonsym} and  \eqref{eq:prob_bound_second_last_round_B_nonsym}. This proves \eqref{eq:failure_prob_block_i_Stage_B_nonsym}, and via \eqref{eq:failure_prob_original_vs_alt_model}, completes the proof of Lemma \ref{lem:non_sym_Stage_B}.

\end{proof}


\section{Conclusion}
The main contribution of this paper is a sketching scheme for sparse, low-rank matrices and an algorithm for recovering the singular vectors of such a matrix with a sample complexity and running time that both depend only on the sparsity level $k$ and not on the ambient dimension $n$ of the matrix. 

A key open question in the noiseless setting is how to extend the two-stage recovery algorithm to  matrices where the singular vectors have overlapping supports. A starting point in this direction would be to consider matrices whose singular vectors have a small fraction of nonzero entries in overlapping  locations. This would ensure that a large fraction of the nonzero matrix entries recovered in stage A are still simple pairwise products. 
\edit{Moreover, we would like to improve the guarantees in part 2) of Theorems \ref{thm:main_result_symm} and \ref{thm:nonsym_main_results}, via a proof similar to that of part 1), using properties of negatively associated random variables. This would also give tighter non-asymptotic bounds for the compressed sensing scheme in \cite{li2019sublinear}. }

There are several open questions in the noisy setting where the matrix is only approximately sparse and low-rank. An important one is to improve the running time  of the recovery algorithm, which is currently $\bigo(\max\{n^2\log(n/k), (rk)^3\})$.  
This will require the sketching operator to be defined via a new bin detection matrix $\bS$ which enables  zeroton and singleton bins to be identified more efficiently. In our construction, for simplicity, we chose $\bS$ to be a random Gaussian matrix. For compressed sensing, \cite{li2019sublinear} proposed an $\bS$  based on LDPC codes that allows for faster classification of the bins via a message passing algorithm. It would be interesting to explore  similar ideas for  matrix sketching. Another future direction is to derive performance guarantees for the noisy setting, similar to the nonasymptotic bounds derived in the noiseless case. It is possible to obtain guarantees for the first stage of the algorithm, similarly to  \cite{li2019sublinear}, but the challenge lies in quantifying how the noisy recovery in the first stage affects the recovery accuracy in the second stage.

\section*{Acknowledgements}
The authors thank Prof. Kannan Ramchandran, \edit{Dr. Dong Yin and Dr. Orhan Ocal} for several helpful discussions about the compressed sensing scheme in \cite{li2019sublinear}, \edit{Dr. Samet Oymak for discussions regarding \cite{oymak2015simultaneously}, and Prof. Justin Romberg and Dr. Sohail Bahmani for sharing their source code for \cite{bahmani2016nearopt} and discussions regarding implementation details. The authors also thank Dr. Mark L. Stone and Dr. Stephen Becker for answering questions about CVX and TFOCS.}


\bibliographystyle{IEEEtran}
{\small \bibliography{sparse_isit22_bib} }
\end{document}